\documentclass[a4paper, UKenglish, autoref, thm-restate, cleveref, dvipsnames]{lipics-v2021}

\usepackage{wasysym}

\hideLIPIcs{}  

\title{Strict Self-Assembly of Discrete Self-Similar Fractal Shapes}

\titlerunning{Strict Self-Assembly of Fractals}

\author{Florent Becker}{LIFO --- Université d'Orléans, rue Léonard de Vinci\\ 45000 Orléans, France }{florent.becker@univ-orleans.fr}{https://orcid.org/0000-0002-6742-9042}{}

\authorrunning{Florent Becker} 

\Copyright{Florent Becker} 



\nolinenumbers 

\usepackage{readarray}
\usepackage{tikz}
\usetikzlibrary{tikzmark}
\usetikzlibrary{automata}
\usetikzlibrary{matrix}
\usetikzlibrary{fit}
\usetikzlibrary{shapes}
\usetikzlibrary{angles}
\usetikzlibrary{math}
\usetikzlibrary{external}
\usepackage{mathtools}
\usepackage{tcolorbox}
\usepackage[ruled, vlined]{algorithm2e}

\newcommand\ezEmbed[1][]{\ensuremath{\operatorname{squash}_{#1}}}
\newcommand\pd{\ensuremath{\mathring{-}}}

\newcommand\movie{\ensuremath{\mathcal{MOVIE}}}
\newcommand\far{\ensuremath{\operatorname{Far}}}
\newcommand\near{\ensuremath{\operatorname{Near}}}
\newcommand\tree{\ensuremath{\mathcal{TREE-DEC}}}
\newcommand\Assemblies[1]{\ensuremath{{#1}^{\subset \mathbb{Z}^2}}}

\newcommand\FreeAssemblies[1]{\ensuremath{{#1}^{\operatorname{Free}}}}
\newcommand\dom{\operatorname{dom}}
\newcommand\Ss{\ensuremath{\mathcal{S}}}
\newcommand\Seqs[1]{\ensuremath{\mathcal{H}[#1]}}
\newcommand\FreeSeqs[1]{\ensuremath{\mathcal{H}^{\operatorname{Free}}[#1]}}
\newcommand\Prods[1]{\ensuremath{\mathcal{A}[#1]}}
\newcommand\TerminalProds[1]{\ensuremath{\mathcal{A}_\square[#1]}}
\newcommand\fizziness{\operatorname{fz}}
\newcommand\moreFizzy{\ensuremath{>_{\fizziness}}}
\newcommand\moreFizzyEq{\ensuremath{\geq_{\fizziness}}}

\newcommand{\surligne}[3][]{%
  \tikz[remember picture, anchor=base, baseline]{\node[fill=#2, inner sep=1pt, rounded corners=2pt] (#1) {#3};}%
}

\newcommand\Z{\ensuremath{\mathbb{Z}}}

\newcommand\FBattribute[2][]{
  \def\tempa{}%
  \def\tempb{#1}%
  \ifx\tempa\tempb{}
  \ensuremath{\operatorname{#2}}
  \else
  \ensuremath{#1{\cdot}\operatorname{#2}}
  \fi
}

\newcommand\glue[1][]{\FBattribute[#1]{glue}}
\newcommand\present[1][]{\FBattribute[#1]{present}}
\newcommand\paths[1][]{\FBattribute[#1]{successor}}

\newcommand\strengthFun[1][]{\FBattribute[#1]{strength}}
\newcommand\temperature[1][]{\FBattribute[#1]{temperature}}
\newcommand\tileset[1][]{\FBattribute[#1]{tileset}}
\newcommand\seed[1][]{\FBattribute[#1]{seed}}
\newcommand\tiles[1][]{\FBattribute[#1]{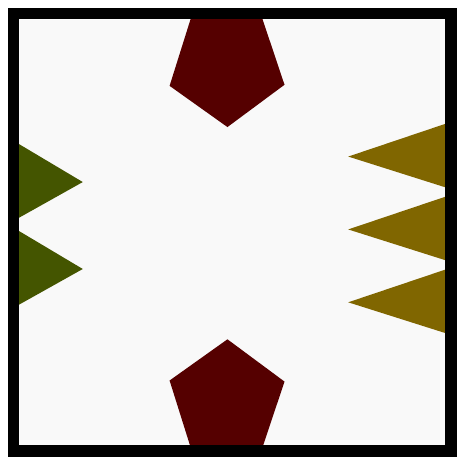}}

\newcommand\gcolor[1]{\ensuremath{\mathfrak{#1}}}

\newcommand\Dir{\ensuremath{\{N, E, S, W\}}}
\newcommand\Wires{\operatorname{Wires}}
\newcommand\Wirings{\mathcal{W}}
\newcommand\PDir{\ensuremath{\mathcal{P}(\Dir)}}

\newcommand\cacarpet{\ensuremath{K_{\infty}}}

\newcommand\bigquot[2]{\ensuremath{\left\lfloor}\frac{#1}{#2}\right\rfloor}
\newcommand\quot[2]{\ensuremath{\lfloor{#1}/{#2}\rfloor}}

\newcommand\Gates[1]{\ensuremath{\operatorname{Gates}(#1)}}
\newcommand\evalfunc[1]{\ensuremath{\bar{#1}}}
\newcommand\extractfunc{\ensuremath{\operatorname{extract}}}
\newcommand\embedfunc{\ensuremath{\operatorname{embed}}}
\newcommand\innerevalfunc[1]{\ensuremath{\widetilde{#1}}}
\newcommand\pinputs{\ensuremath{\operatorname{pred-inputs}}}
\newcommand\psymb{\ensuremath{\operatorname{pred-symbols}}}

\newcommand\decgate{\ensuremath{\operatorname{dec-gate}}}
\newcommand\decgatel{\ensuremath{\operatorname{dec-gate}_l}}
\newcommand\decgatew{\ensuremath{\operatorname{dec-gate}_w}}

\newcommand\strength{\ensuremath{\operatorname{st}}}

\newcommand\lmpos[1][]{\FBattribute[#1]{pos}} 

\newcommand\lmpar[1][]{\ensuremath{#1\cdot\operatorname{parent}}}
\newcommand\lmpwiring[1][]{\FBattribute[#1]{parent-wiring}}
\newcommand\lmpcolor[1][]{\FBattribute[#1]{parent-label}}

\newcommand\gmcolor[1][]{\FBattribute[#1]{label_{2}}}
\newcommand\gmanclocal[1][]{\FBattribute[#1]{ancestor-msg}}

\newcommand\mloc[1][]{\FBattribute[#1]{parent-label_2}}
\newcommand\mglob[1][]{\FBattribute[#1]{global}}
\newcommand\mglobb{\operatorname{global}}

\newcommand\winputset[1][]{\FBattribute[#1]{Inputs}}
\newcommand\winertset[1][]{\FBattribute[#1]{Inert}}
\newcommand\woutputset[1][]{\FBattribute[#1]{Outputs}}

\newcommand\woutputmap[1][]{\FBattribute[#1]{output-num}}
\newcommand\woutputwires[1][]{\FBattribute[#1]{outwires}}

\newcommand\wirepos[2]{#1 (\textrm{with wire #2})}

\newcommand\gatefun[1][]{\FBattribute[#1]{func}}
\newcommand\gatewiring[1][]{\FBattribute[#1]{wiring}}
\newcommand\gatecolor[1][]{
  \def\tempa{}%
  \def\tempb{#1}%
  \ifx\tempa\tempb{}
    \ensuremath\operatorname{label}
  \else
    \ensuremath{\operatorname{label}(#1)}
  \fi
}

\newcommand\frepar[1]{
  \def\tempa{}%
  \def\tempb{#1}%
  \ifx\tempa\tempb{}
    \ensuremath{\operatorname{set-parent}}
  \else
    \ensuremath{\operatorname{set-parent}[#1]}
  \fi
}

\newcommand\fdecode[2][]{
  \def\tempa{}%
  \def\tempb{#1}%
  \ifx\tempa\tempb{}
  \ensuremath{\operatorname{decode}_{#2}}
  \else
  \ensuremath{\operatorname{decode}_{#2}[#1]}
  \fi
}
\newcommand\fincr[1]{
  \def\tempa{}%
  \def\tempb{#1}%
  \ifx\tempa\tempb{}
    \ensuremath{\operatorname{incr}}
  \else
    \ensuremath{\operatorname{incr}[#1]}
  \fi
}

\newcommand\fgincr[2]{
  \def\tempa{}%
  \def\tempb{#1}%
  \ifx\tempa\tempb{}
  \ensuremath{\operatorname{incr}_{#2}}
  \else
  \ensuremath{\operatorname{incr}_{#1}[#2]}
  \fi
}

\newcommand*\finst[1]{\ensuremath{\operatorname{instantiate}_{#1}}}

\colorlet{gold}{yellow!70!black}
\colorlet{rose}{magenta!50}
\newcommand\coldot[2][]{ \tikz{\node[yshift=-.2ex, circle, fill=#2] (#1) {};} }
\newcommand\tealdot[1][]{\coldot[#1]{teal}}
\newcommand\pinkdot[1][]{\coldot[#1]{rose}}
\newcommand\golddot[1][]{\coldot[#1]{gold}}

\newcommand\gammashape{
  \node[draw] (O) at (0, 0) {};
  \node[draw, right=0 of O] {};
  \node[draw, above=0 of O] {};
  \node[fill=red, circle, scale=.2, anchor=center] at (O.center) {};
}

\newcommand\drawcacarpet[1]{ 
  \pgfmathsetmacro\level{int(#1)}
  \ifnum #1 = 0 \relax
  \filldraw[fill=gray!50] (0, 0) -- (0, 1) -- (1, 1) -- (1, 0) -- cycle;
  \else
  \pgfmathsetmacro\next{int(#1 - 1)}
  \begin{scope}[scale=.166666666]
    \foreach \x in {0, ..., 5}
    \foreach \y in {0, ..., 5}
    {
      \ifthenelse{\numexpr\(\x=3\and\y=3\)\OR\(\x=2\and\y=2\)\OR\(\x=3\AND\y=2\)\OR\(\x=2\AND\y=3\)}
      {

      }
      {
        \begin{scope}[shift={(\x, \y)}]
          \drawcacarpet{\next}
        \end{scope}
      }
    }
  \end{scope}
  \fi
}

\newcommand\drawU[1]{
  \pgfmathsetmacro\level{int(#1)}
  \ifthenelse{\numexpr\level=0}{
    \filldraw[fill=gray!50] (0, 0) -- (0, 2) -- (3, 2) -- (3, 0) -- cycle;
  }{
    \pgfmathsetmacro\nextlevel{int(\level - 1)}
    \begin{scope}[scale=.33333333333333333333]
      \foreach \x in {0, ..., 2}
      \foreach \y in {0, ..., 1}{
        \ifthenelse{\numexpr\(\x=1\)\and\(\y=1\)}
        {
        }{
          \begin{scope}[shift={(3*\x, 2*\y)}]
            \drawU{\nextlevel};
          \end{scope}
        }
      }
    \end{scope}
  }
}
  
\newcommand\Cause{\operatorname{Cause}}
\newcommand\headroom{\operatorname{headroom}}
\newcommand\altitude{\operatorname{altitude}}
\newcommand\wexcursion{\operatorname{excursion}}

\newcommand\theFloor{\operatorname{floor}}

\newcommand\complete[1]{#1^\bullet}
\newcommand\treedec[1]{\tree[\complete{#1}]}

\newcommand\support[1][]{\FBattribute[#1]{support}}

\SetKw{Let}{let}
\SetKw{Be}{be}
\SetKw{Assert}{assert that}
\SetKw{Update}{update}
\SetKwFor{ElseForEach}{else foreach}{do}{endfch}
\usetikzlibrary{shapes.geometric}
\usetikzlibrary{calc}
\usetikzlibrary{patterns}
\usetikzlibrary{patterns.meta}
\usetikzlibrary{arrows.meta}
\usetikzlibrary{fadings}
\usetikzlibrary{backgrounds}
\usetikzlibrary {graphs} 
\usetikzlibrary{decorations, decorations.markings}
\usetikzlibrary{quotes}
\usetikzlibrary{positioning}

\tikzset{gate/.style={draw,circle}}

\newcommand\drawinputsW[2][]{\draw[-{Stealth}] ($ #2 +(-.7, -.3) $) .. controls ($ #2 +(-.5, -.5) $) .. #2 node[pos=-.2] {#1};}

\newcommand\drawinputNe[2][]{\draw[-{Stealth}] ($ #2 +(.3, .7) $) .. controls ($ #2 +(.5, .5) $) .. #2 node[pos=-.2] {#1};}

\newcommand\drawinputEw[2][]{\draw[-{Stealth}{Stealth}] ($ #2 +(-.7, 0) $) -- #2 node[pos=-.4] {#1};}
\newcommand\drawinputWe[2][]{\draw[-{Stealth}{Stealth}] ($ #2 +(.7, 0) $) -- #2 node[pos=-.4] {#1};}

\newcommand\drawinputE[2][]{\draw[-{Stealth}] ($ #2 +(.7, 0) $) -- #2 node[pos=-.5] {#1};}
\newcommand\drawinputW[2][]{\draw[-{Stealth}] ($ #2 +(-.7, 0) $) -- #2 node[pos=-.5] {#1};}

\newcommand\drawinputEn[2][]{\draw[-{Stealth}] ($ #2 +(.7, .3) $) .. controls ($ #2 +(.5, .5) $) .. #2 node[pos=-.2] {#1};}

\newcommand\drawinputSw[2][]{\draw[-{Stealth}] ($ #2 +(-.3, -.7) $) .. controls ($ #2 +(-.5, -.5) $) .. #2 node[pos=-.2] {#1};}
\newcommand\drawinputWs[2][]{\draw[-{Stealth}] ($ #2 +(-.7, -.4) $) .. controls ($ #2 +(-.5, -.5) $) .. #2 node[pos=-.2] {#1};}

\newcommand\drawoutputE[2][]{\draw[-{Stealth}] #2 edge node [pos=1.2] {#1} +(.9, 0);}
\newcommand\drawoutputN[2][]{\draw[-{Stealth}] #2 edge node [pos=1.2] {#1} +(0, .9);}
\newcommand\drawoutputnE[2][]{\draw[-{Stealth}] #2 .. controls +(.5, .5) .. +(.9, .1) node[pos=1.1] {#1};}
\newcommand\drawoutputnW[2][]{\draw[-{Stealth}] #2 .. controls +(-.5, .5) .. +(-.9, .1) node [pos=1.1] {#1};}

\newcommand\drawoutputNe[2][]{\draw[-{Stealth}] #2 .. controls +(.5, .5) .. +(.1, .9) node [pos=1.1] {#1};}
\newcommand\drawoutputsE[2][]{\draw[-{Stealth}] #2 .. controls +(.5, -.5) .. +(.9, -.1) node [pos=1.1] {#1};}
\newcommand\drawoutputSe[2][]{\draw[-{Stealth}] #2 .. controls +(.5, -.5) .. +(.1, -.9) node [pos=1.1] {#1};}

\tikzset{
  old inner xsep/.estore in=\oldinnerxsep,
  old inner ysep/.estore in=\oldinnerysep,
  double circle/.style 2 args={
    circle,
    old inner xsep=\pgfkeysvalueof{/pgf/inner xsep},
    old inner ysep=\pgfkeysvalueof{/pgf/inner ysep},
    /pgf/inner xsep=\oldinnerxsep+#1,
    /pgf/inner ysep=\oldinnerysep+#1,
    alias=sourcenode,
    append after command={
      let     \p1 = (sourcenode.center),
      \p2 = (sourcenode.east),
      \n1 = {\x2-\x1-#1-0.5*\pgflinewidth}
      in
      node [inner sep=0pt, draw, circle, minimum width=2*\n1,at=(\p1),#2] {}
    }
  },
  double circle/.default={2pt}{blue}
}

\newcommand\fractalify[4]{
  \pgfmathsetmacro\level{int(#1)}
  \ifthenelse{\numexpr\level=0}{
    \filldraw[fill=gray!50] (0, 0) -- (0, 1) -- (1, 1) -- (1, 0) -- cycle;
  }{
    \pgfmathsetmacro\nextlevel{\level - 1}
    \foreach \x [evaluate=\x as \expx using (\x - 1)*#3^\nextlevel] in {1, ..., #3} {
      \foreach \y [evaluate=\y as \expy using (\y-1)*#4^\nextlevel] in {1, ..., #4} {
        \arraytomacro#2[\y,\x]\cell
        \begin{scope}[shift={(\expx, \expy)}]
          \ifthenelse{\numexpr\cell=1}{
            \fractalify{\nextlevel}{#2}{#3}{#4}
          }{}
        \end{scope}
      }
    }
  }
}

\tikzset{
  singleGlue/.pic={
    \begin{scope}[scale=.2]
      \fill[fill=red] (-.4,0) -- (.4, 0) -- (.2, 1) -- (-.2, 1) -- cycle;
    \end{scope}
  }
}

\tikzset{
  doubleGlue/.pic={
    \begin{scope}[scale=.2]
      \fill[fill=black] (-.7, -1) rectangle (-.2, 0);
      \fill[fill=black] (.2, -1) rectangle (.7, 0);
    \end{scope}
  }
}

\tikzset{
  spaceinvader/.pic={
    \begin{scope}[scale=.08]
      \foreach \x/\y in {0/0, 0/1, 1/1, 1/2, 1/3, 2/0, 2/1, 2/3, 2/4, 3/1, 3/2, 3/3, 4/0, 4/1, 4/3, 4/4, 5/1, 5/2, 5/3, 6/0, 6/1}
      \filldraw[pic actions, shift={(\x, \y)}] (0, 0) rectangle (1, 1);
      \coordinate (-right-foot) at (6,0);
      \coordinate (-right-antenna) at (4,5);
      \coordinate (-left-foot) at (0,0);
    \end{scope}
  }
}

\tikzset{
  pics/posWithinK/.style args={#1/#2}{
    code={
      \begin{scope}[scale=.1, shift={(-3,-3)}]
        \draw[-] (0,0) grid (6,6);
        \fill[#2] (2,2) rectangle(4,4);
        \fill[black] #1 rectangle +(1,1);
        \coordinate(-east) at (6,3);
        \coordinate(-west) at (0,3);
      \end{scope}  
    }
  }
}

\usetikzlibrary{shapes.callouts}

\begin{document}

\maketitle

\begin{abstract}
  This paper gives a (polynomial time) algorithm to decide whether a given Discrete Self-Similar Fractal Shape can be assembled in the aTAM model.

  In the positive case, the construction relies on a Self-Assembling System in the aTAM which strictly assembles a particular self-similar fractal shape, namely a variant $K^\infty$ of the Sierpinski Carpet. We prove that the aTAM we propose is correct through a novel device, \emph{self-describing circuits} which are generally useful for rigorous yet readable proofs of the behaviour of aTAMs.

  We then discuss which self-similar fractals can or cannot be strictly self-assembled in the aTAM. It turns out that the ability of iterates of the generator to pass information is crucial: either this \emph{bandwidth} is eventually sufficient in both cardinal directions and $K^\infty$ appears within the fractal pattern after some finite number of iterations, or that bandwidth remains ever insufficient in one direction and any aTAM trying to self-assemble the shape will end up with an ultimately periodic pattern.
\end{abstract}

\section{Introduction}

There is an area of particular interest at the intersection between dynamical systems, models of computation and automata, and discrete geometry. There, researchers try to determine the limits of various models of computation when faced with geometrical constraints, and to relate them with the dynamic properties of the models. This paper deals with Self-Assembling Tilings in the Abstract Tile Assembly Model (aTAM), as introduced by Winfree~\cite{winfree_design_1998}. In addition to its \emph{raison d'être} as an abstract model of DNA computing, this model enjoys an interest as an intermediary between one-dimensional and two-dimensional Cellular Automata (CAs), as well as a kind of asynchronous version of CAs and tilings.

Two features distinguish this model from more classical models such as Cellular Automata: the fact that it consumes the space it computes in, forcing a steady growth of its output; as well as its asynchronicity. The aTAM is able to simulate classical models such as Turing Machines and Cellular Automata~\cite{winfree_design_1998} by covering large rectangles in a synchronized manner. The synchronization can be wrought from the jaws of the environment's asynchronicity using so called ``Temperature 2'' systems where the local attachement of a tile to a site can depend from the concommitent presence of two tiles on neighboring sites. This synchronization can also be substitued for by a richer local geometry, such as 3D space~\cite{cook_temperature_2011}, or the presence of a variety of tiles shapes~\cite{hendricks_power_2018}.

This hints that much of the subtelty of the asynchronous and write-only nature of the aTAM reflects in its ability to assemble fine geometrical features. Since 2010~\cite{patitz_self-assembly_2010}, a particular focus has been set on so-called Discrete Self-Similar Fractal Shapes (DSSFS) as candidates for being tough to assemble in the aTAM and allowing to pinpoint the subtelty of its variants. In order to meaningfully constrain the aTAM using the geometry of the shape that is being assembled, this line of research investigates its \emph{strict} self-assembly. That is, the authors demand that in order to assemble a shape $S$, no tile is ever put outside of $S$.

This program has seen some progress, with two lines of results. The first line of results~\cite{furcy_scaled_2017, barth_scaled_2014, DBLP:journals/tcs/LathropLS09, patitz_self-assembly_2010, kautz_self-assembling_2013} looks at close variants of DSSFS, either slightly laxer and showing that they \emph{can} be obtained by aTAM system, or restricted sets of DSSFS and showing that they cannot be obtained. The second line of results~\cite{hendricks_hierarchical_2020, hendricks_hierarchical_2018-2} looks at variants of the aTAM, generally more powerful ones, showing that they can assemble (some) DSSFS.

The big question since~\cite{patitz_self-assembly_2010} has been whether the \emph{plain} model (aTAM) can self-assemble some \emph{unadultered} DSSFS. The conjecture has generally been that their geometry is too complicated for the aTAM, that is that \emph{there was no aTAM system able to uniquely and strictly self-assemble a DSSFS}. In order to characterize these geometrical constraints, several notions have been proposed, such as $\zeta$-dimension and sparsity~\cite{patitz23, hader_fractal_2021}, and ease of disconnecting~\cite{hendricks_hierarchical_2020, furcy_scaled_2017, barth_scaled_2014}. This paper answers the question differently: not only there \emph{is} a DSFSS which can be strictly self-assembled, in fact most DSSFS can indeed be self-assembled, as long as a periodic tiling by their generator is well connected.

The positive construction of this paper relies on a family of ideas from the theory of Tilings. The hierarchical structure of discrete fractals is a lot alike that of many aperiodic tilings. In trying to assemble such a hierarchical shape, one can take inspiration from authors such as Mozes or Goodman-Strauss~\cite{goodman1998matching, mozes1989tilings} who use a system with several layers of information checking each other's work in order to impose from local rules the desired global hierarchy. Like theirs, the construction uses several layers carrying the same information at different scales. This part of the construction is somewhat simpler than in tilings thanks to the \emph{seeded} nature of the assembly. In tilings again, Durand, Romashchenko and Shen show in~\cite{DBLP:journals/jcss/DurandRS12} another way to impose such a structure, by obtaining their tiling through a Fixed-Point Theorem. The approach here through \emph{self-describing circuits} is a bastardization of these two, with some tricks of its own to deal with the lack of global synchronisation in the assembly. The way it is presented, through a circuit, is a necessity to deal with the proof of a big aTAM system acting in an environment with non-trivial geometric constraints, but it can certainly be useful as a tool to express and study many constructions in the aTAM model.

The characterization of the fractal shapes which \emph{cannot} be assembled in the aTAM relies on a Pumping Lemma in the tradition of~\cite{meunier_intrinsic_2014}. This Pumping Lemma limits the amount of computation that can be done in the assembly by looking at its \emph{treewidth}, or equivalently, the size of the largest square it encircles. This is a formalization of a new limit to computing power in the aTAM, and it should apply in a variety of settings, with a \emph{je-ne-sais-quoi} of Complexity Theory.

The presentation will start in~\cref{sec:defs} with the definitions of the abstract Tile Assembly Model, as well as the statement of the Tree Pump Lemma (\cref{lem:tree_pump}) and the difference between \emph{weak} and \emph{strict} assembly of shapes. The paper then covers the basics about discrete self-similar fractal shapes, and states the main positive result,~\cref{thm:main}. Then,~\cref{sec:abstract_models} presents Self-Describing Embedded Circuits which form the main proof device for~\cref{thm:main}; they are put to use in~\cref{sec:circuit}, concluding the proof of the positive construction. Finally,~\cref{sec:limits} delimitates the fractal shapes which can be constructed through the aTAM. The dichotomy there is quite sharp: for a given shape $S$, either the technique of \cref{thm:main} can be straightforwardly adaptated to $S$ or there is no way at all to obtain it in the aTAM. \cref{thm:kill_llama} concludes the paper with a polynomial-time algorithm to distinguish between the two situations.

\section{Definitions and statement of the main results}
\label{sec:defs}

\subsection{Notations}

A variable with an arrow such as $\vec{a}$ denotes a vector of values indexed by some set $S$. Given $s \in S$, $a_s$ is the element of $\vec{a}$ associated with $s$.

For any sets $A, B$ and $S$, given $\vec{a} \in A^S$ and $\vec{b} \in B^S$, the standard notation $\vec{a} \otimes \vec{b}$ is the vector $\vec{v} \in (A \times B)^S$ defined by $\forall s \in S, v_s = (a_s, b_s)$. Likewise, given $f: A^I \mapsto A^O$ and $g: B^I \mapsto B^O$, the function $f \otimes g: (A \times B)^I \mapsto (A \times B)^O$ is defined by $(f \otimes g)(\vec{a} \otimes \vec{b}) = f(\vec{a}) \otimes g(\vec{b})$.

The unit vectors $(0, 1), (1, 0), (0, -1), (-1, 0)$ of $\mathbb{Z}^2$ are noted $N$, $E$, $S$, $W$ respectively. An oriented edge of $\mathbb{Z}^2$ is an \emph{arc}, noted $(a \to b)$; its \emph{direction} is the unit vector $b - a \in \Dir$. Given an arc $e = (a \to b)$ and a vector $v \in \mathbb{Z}^2$, $e + v = (a + v \to b + v)$. The arc out of $a \in \mathbb{Z}^2$ in direction $d \in \Dir$ is $(a \to a + d)$.

Let $P$ be a finite subset of $\mathbb{N}^2$, with $w = \max(x | (x,y) \in P)$ and $h = max(y | (x,y) \in P)$. Then for $z = (x,y) \in \mathbb{Z}^2$, $\lfloor z / P\rfloor$ is the pair $(\lfloor x / w \rfloor, \lfloor y / h \rfloor)$ and $z \bmod P$ is the pair $(x \bmod w, y \bmod h)$.

The constructions in~\cref{sec:abstract_models} and~\ref{sec:circuit} make heavy use of sets built as cartesian products. In an attempt to make their exposition more readable and distinguish the function of each component of the product, they are given in a ``record'' or ``object'' like syntax using '$\{\}$' for record construction and ``$x \cdot \operatorname{field}$'' for field access. That is, given a base set $X$ and a finite set of $k$ labels such as $L = \{\operatorname{zeroth}, \operatorname{first}, \operatorname{second}, \ldots, \operatorname{not-quite-k-th}\}$, the element $x = (x_0, x_1, x_2, \ldots, x_{k-1})$ is written $x = \{ \operatorname{zeroth} = x_0, \operatorname{first} = x_1, \ldots, \operatorname{not-quite-k-th} = x_{k-1}\}$. Symetrically, $x_0$ can be extracted from $x$ by writing $\FBattribute[x]{zeroth}$. A judicious set of labels can make this notation actually useful.

\subsection{Self-assembly and the Abstract Tile-Assembly Model}

\subsubsection{Definitions}

There follows a brief exposition of the basic definition of the aTAM. The survey by Patitz~\cite{DBLP:journals/nc/Patitz14} as well as the classical article by Winfree~\cite{winfree_design_1998} give a far less telegraphic exposition.

\begin{definition}[Wang Tile]
  Given an alphabet $\Sigma$, a \emph{Wang Tile} is an element of $\Sigma^{\Dir}$, i.e. a unit square with an element of $\Sigma$ on each of its sides. Given $w \in \Sigma^{\Dir}$ and $d \in \Dir$, $w(d)$ is referred to as the \emph{color} of the $d$ \emph{side}  of $w$.
\end{definition}

\begin{definition}[Assembly]
  Given a set $\mathcal{W}$ of Wang Tiles, an \emph{assembly} of $\mathcal{W}$ is a partial function $A: \mathbb{Z}^2 \to \mathcal{W}$. It is finite if its domain is. The notation $z \in A$ should be read as $z \in \operatorname{dom}(A)$, or equivalently ``$A(z)$ is defined''.

  The set of assemblies on $\mathcal{W}$ is noted $\Assemblies{\mathcal{W}}$.
\end{definition}

\begin{definition}[aTAM]
  Let $\Sigma$ be an alphabet, an \emph{unseeded Tile Assembly System} $T$ is a triplet
  \[
      \begin{Bmatrix}
        \operatorname{tileset} &\in& \Sigma^{\Dir}\\
        \operatorname{strength-function} &:& \Sigma \to \mathbb{Z}\\
        \operatorname{temperature} &\in& \mathbb{N}
      \end{Bmatrix}
    .\]

  A \emph{seeded Tile Assembly System} $T'$ is a quadruplet
  \[
    \begin{Bmatrix}
      \operatorname{tileset} &\in& \Sigma^{\Dir}\\
      \operatorname{strength-function} &:& \Sigma \to \mathbb{Z}\\
      \operatorname{temperature} &\in& \mathbb{N}\\
      \operatorname{seed} &\in& \Assemblies{(\FBattribute[T']{tileset})}
    \end{Bmatrix}.\]

  In both cases, the notation $\Assemblies{T}$ refers to the set of assemblies $\Assemblies{(\FBattribute[T]{tileset})}$.
\end{definition}

The crucial differences between a Tile Assembly System and a mere set of Wang Tiles are its strength function and its temperature. The strength function defines the \emph{binding strength} of an edge of an assembly.

\begin{definition}[Binding]
  Let $\Ss$ be a Tile Assembly System, and $A \in \Assemblies{\Ss}$. Let $e$ be an edge between two positions $z$, $z' = z + d$ in $A$. The \emph{binding} $b(e)$ of $e$ in $A$ is
  \begin{itemize}
  \item $0$ if $A(z)(d) \neq A(z')(-d)$, and
  \item  $\strengthFun[\Ss](g)$ if $A(z)(d) = A(z')(-d) = g$.
\end{itemize}

  Given a set $E$ of edges of $A$, the binding strength of $E$ is $b(E) = \sum_{e \in E}b(e)$.
\end{definition}

This binding strength corresponds to forces tying the assembly together in the face of thermal agitation. This thermal agitation is modeled by $\temperature[\Ss]$. Whenever an assembly has a cut with a binding strength less than the temperature, it will tend to be torn along this cut by thermal agitation. Assemblies which do not have such cuts are \emph{stable}.

\begin{definition}[Stable assembly]
  Given a Tile Assembly System $\Ss$, an assembly $A \in \Assemblies{\Ss}$ is \emph{stable} if for any cut $C$ of $A$, $b(C) \geq \temperature[\Ss]$
\end{definition}

Stability begets a definition for the dynamics of the process of self-assembly: tiles are added to an assembly as long as the resulting new assembly is stable.

\begin{definition}[Attachment, Assembly Sequence]
  Given a Tile Assembly System $\Ss$ and an assembly $A \in \Assemblies{\Ss}$, an \emph{attachment candidate} is a pair $t@z$, with $t \in \tileset[\Ss]$ and $z \in \mathbb{Z}^2$. It is \emph{valid} if $z \notin \dom A$ and \emph{stable} if $A' = A \cup \{z \mapsto t\}$ is stable. If it is both valid and stable, it is an \emph{attachment}, noted $A \xrightarrow{t@z} A'$

  An assembly $A'$ \emph{follows} from an assembly $A$ if there is a sequence of attachments $A_0 = A \xrightarrow{t_1@z_1} A_1 \xrightarrow{t_2@z_2} \ldots \xrightarrow{t_{k}@z_{k}} A_k = A'$. Such a sequence is a \emph{assembly sequence} from $A$ to $A'$. It is noted $\alpha = A \rightarrow_\Ss A'$, and $A'$ is noted $\lim \alpha$. If $\alpha$ is infinite, then $\lim \alpha = \bigcup_i A_i$.

  The set of assembly sequences following from a given assembly $A$ is noted $\Seqs{\Ss, A}$, and $\Prods{\Ss, A} = \{ \lim \alpha | \alpha \in \Seqs{\Ss, A} \}$.

  An assembly is terminal if $\Prods{\Ss, A} = \{A\}$. The set of terminal assemblies following from $A$ is $\TerminalProds{\Ss, A}$.
\end{definition}

Given a seeded Tile Assembly System, its \emph{productions} are the assemblies following from its seed.

\begin{definition}[Production, Terminal Production]
  Given a seeded Tile Assembly System $\Ss$, a \emph{production} is an assembly which follows from $\seed[\Ss]$.

  The set of productions of $\mathcal{S}$ is noted $\Prods{\Ss}$, and the set of terminal productions is noted $\TerminalProds{\Ss}$. Likewise, $\Seqs{\Ss} = \Seqs{\Ss, \seed[\Ss]}$.
\end{definition}

Note that a production may be infinite but not terminal.

Given a subset $X \subset \mathbb{Z}^2$, the restriction of an assembly sequence $\alpha$ to $X$, noted $\alpha_{|X}$ is the sequence of attachments of $\alpha$ taking place within $X$. Likewise, for an assembly $A$, $A_{|X}$ is the restriction of $A$ to the positions within $X$.

\subsubsection{The Tree Pump Lemma}

No paper about self-assembly would be complete without a ``pump or die'' lemma in the style of \cite{meunier_intrinsic_2014}. This one examines the case of skinny productions, that is those who do not encircle any square larger than $N$ for some fixed $N$. The crucial property which makes them simple(-ish) is their bounded treewidth. The interest of this lemma ---especially for eight year olds--- is that it lets the user direct the flow of the firehose.

\begin{restatable}[Tree Pump]{lemma}{treepump}
  \label{lem:tree_pump}
  For any aTAM system $\mathcal{S}$ with $\seed[\Ss]$ finite and connected, define the following sets of assemblies:
  \begin{itemize}
  \item for any integer $m$, $C_{m}[\mathcal{S}]$ is the set of assemblies of $\mathcal{S}$ which encircle an $m \times m$ square;
  \item for any real $k$ and vector $\vec{d}$, $B_{k,\vec{d}}[\mathcal{S}]$ is the set of assemblies of $\mathcal{S}$ which do not cover any position $\vec{p}$ such that $\vec{p} \cdot \vec{d} > k + |\seed[\Ss]|$
  \item for any vector $\vec{d}$, $P_{\vec{d}}[\mathcal{S}]$ is the set of \emph{ultimately periodic} assemblies of $\mathcal{S}$ such that there is a vector $\vec{p}$ with $|\vec{p} \cdot \vec{d}| > 0$ and a non-empty sub-assembly $a \subseteq A$ such that $a + \vec{p} \subseteq a$.
  \end{itemize}

  Then, there is a function $F: \mathbb{N} \times \mathbb{N} \to \mathbb{N}$ such that for aTAM system $\mathcal{S}$ with $n$ tiles and a 1-tile seed, integer $m$ and unit vector $\vec{d}$ of $\mathbb{R}^2$,
  \[ \TerminalProds{\mathcal{S}} \cap (C_m[\mathcal{S}] \cup B_{F(n,m),\vec{d}}[\mathcal{S}] \cup P_{\vec{d}}[\mathcal{S}]) \neq \emptyset. \]
\end{restatable}

\begin{figure}[h]
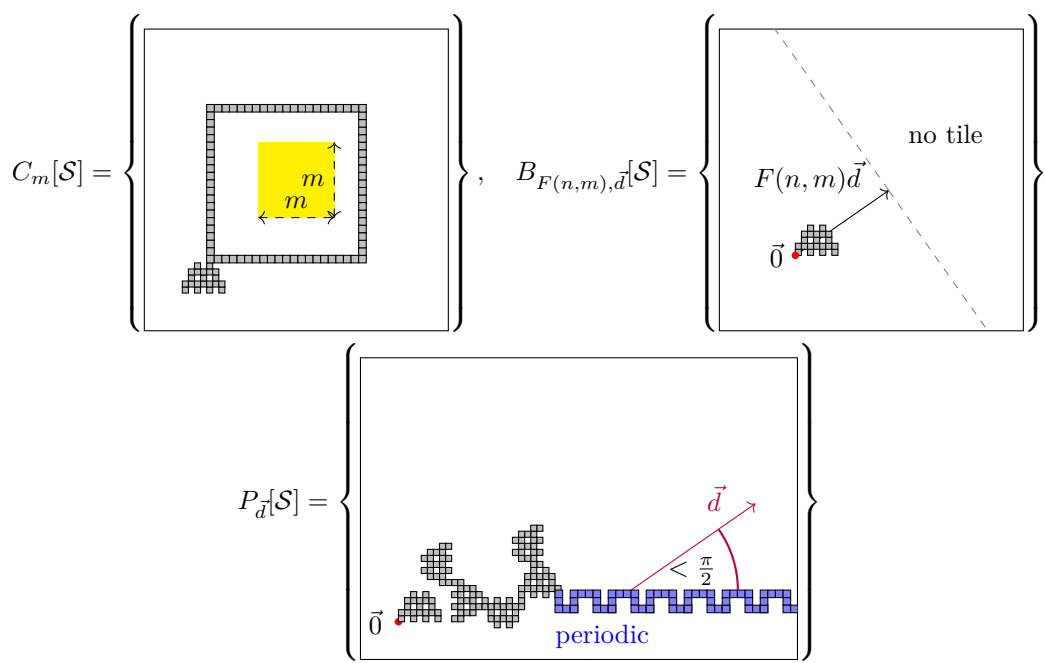

  \centering
    \tikzset{
      chemin/.pic={
        \begin{scope}[scale=.1]
          \foreach \x/\y in {0/0, 1/0, 2/0, 2/1, 2/2, 3/0, 1/2, 4/0, 4/1, 4/1}
          \path[draw, fill=gray!50!white] (\x, \y) -- ++(1,0) -- ++(0, 1) -- ++(-1,0) -- cycle;
          \coordinate (-bend) at (4, 0);
        \end{scope}
      }
    }

    \tikzset{
      u/.pic={
        \filldraw[fill=blue!50!white] (0,0) rectangle (.1, .1);
        \filldraw[fill=blue!50!white] (0,-.1) rectangle (.1, 0);
        \filldraw[fill=blue!50!white] (0,-.2) rectangle (.1, -.1);
        \filldraw[fill=blue!50!white] (.1,-.2) rectangle (.2, -.1);
        \filldraw[fill=blue!50!white] (.2,-.2) rectangle (.3, -.1);
        \filldraw[fill=blue!50!white] (.2,-.2) rectangle (.3, -.1);
        \filldraw[fill=blue!50!white] (.2,-.1) rectangle (.3, 0);
        \filldraw[fill=blue!50!white] (.2,0) rectangle (.3, .1);
        \filldraw[fill=blue!50!white] (.3,0) rectangle (.4, .1);
        \filldraw[fill=blue!50!white] (.4,0) rectangle (.5, .1);
      }
    }
    
    \tikzset{
      unbounded/.pic={
        \begin{scope}[scale=.5, shift={(-3, -3)}]
          \clip[draw] (-1, -1) rectangle (10.5,7);
          \node[fill=red, circle, inner sep=1pt] (origin) at (0, 0) {};
          \node[left=.1em of origin] {$\vec{0}$};
          \coordinate (path_dest) at (5.4, 2);
          \coordinate (foot) at (1.4, 0);
          \pic[fill=gray!50!white] (invader) at (0, 0) {spaceinvader};
          \pic[fill=gray!50!white, yscale=-1, rotate=-90] (invader2) at (foot) {spaceinvader};
          \pic[fill=gray!50!white, yscale=-1] (invader3) at (invader2-right-antenna) {spaceinvader};
          \pic[fill=gray!50!white] (invader4) at (invader3-right-foot) {spaceinvader};
          \pic[fill=gray!50!white, rotate=90] (invader5) at (invader4-right-antenna) {spaceinvader};
          \pic[fill=gray!50!white, rotate=90] (invader6) at (invader2-right-foot) {spaceinvader};
          
          \path (invader4-right-foot) -- +(10, 0) foreach \t in {0, .1, ..., 1} { pic[pos=\t] {u}};
          \node[blue, inner sep=0, fill=white, below=10pt of invader4-right-foot, anchor=north west] {periodic};

          \coordinate[shift={(1, .1)}] (a) at (invader4-right-foot);
          \coordinate[shift={(35:2)}] (b) at (a);
          \coordinate[shift={(1,0)}] (c) at (a);
          \path[purple, ->] (a) edge node[pos=.8, auto] {$\vec{d}$} (b); 
          \pic[pic text={$< \frac{\pi}{2}$}, draw=purple, thick, angle radius=4em] {angle = c--a--b};
        \end{scope}
      }      
    }

    \tikzset{
      bigCycle/.pic={
        \begin{scope}[scale=.5, shift={(-4, -4)}]
          \clip[draw] (-1, -1) rectangle (7,7);
          \pic[fill=gray!50!white] (invader) at (0, 0) {spaceinvader};
          \path [draw=none, fill=yellow] (2,2) rectangle (4,4);
          \path[<->, dashed] (2,2) edge node[midway, above] {$m$} (4,2);
          \path[<->, dashed] (4,2) edge node[midway, left] {$m$} (4,4);
          \path (invader-right-antenna) -- ++(4, 0) 
             foreach \t in {0, .05, ..., 1} { pic[pos=\t] {code={\filldraw[thin,fill=gray!50!white] (0,0) rectangle (.1, .1);}}}
             -- ++(0, 4)
             foreach \t in {0, .05, ..., 1} { pic[pos=\t] {code={\filldraw[thin,fill=gray!50!white] (0,0) rectangle (.1, .1);}}}
             -- ++(-4, 0)
             foreach \t in {0, .05, ..., 1} { pic[pos=\t] {code={\filldraw[thin,fill=gray!50!white] (0,0) rectangle (.1, .1);}}}
             -- cycle
             foreach \t in {0, .05, ..., 1} { pic[pos=\t] {code={\filldraw[thin,fill=gray!50!white] (0,0) rectangle (.1, .1);}}};
        \end{scope}
      }
    }

    \tikzset{
      bounded/.pic={
        \begin{scope}[scale=.5]
          \clip[draw] (-2, -2) rectangle (6, 6);
          \coordinate (origin) at (0, 0);
          \coordinate (d_circle) at (35:3);
          \coordinate (d_oppcircle) at (35:-2.5);
          \path[->] (origin) edge node[pos=.8, auto] {$F(n, m) \vec{d}$} (d_circle);
          \pic[fill=gray!50!white] (invader) at (0, 0) {spaceinvader};
          \node[fill=red, circle, inner sep=1pt] at (origin) {};
          \node[left=.1em of origin] {$\vec{0}$};
          \path[draw, gray, dashed] (d_circle) -- +(125:10);
          \path[draw, gray, dashed] (d_circle) -- +(125:-10);
          \fill[pattern=checkerboard, pattern color=red!50, path fading=east, rotate around={35:(d_circle)}] (d_circle) +(0, -10) rectangle +(10, 10);
          \path (d_circle) -- +(1.5,1.5) node [fill=white, rounded corners=2pt] {no tile};
        \end{scope}
      }
    }

    \begin{eqnarray*}
      &C_{m}[\mathcal{S}] = \left\{\tikz[baseline=(current bounding box.center)]{\pic{bigCycle};}\right\},\quad
      B_{F(n,m),\vec{d}}[\mathcal{S}] = \left\{\tikz[baseline=(current bounding box.center)]{\pic{bounded};}\right\}&\\
      &P_{\vec{d}}[\mathcal{S}] = \left\{\tikz[baseline=(current bounding box.center)]{\pic{unbounded};}\right\}&\\
    \end{eqnarray*}
    \caption{The three sets defined by~\cref{lem:tree_pump}: $C_{m}[\mathcal{S}]$ is the assemblies which encircle an $m \times m$ square, $B_{F(n, m), \vec{d}}[\mathcal{S}]$ is the set of assemblies which do not reach further than $F(n,m)$ in direction $\vec{d}$, and $P_{\vec{d}}[\mathcal{S}]$ is the assemblies which contain a periodic path with period $\vec{p}$ such that $\vec{p}\cdot\vec{d} > 0$.}
\label{fig:tree_pump}
\end{figure}

The Tree Pump Lemma states that in order to do a meaningful amount of computation, a self-assembling system with $n$ tiles needs to encircle large squares ---say, of size $m$, thus hitting $C_{m}[\Ss]$. If it does not, then its productions look very much like trees drawn on $\mathbb{Z}^2$ with a $m$-cell wide brush. A branch of that tree then behaves like a finite automaton. If that automaton stops, the assembly goes no further than $F(n, m)$ in any given direction $\vec{d}$, which gives a final production in $B_{F(n,m), \vec{d}}[\Ss]$. If not, these long branches must have an ultimately periodic behavior which gives a final production in $P_{\vec{d}}[\Ss]$.


The full proof is given in~\cref{apx:tree_pump}, as it involves a fair few ancillary definitions.


\subsection{Strict versus Weak Assembly of shapes}

\begin{definition}[Strict Assembly]
  An aTAM system $\Ss$ \emph{strictly self-assembles} a set $X \subset \Z^2$ if for all terminal production $\mathcal{P} \in \TerminalProds{\Ss}$, the domain of $\mathcal{P}$ is $X$.
\end{definition}

\begin{definition}[Weak Assembly]
  An aTAM system $\Ss$ \emph{weakly self-assembles} a set $X \subset \Z^2$ if there is $Y \subset \tileset[\Ss]$ such that for all terminal production $\mathcal{P} \in \TerminalProds{\Ss}$ and any position $z \in \mathbb{Z}^2$, $\mathcal{P}(z) \in Y \Leftrightarrow z \in X$.
\end{definition}

Weak and strict assembly are called thus because it is generally easier to weakly assemble a given shape than to strictly assemble it. Indeed, when weakly self-assembling a given shape $S$, it is possible to use positions in $\Z^2 \setminus S$ for computation. In strict self-assembly, doing so is impossible.

\begin{lemma}[Folklore]
  Any shape that can be strictly assembled can also be weakly assembled.
\end{lemma}

\begin{proof}
  Set $Y = \tileset[\Ss]$ in the definition of weak assembly.
\end{proof}

\subsection{Discrete Self-Similar Fractal Shapes}

Fractals are usually defined in some continuous space such as $\mathbb{R}^2$ or $\mathbb{C}^2$. In the context of self-assembly, one needs to work with \emph{discrete} fractals, that is, subsets of $\Z^2$ which are self-similar. In this paper, discrete self-similar fractal shapes are defined as fixed points of some 2 dimensional substitution. In~\cite{patitz_self-assembly_2010}, self-similar fractal shapes are defined as the union of an infinite hierarchy of shapes. The following definition is identical, as long as the generator contains $(0,0)$. Explicitely defining a geometrical substitution will help, as that substitution can guide constructions happening on the fractal. If the generator does not contain $(0, 0)$, the classical definition can be recovered by iterating the substitution associated with $G \cup (0, 0)$ to reach the fix-point, then iterating once the substitution associated with $G$.

The substitution $\sigma_G$ based on a finite pattern $G$ replaces each pixel in a pattern $X$ with a copy of $G$.

\begin{definition}[Rectangular substitution]
  Let $G$ be a finite subset of $\mathbb{N}^2$. Let $w = \max(\{x | (x,y) \in G\})$ and $h = \max(\{y | (x,y) \in G\})$. The substitution $\sigma_G: \mathcal{P}(\mathbb{N}^2) \to \mathcal{P}(\mathbb{N}^2)$ associated with $G$ is defined by:
  \[ \forall X \subset \mathbb{N}^2, \sigma_G(X) = \{ p \in \mathbb{N}^2 | \quot{p}{G} \in X \wedge p \bmod G \in G \} \]
\end{definition}

Using substitution allows manipulating the recursive definition of DSSF with richer data than one bit per position.

\begin{definition}[Colored substitution]
  Let $X$ be an alphabet and $G$ a finite subset of $\mathbb{N}^2$; given for each $x \in X$ a coloring $C(x) \in X^G$, the substitution $\sigma_C$ associated with $C$ is defined for each partial coloring $Y: \mathbb{N}^2 \dashrightarrow X$ by:
  \[
  \begin{cases}
    \sigma_C(Y): \mathbb{N}^2 \dashrightarrow X\\
    \sigma_C(Y): \vec{z} \mapsto C(Y(\quot{z}{G}))(z \bmod G),\\
  \end{cases}
  \]

  Where $\sigma_C(Y)(\vec{z})$ is defined whenever $Y(\quot{z}{G})$ is defined and $z \bmod G \in G$.
\end{definition}

\begin{figure}
  \centering
  \begin{tikzpicture}
    \input{tikz/colored_subst}
  \end{tikzpicture}
  \caption{A colored substitution on a 3-colored alphabet $\Sigma = \{ \tealdot, \pinkdot, \golddot \}$ is defined from a shape $G = \tikz{\gammashape}$ and a coloring of $G$ for each color of $\Sigma$. It can then be applied to any colored shape (right).}
  \label{fig:colored_subst}
\end{figure}

\begin{definition}[Self-Similar Fractal Shape]
  Let $G$ be a finite subset of $\mathbb{N}^2$ with $(0, 0) \in G$. The discrete self-similar fractal shape with generator $G$ is the following subset of $\mathbb{N}^2$:
  \[ G^{\infty} = \bigcup_{i \in \mathbb{N}} \sigma_G^i(\{(0,0)\}). \]
  The $i$-th step of the fractal is $G^i = \sigma_G^i(\{(0, 0)\}$.
\end{definition}

One self-similar fractal shape will be of particular interest to us in this paper, wich we name \emph{Sierpinski's Cacarpet}. It is the self-similar fractal shape $\cacarpet$ with generator \(K = \{0,\ldots, 5\}^2 \setminus \{2, 3\}^2\), represented on figure~\ref{fig:cacarpet}. We note $\kappa = \sigma_K$ the associated substitution.

\begin{figure}
  \centering
  \begin{tikzpicture}
    \input{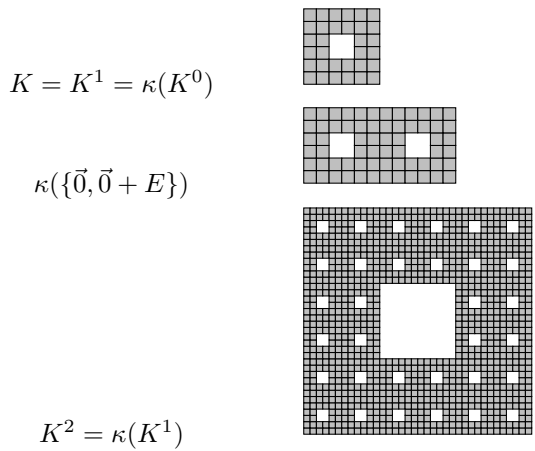}
  \end{tikzpicture}
  \caption{The substitution $\kappa$ of the Sierpinski's Cacarpet.}
  \label{fig:cacarpet}
\end{figure}






  



\subsection{The Sierpinski Cacarpet can be Strictly Self-Assembled}

Until this paper, in order to assemble a DSSF $\mathcal{S}$ in the aTAM, one would typically embed counters into $\mathcal{S}$, or rather some small superset of it \cite{patitz_self-assembly_2010, lutz_approximate_2012, kautz_self-assembling_2013}. Strict self-assembly of a DSSFS can be obtained in another way:
\begin{itemize}
\item first, choose a \emph{suitable} DSSF to assemble, namely the Sierpinski Cacarpet $K^\infty$
\item observe that $K^\infty$ is the fixed-point of a discrete substitution $\kappa$,
\item in consequence, define a hierarchy of systems: a \emph{self-describing circuit} $C_\square$, which entails an \emph{Locally Deterministic Oriented Pattern} $P_\square$, whose tiles make a temperature $2$ Tile Assembly System $S_\square$ in the abstract Tile Assembly Model which uniquely assembles $P_\square$.
\end{itemize}

\emph{In fine}, the correctness of $S_{\square}$ follows thus:
\begin{itemize}
\item $C_{\square}$ is defined as a fixed-point of $\kappa$, and thus has $K^\infty$ as its support,
\item the structure of $C_{\square}$ allows it to be \emph{evaluated}, which yields a pattern $P_{\square}$, with the same support,
\item moreover, $C_{\square}$ is \emph{self-describing} (Definition~\ref{def:self_describing} and Theorem~\ref{thm:cacarpet_circuit}), which gives $P_{\square}$ its \emph{local determinism} (Definition~\ref{def:loc_det_pat} and Lemma~\ref{lem:self-descr-local-det})
\item the set of tiles in $P_{\kappa}$ defines an aTAM System $S_\square$; since $P_{\square}$ is locally deterministic, the unique final production of $S_{\square}$ is $P_{\square}$ (Lemma~\ref{lem:det-pat-atam}).
\end{itemize}

\begin{theorem}
  \label{thm:main}
  There is an aTAM system at temperature 2 which self-assembles the DSSF pattern $K^\infty$.
\end{theorem}

The proof of \cref{thm:main} spans sections~\ref{sec:abstract_models} and \ref{sec:circuit}. The constructivist reader will find the detailed construction of $C_{\kappa}$ in section~\ref{sec:circuit}, this should enlighten them about most of the ideas of the construction. The formalist will find the tower of models and the links between them in section~\ref{sec:abstract_models}.

\section{Abstract models: Self-Describing Embedded Circuits and Locally Deterministic Patterns}
\label{sec:abstract_models}

\subsection{Embedded circuits}

A circuit is made of gates. Each of these gates has a \emph{function}, computing the outputs of the gates from the inputs. The gate is embedded on a unit square according to its \emph{wiring}, which states where its inputs come from and where its outputs go to. The wiring and the function must be compatible: the wiring must have as many inputs as the function.

\begin{definition}[Wire]
  A wire is a pair of a list of distinct directions $\vec{D} \in \Dir^*$ and a distinguished element $d \in D$. It is noted as a word on the alphabet $\{n, e, s, w\}$ containing the elements of $D$, with $d$ capitalized. For instance, the wire $((N, E), E)$ is noted $nE$, while $((N, E), N)$ is noted $Ne$.

  Given a wire $w = (\vec{D}, d)$, its opposite $-w$ is $(\vec{(-D_i)}, -d)$, i.e. $-Ne = Sw$.
  The set of wires is noted $\Wires$.
\end{definition}

\begin{definition}[Wiring]
  Given two integers $i, o$ such that $i + o \leq 4$, a \emph{$(i,o)$-wiring} $w$ (i.e. a wiring with \emph{in-degree} $i$ and \emph{out-degree} $o$) is given by
  \begin{itemize}
  \item a partition of $\Dir$ into three sets, $\winputset[w]$, $\woutputset[w]$ and $\winertset[w]$, where $\winputset[w]$ is ordered
  \item a map $\woutputmap[w]: \woutputset[w] \to \{0, \ldots, o - 1\}$;
  \item a map $\woutputwires[w]: \woutputset[w] \to \Wires$ such that $-d$ is the distinguished element of  $\woutputwires[w](d)$.
  \end{itemize}

  The set of all wirings is noted $\Wirings$.
\end{definition}

\begin{definition}[Input / Output / Inert Sides, Degree]
  Given a wiring $w \in \Wirings$,
  the elements of $\winputset[w]$ are the \emph{input sides} of $w$, the elements of $\woutputset[w]$ are its \emph{output sides}, and the elements of $\winertset[w]$ are its \emph{inert sides}.
  The in-degree of $w$ is the number of its input sides, its out-degree of $w$ is its number of output sides.
\end{definition}


Additionally, each output side also determines the \emph{input} sides on the following gate and their ordering. That is, a wiring $w$ with $nE \in \woutputwires[w]$ can only have a wiring $w'$ with $\winputset[w'] = (S, W)$  to its right.

The graphical representation of a wiring is given on \cref{fig:wiring_example}. Each direction in $\winputset[w]$ has an incoming arrow, and each direction in $\woutputset[w]$ has an outgoing arrow. If $\winputset[w]$ has only one element, the corresponding arrow is straight and simple. Likewise for the outgoing arrow when $\woutputwires[w](d)$ is a singleton. When $\woutputwires[w](d)$ is made of two adjacent directions, the corresponding arrows bend to come near each other; likewise, modulo rotation, when $\woutputwires[w](N) \in \{Ne, eN\}$, the corresponding arrow bends left, and right if it is $Nw$ or $wN$. For pairs of opposite directions, an arrow with a double tip is used. Triple and quadruple inputs are not used in this paper. Inputs are numbered according to the order of $\winputset[w]$ and output wires are numbered according to $\woutputmap[w]$.

\begin{figure}
  \centering
  \begin{tikzpicture}
    \matrix (b) [matrix of nodes, column sep=10em, row sep=10ex, every node/.style={draw, circle}]{
  $w_1$ & $w_2$ & $w_3$ \\
};

\begin{scope}[every node/.style={font=\footnotesize}, entry/.style={pos=.9,auto}, exit/.style={pos=.1,auto,swap}]
\draw (b-1-1)+(-1,0) edge[-{Stealth}] node[entry] {$0$} (b-1-1);
\draw[-{Stealth}] (b-1-1) edge node[exit] {$0$} node[entry] {$0$} ++(1,0);
\draw[-{Stealth}] (b-1-1) edge[bend right] node[exit] {$1$} node[entry] {$0$} ++(0,1);
\draw[-{Stealth}] (b-1-1) edge[bend left] node[exit] {$1$} node[entry] {$1$} ++(0,-1);

\draw (b-1-2)+(-1,0) edge[-{Stealth}, bend right] node[entry] {$0$} (b-1-2);
\draw (b-1-2)+(0,-1) edge[-{Stealth}, bend left] node[entry, pos=.8] {$1$} (b-1-2);
\draw (b-1-2) edge[-{Stealth}{Stealth}] node[exit] {$0$} node[entry] {$0$} ++(0,1);

\draw (b-1-3)+(-1,0) edge[-{Stealth}{Stealth}] node[entry] {$0$} (b-1-3);
\draw (b-1-3)+(1,0)  edge[-{Stealth}{Stealth}] node[entry] {$1$} (b-1-3);
\draw (b-1-3) edge[-{Stealth}, bend left] node[exit] {$0$} node[entry] {$1$} ++(0,-1);
\end{scope}

  \end{tikzpicture}
  \caption{The graphical representation of three wirings: \( w_1 = \{ \winputset = (W), \woutputset = \{ N, S, E \}, \woutputmap = \{ N \to 1, E \to 0, S \to 1 \}, \woutputwires = \{ Se, W, wN \} \} \), \( w_2 = \{ \winputset = (W, S), \woutputset = \{N\}, \woutputmap = \{ N \to 0 \}, \woutputwires = \{ Sn \} \} \) and \( w_3 = \{ \winputset = (W, E), \woutputset = \{ S \}, \woutputmap = \{ S \to 0 \}, \woutputwires = \{ wN \}  \} \). Neighboring input wires bend to come together into the wiring. Double arrows indicate a pair of opposite input arrows. The output wire in direction $d$ bends to mirror \(\woutputwires[w](d)\).}
  \label{fig:wiring_example}
\end{figure}

\begin{definition}[Gate]
  Given an alphabet $\Sigma$ two integers $i, o$, a gate $g$ is a pair $\{\gatefun: f, \gatewiring: w\}$, with $f$ a function $f: \Sigma^i \to \Sigma^o$ and $w$ an $(i, o)$-wiring.

  The set of gates on alphabet $\Sigma$ is noted $\Gates{\Sigma}$.
\end{definition}

These gates may be placed onto a subset of the grid $\Z^2$ to make \emph{circuits}. An example of a circuit, a multiplier built from a handful of adders and auxiliary gates is given on \cref{fig:mod3-adder}.

%
%
\begin{definition}[Circuit]
   Let $\Sigma$ be an alphabet. Let $D$ be an oriented subgraph of $\Z^2$, $\vec{I}$ be a sequence of arcs from $\Z^2 \setminus D$ to $D$ and $\vec{O}$ a sequence of arcs from $D$ to $\Z^2 \setminus D$.

   A circuit $C$ with dependency graph $D$, input bus $\vec{I}$ and output bus $\vec{0}$ is a map $D \to \Gates{\Sigma}$ such that:
   \begin{itemize}
   \item for any arc $E$ in direction $d \in \Dir$ between two positions $a, b \in D$, let $w_a = \gatewiring[C(a)]$ and $w_b = \gatewiring[C(b)]$; then \( d \in \woutputset[w_a] \), \( -d \in \winputset[w_b] \), and \(\woutputset[w_a](d) = (-d, \winputset[w_b])\),
   \item for any adjacent positions $a, b$ in $\mathbb{Z}^2$ with no arc between $a$ and $b$ in $D$, $d \in \winertset[ { \gatewiring[C(a)]  } ]$ and $-d \in \winertset[ {\gatewiring[ {C(b)} ]} ]$
   \item and for any arc $e$ of $\mathbb{Z}^2$ in direction $d \in \Dir$ between two positions $a \in D$ and $b \in  \Z^2 \setminus D$, either:
    \begin{itemize}
    \item $d \in \winertset[ { \gatewiring[ {C(a)} ] } ]$ and $e$ is neither an element of $\vec{I}$ nor of $\vec{O}$,
    \item $d \in \woutputset [{ \gatewiring[ {C(a)} ] } ]$ and $e$ is an element of $\vec{O}$,
    \item $-d \in \winputset[ { \gatewiring[ {C(a)} ] } ]$ and $-e$ is an element of $\vec{I}$.
    \end{itemize}
   \end{itemize}

   The support of $C$ is the vertex-set of $D$.
\end{definition}

\begin{figure}
  \centering
  \begin{tikzpicture}
    \matrix (c) [every node/.style={draw, circle}, matrix of nodes, column sep=3em, row sep = 7ex, a/.style={rectangle, draw=none, inner sep=1mm}] {
  &       & |[a]| $I_0$ & |[a]| $I_1$ &  & \\
  & $i_1$	& $i_2$       & $m$         & $i_1$ & \\
  |[a]| $I_2$	& $i_1$ & $m$         & $a$ & $a$        & |[a]| $O_2$ \\
  &       & |[a]| $O_0$ & |[a]| $O_1$ & & \\
};

\begin{scope}[every node/.style={font=\footnotesize}, entry/.style={pos=.9,auto}, exit/.style={pos=.1,auto,swap}, value/.style={pos=.5, fill=gray!20}]
\draw[->] (c-1-3) edge[bend right] node[entry] {$0$} node[value] {$2$} (c-2-3);
\draw[->] (c-1-4) edge[bend right] node[entry] {$0$} node[value] {$3$} (c-2-4);
\draw[->] (c-2-3) edge[bend right] node [exit] {$0$} node[entry] {$0$} node[value] {$2$} (c-3-3);
\draw[->] (c-2-4) edge[bend right] node [exit] {$0$} node[entry] {$0$} node[value] {$1$} (c-3-4);
\draw[->] (c-3-2) edge node [exit] {$0$} node[entry] {$0$} node[value] {$7$} (c-2-2);
\draw[->] (c-3-3) edge node[exit] {$0$} node[value] {$4$} (c-4-3);
\draw[->] (c-3-4) edge node[exit] {$0$} node[value] {$2$} (c-4-4);
\draw[->] (c-2-4) edge node[exit] {1} node[value] {$2$} node[entry] {0} (c-2-5);
\draw[->] (c-3-5) edge node[exit] {0} node[value] {$2$} (c-3-6);

\draw[->] (c-2-2) edge[bend left] node[exit] {$0$} node[entry, swap] {$1$} node[value] {$7$} (c-2-3);  
\draw[->] (c-2-3) edge[bend left] node[exit] {$1$} node[entry, swap] {$1$} node[value] {$7$} (c-2-4);

\draw[->] (c-3-1) edge node[entry] {$0$} node[value] {$7$} (c-3-2);
\draw[->] (c-3-2) edge[bend left] node[exit] {$0$} node[entry, swap] {$1$} node[value] {$7$} (c-3-3);
\draw[->] (c-3-3)  edge[bend left] node[exit] {$1$} node[entry, swap] {$1$} node[value] {$1$} (c-3-4);
\draw[->] (c-3-4) edge[bend left] node[exit] {$1$} node[value] {0} node[entry] {0} (c-3-5);
\draw[->] (c-2-5) edge[bend right] node[exit] {0} node[entry, swap] {1} node[value] {2} (c-3-5);
\end{scope}

  \end{tikzpicture}
  \[
    \begin{cases}
      i_1 : \Sigma \to \Sigma \\
      i_1(x) = x\\
      i_2 : \Sigma^2 \to \Sigma^2 \\
      i_2(x,y) = (x,y)\\
    \end{cases}\quad%
    \begin{cases}
      m : \Sigma^2 \to \Sigma^2 \\
      m(x,y) = (xy \bmod 10, \quot{xy}{10})\\
      a : \Sigma^2 \to \Sigma^2 \\
      a(x,y) = ((x + y) \bmod 10, \quot{(x + y)}{10})\\
    \end{cases}
  \]
  
  \caption{A multiplier circuit $M_2^1$ and the definition of the functions of its gates.  The alphabet is the set of digits $\mathbb{Z} / 10 \mathbb{Z}$. The values in the grey squares are an example of evaluation: $32 \times 7 = 224$}
  \label{fig:mod3-adder}
\end{figure}


Let $C$ be a circuit, and $A$ the arc-set of its dependency graph $D$. A function $v: A \mapsto \Sigma$ \emph{conforms} to $C$ at a position $\vec{z}$ in the support of $C$, if for any \(d \in \woutputset[ {\gatewiring[ C(\vec{z}) ]} ] \), \( v(o) = (\gatefun[C(v)])(v(i_0), \ldots, v(i_m))_k \) with $i_0, \ldots, i_m$ the input arcs of $C(\vec{z})$, $o$ its output arc in direction $d$ and $k = \woutputmap[ {\gatewiring[ C(\vec{z}) ]} ](d)$.

A circuit is \emph{well-founded} when its dependency graph has no cycle or infinite backwards paths. For a well-founded circuit, given a vector of inputs $\vec{ \imath }$ indexed by $\vec{I}$, there is a \emph{unique} function $\innerevalfunc{C}(\vec{ \imath }, -)$ which conforms to $C$ and such that the values on $\vec{I}$ are $\vec{ \imath }$. The circuit function $\evalfunc{C}: \Sigma^{ \vec{I}} \to \Sigma^{\vec{O}}$ is defined by $\evalfunc{C}(\vec{ \imath }) = (\innerevalfunc{C}(\vec{ \imath }, O_k ))_k$.

A circuit is \emph{finitely rooted} when its input bus is finite and its number of gates with in-degree $0$ is finite. It is \emph{evaluable} if it is well-founded and finitely rooted.

The circuit $M^1_2$ on figure~\ref{fig:mod3-adder} is a multiplier. Take two natural integers $a < 10, b < 100$, pose $b = b_0 + 10 b_1$ with $\forall i, 0 \leq b_i < 10$. Let $\vec{\imath}_{a,b} = (a_0, b_0, b_1)$. Then $\evalfunc{C}(\vec{\imath}_{a,b}) = \vec{o}$, with $ab = \sum o_i 10^i$. Figure~\ref{fig:mod3-adder} gives an example of the evaluation of the circuit on inputs $32$ and $7$; the values of $\innerevalfunc{C}(\vec{i}, -)$ on each arc are given on the figure.

When the alphabet $\Sigma$ is structured in two layers, i.e. $\Sigma = A \times B$, the function $f$ of a gate $g$ may act independently on layer $A$ and layer $B$, i.e. $f = f_A \otimes f_B$. Then, taking $g_A = \{ \gatewiring : w, \gatefun: f_A \} \in \Gates{A}$ and $g_B = \{ \gatewiring : w, \gatefun: f_B \} \in \Gates{B}$, $g$ can be written as $g_A \otimes g_B$. In other words, the tensor product $\otimes$ applies to pairs of gates \emph{with the same wiring}.




\subsection{Self-description}

The process of self-assembly in the aTAM and the evaluation of an evaluable embedded circuit are somewhat alike, in that information propagates from an initial region outwards. In both processes, there is no global synchronisation, but each step can take place as soon as its inputs are ready. In the aTAM, the initial region is the seed, while in a circuit, it is the input bus together with the gates of in-degree \(0\).

There are three differences between these processes. First, the input bus of a circuit does not have an analog in the aTAM. Hence, aTAM systems are akin to \emph{closed} circuits. Secondly, there can be competition between attachments in the aTAM, as well as mismatches, both of which are ruled out in cricuits. This entails that an aTAM derived from a circuit by the compilation process detailed below will not exhibit mismatches or concurrent attachments. The last and most important difference is that in a circuit, the outputs of each gate depends not only on its input, but also on the function of the gate. For an aTAM system to simulate a circuit, this information needs to be derived from the inputs of the gate. Self-description captures the possibility to do so.

\begin{definition}[Self-describing circuit]
  \label{def:self_describing}
  A normal circuit $C$ on alphabet $\Sigma$ is \emph{self-describing} on input vector $\vec{ \imath }$ if there is a decoding function \(\decgate: (\Sigma \times \Dir)^{<4} \to \Gates{\Sigma}\) such that for any position $p$ with incoming arcs \((e_0, \ldots, e_{k-1}) = p + (\winputset[ {\gatewiring[C(p)] } ])\):
  \[\decgate((\innerevalfunc{C}(\vec{\imath}, e_0), d_0), \ldots, (\innerevalfunc{C}(\vec{\imath}, e_{k-1}), d_{k-1})) = C(p),\]
  where $d_j$ is the direction of $e_j$.
\end{definition}

A closed self-describing circuit really is a circuit with only one function of each in-degree. Having one function per in-degree is a harmless technicality: in a circuit, successive gates have compatible wirings, so the in-degree of each gate is known from the wiring of any of its predecessor gates, whether or not the circuit is self-describing.

\begin{lemma}
  \label{lem:one-gate}

  Let $C$ be a closed self-describing circuit. There is a closed circuit $C_1$ with only only one function of each in-degree such that $\innerevalfunc{C} = \innerevalfunc{C_1}$.
\end{lemma}

\begin{proof}
  $C_1$ has the same wirings as $C$, and for each $i$, all of its gates of in-degree $i$ have as function \(f_i: x_0 \ldots x_{i-1} \mapsto (\gatefun[\delta_C(x_0, \ldots x_{i-1})])(x_0, \ldots x_{i-1})\), where $\delta_C$ is the decoding function of $C$.
\end{proof}

\subsection{Locally deterministic patterns}
\label{sec:loc_det_pat}

If $C$ is a self-describing closed circuit, then from $\innerevalfunc{C}$, one gets a coloring of the arcs of its dependency graph with the property of \emph{local determinism}.

\begin{definition}[Locally deterministic arc coloring]
  \label{def:loc_det_pat}
  Let $P$ be a coloring of the arcs of some acyclic oriented subgraph $G$ of $\Z^2$. For each vertex $v \in G$ with in-degree $\delta$, note \(\vec{\imath}_v = ((d_0, P(e_0)), \ldots, (d_{\delta-1}, P(e_{\delta-1})))\) for the input vector of $v$, where $e_i$ is the $i$-th incoming edge of \(v\) and $d_i$ its direction

  $P$ is locally deterministic if there  are two prediction functions $\pinputs: \Dir \times \Sigma \to \PDir$ and $\psymb: (\Dir \times \Sigma)^{< 4} \times \Dir \to \Sigma$ such that for each arc $e$ from vertex $a$ to vertex $b$ in direction $d \in \Dir$,
  \begin{itemize}
  \item \(P(e) = \psymb(\vec{\imath}_a, d)\)
  \item \(\pinputs(d, P(e))\) is the set of directions of the incoming arcs of $b$,
  \end{itemize}
\end{definition}

\begin{figure}
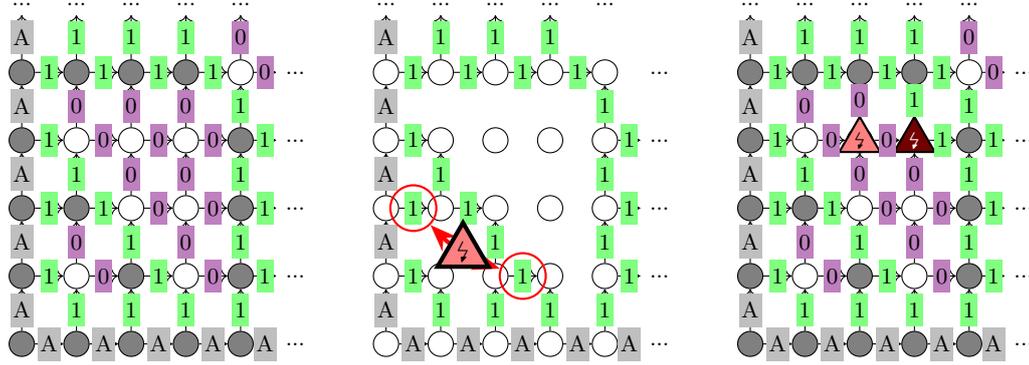

  \centering
  \def\SierpData{%
    1 1 1 1 1
    1 0 1 0 1
    1 1 0 0 1
    1 0 0 0 1
    1 1 1 1 0
  }
  \readarray\SierpData\Sierp[-,5]

  \def\z{\Sierp[\x,\y]}

  \begin{tikzpicture}[xscale=.48, yscale=.6, font=\footnotesize]
    \input{tikz/locally_det_pat_ok}
  \end{tikzpicture}\qquad%
  \begin{tikzpicture}[xscale=.48, yscale=.6, font=\footnotesize]
    \input{tikz/locally_det_pat_no_zero}
  \end{tikzpicture}\qquad%
  \begin{tikzpicture}[xscale=.48, yscale=.6, font=\footnotesize]
    \input{tikz/locally_nondet_pat}
  \end{tikzpicture}
  \caption{Three arc colorings on the alphabet $\Sigma = \{ A, \textcolor{violet}{0}, \textcolor{green!50!black}{1} \}$. All arcs go either to the right or up. The coloring $C_1$, in the upper left is locally deterministic. In the upper-right, $C_2$ is not because the two circled arcs have value $1$ and direction $E$, but the vertices they point into have different incoming directions ($\{E\}$ versus $\{N, E\}$). At the bottom, $C_3$ is not locally deterministic, because the two triangle vertices get the same incoming values, but have different outputs: $(0, 0)$ versus $(1, 1)$.}
  \label{fig:ex-locally-deterministic}
\end{figure}

\begin{lemma}
  \label{lem:self-descr-local-det}
  Let $C$ be a closed self-describing circuit on alphabet $\Sigma$, with dependency graph $D_C = (V, A)$.
  Let $C'$ be the self-describing circuit on alphabet $\Sigma \times \Wirings$ where each gate outputs its wiring in addition to its normal output.
  Then the function $\innerevalfunc{C'}$ is a locally deterministic coloring of $D_C$.
\end{lemma}

\begin{proof}

  A candidate function for $\pinputs$ takes as input the direction $d$ of an arc $e: a \to b$, and the value of $\innerevalfunc{C'}(e)$, that is, an element of $\Sigma$ corresponding to $\innerevalfunc{C}(e)$ as well as the wiring $w$ of the gate at $a$. Thus, the function $(d, x, w) \mapsto \woutputwires[w](d)$ must, by the constraints on neighboring gates in a circuit, output the set of directions of the incoming arcs at $b$.

  Likewise, a candidate function for $\psymb$ takes as input the direction of the arc $e: a \to b$, and a vector $\vec{m} \otimes \vec{d}$ of the inputs coming into $a$ with their directions. Reorder $\vec{m} \otimes \vec{d}$ so that $\vec{d}$ is in the order of the inputs of $C(a)$. Let $\{ \gatewiring: w, \gatefun : f \} = \decgate(\vec{m} \otimes \vec{d})$. Then taking $\psymb(e) = (f(\vec{m}), w)$ works.
\end{proof}

\paragraph*{From locally deterministic patterns to aTAM systems}

\begin{definition}[vertex type, atlas]
  Let $P$ be a coloring of the arcs of some acyclic oriented subgraph $G$ of $\Z^2$. For a vertex $v$ of $G$, its \emph{vertex type} is the partial function $\bar{v}: d \in \Dir \mapsto (o, x)$, where $o = -d$ if the edge $e$ in direction $d$ from $v$ is an incoming edge, $o = d$ if it is an outgoing edge, and $x \in \Sigma$ is $P(e)$. If $v$ has no edge in direction $d$, then $\bar{v}(d)$ is undefined.

  The \emph{atlas} of $P$ is the set of its vertex types.
\end{definition}

A locally deterministic coloring can be reconstructed from its atlas and the positions of its in-degree $0$ vertices.

\begin{lemma}
\label{lem:same-atlas-same-pattern}
  Let $P_1, P_2$ be locally deterministic colorings of the arcs of some oriented subgraphs of $\Z^2$, $G_1$ and $G_2$ respectively. If
  \begin{enumerate}[(1)]
  \item $G_1$ and $G_2$ are acyclic and have no infinite backwards paths,
  \item $P_1$ and $P_2$ have the same atlas,
  \item $G_1$ and $G_2$ have the same vertices with in-degree $0$,
  \end{enumerate}
  then $G_1 = G_2$ and $P_1 = P_2$.
\end{lemma}

\begin{proof}
  By induction on the longest path to each position in $G_1$.
\end{proof}

Finally, a locally deterministic pattern with a unique in-degree 0 vertex can be self-assembled in the aTAM.

\begin{lemma}\label{lem:det-pat-atam}
  Let $\Sigma$ be some finite alphabet, and let $P$ be a locally deterministic coloring of the arcs of some acyclic graph $G$ with no infinite backwards path. Assume the maximal in-degree of a vertex in $G$ is $\delta$ and that $G$ has only one vertex with in-degree $0$.

  Then there is an aTAM system $S_P$ with temperature $\delta$ with $P$ as its only final production.
\end{lemma}

\begin{proof}
  Let $A_P$ be the atlas of $P$, $\pinputs$ and $\psymb$ the prediction functions of $P$. By convention, suppose $\pinputs$ returns an input vectors which is ordered clockwise, starting from direction $N$.

  First define the glues of $S_p$ and their strength function $\strength$. Each glue is a pair $(s, d) \in \Sigma \times \Dir$. For any symbol $s \in \Sigma$ and direction $d \in \Dir$, let $\vec{\imath} = \pinputs(s, d)$. If $d$ is the first element of $\vec{\imath}$, then its strength is \(\strength((s,d))=(\delta + 1 - |\vec{\imath}|\)), otherwise it is $1$. Hence, for any vertex type $v \in A_P$ with input vector $\vec{\imath}_v$, \(\sum_{g \in \vec{\imath}_v} \strength(g) = \delta\).

  Then define the set of tile types $S_p$ to be $A_p$, with the glue on the $d$ side of the tile type $\bar{v}$ being $\bar{v}(d)$ if defined, and the null glue $\null$ otherwise. The seed tile of $S_p$ is the only vertex type of $A_p$ with in-degree $0$.

  Indeed, any production of $S_P$ is an initial subset of $P$, as can be seen by induction on the attachments.

  Then, consider a production $p$ of $S_P$. If \(G\) contains some position not covered by $p$, then because $G$ contains neither loops nor infinite paths, it must contain some $v$ with all its predecessors within \(\operatorname{dom}(p)\). But then the total strength of glues into $v$ is $\delta$, and the tile corresponding to $P(v)$ must be attachable there. Hence $p$ is not terminal.

  Finally, $S_P$ assembles $P$.
\end{proof}

Putting all this together, it is possible to compile a self-describing circuit into a self-assembling system.

\begin{theorem}
  \label{thm:circuit-to-atam}
  Let $C$ be a self-describing normal circuit with only one gate without inputs. Then there is an aTAM system $S_C$ which strictly self-assembles $\operatorname{dom}(C)$.
\end{theorem}
\begin{proof}
  By lemmas~\ref{lem:self-descr-local-det} and~\ref{lem:det-pat-atam}.
\end{proof}

\section{A Self-Describing Embedded Circuit for the Sierpinski Cacarpet}
\label{sec:circuit}

The next step is to define a self-describing circuit on the Sierpinski Cacarpet.

\begin{theorem}
  \label{thm:cacarpet_circuit}
  There is an evaluable circuit $C_{\square}$ with domain $\cacarpet$ and with an empty input bus which is self describing. Moreover, $C_{\square}$ has only one gate without inputs.
\end{theorem}

The remainder of this section is the description of $C_{\square}$. This circuit is built from two sets of messages: layer 1 messages in $\Sigma_1$ and layer 2 messages in $\Sigma_2$, and three fractal structures: a wiring layer $W_{\square}: K^\infty \to \Wirings$, and two function layers, $F_1$ operating on $\Sigma_1$ and $F_2$ operating on $\Sigma_2$.

The wiring layer $C_{w}: K^{\infty} \to \Wirings$ is the fixed point of a substitution $\kappa_w: \Wirings \to \Wirings^K$ starting from a seed wiring $s_w$:
\[
  \begin{cases}
    C_{w}(\vec{0}) = s_w\\
    C_{w}(\vec{z}) = \kappa_w(C_{w}\lfloor \frac{z}{K} \rfloor)(z \bmod K)\\
  \end{cases}
\]

The definition of the function layers is a bit more indirect: there are two sets of labels $L_1$ and $L_2$ respectively, with two rules $\kappa_1: \Wirings \to L_1^K$ and $\kappa_2: \Wirings \times L_2 \to L_2^K$. The rule $\kappa_2$ takes the form of a substitution.

Their fixpoints, starting from two seed labels $\mathfrak{s}_1$ and $\mathfrak{s}_2$ define two tiling of $K^\infty$ with labels of $L_1$ and $L_2$ respectively: $\Lambda_1$ and $\Lambda_2$. On layer $1$, $\Lambda_1$ is defined for each position according to the wiring of its parent:
\[
  \Lambda_1(\vec{z}) = \kappa_1(C_{w}(\lfloor\frac{z}{K}\rfloor))(z \bmod K)
\]

On layer $2$, $\Lambda_2$ is defined for each position according to the wiring \emph{and label} of its parent in a substitutive manner:
\[
  \begin{cases}
    \Lambda_2(\vec{0}) = \mathfrak{s}_2\\
    \Lambda_2(\vec{z}) = \kappa_2(C_{w}(\lfloor\frac{z}{K}\rfloor), C_{2}(\lfloor\frac{z}{K}\rfloor))(z \bmod K)\\
  \end{cases}.
\]

Finally, the actual gate functions are defined from two functions, $\finst{1}$ on layer 1 and $\finst{2}$ on layer 2, which define a function on $\Sigma_1$ (respectively $\Sigma_2$) from four elements, two wirings and two labels in $L_1$ (respectively $L_2$), one each for the gate and its parent. Together with the substitutions $\kappa_i$, they define a circuit $\mu_i(w, l)$ with domain $K$ known as the layer-$i$ meta-gate obtained by $w$ and $l$:
\[
  \begin{cases}
    \mu_1(w, l): K \to \Gates{\Sigma_1}\\
    \mu_1(w, l): \vec{z} \mapsto \{ \gatewiring: \kappa_w(w)(\vec{z}), \gatefun: \operatorname{instantiate}_1(w, l, \kappa_w(w)(\vec{z}), \kappa_1(w)(\vec{z}))\}\\
    \mu_2(w, l): K \to \Gates{\Sigma_2}\\
    \mu_2(w, l): \vec{z} \mapsto \{ \gatewiring: \kappa_w(w)(\vec{z}), \gatefun: \operatorname{instantiate}_2(w, l, \kappa_w(w)(\vec{z}), \kappa_2(w, l)(\vec{z}))\}.\\
\end{cases}
\]

In this circuit, the wiring of each gate is given by $\kappa_w(w)$, and its function is given by instantiating the label given by $\kappa_i$. When iterating this process, a gate $g$ appearing at position $\vec{z}$ in some parent meta-gate $\mu_i(l,w)$, the child meta-gate $\mu_i(g)$ associated with $g$ is defined as $\mu_i(\gatewiring[g], \kappa_i(l,w))$. The functions $\operatorname{instantiate}_i$ are each injective in their last argument, ensuring this does not create any ambiguity.

Finally, these meta-gates beget the two layers of circuits $C_1$ and $C_2$ by:
\[
  \begin{cases}
    C_i(\vec{0}) = \{ \gatewiring: s_w, \gatefun: s_i \}\\
    C_i(\vec{z}) = \mu_i(C_i(\lfloor\frac{\vec{z}}{K}\rfloor)(\vec{z} \bmod K)\\
  \end{cases}.
\]

By this definition, the wirings of $C_1(\vec{z})$ and $C_2(\vec{z})$ are the same ---they are given by $C_w$, so $C_\square: \vec{z} \mapsto C_1(\vec{z}) \otimes C_2(\vec{z})$ defines a circuit with alphabet $\Sigma_1 \times \Sigma_2$ and domain $K^\infty$.

The next subsections are the description of $C_w$, followed by those of $C_1$, then $C_2$, and finally the proof of that $C_{\square}$ is self-describing.

\subsection{The wiring layer}

The seed wiring $s_w$ is represented on the left of figure \ref{fig:kappa_seed}. On the right of the same figure is its image $\kappa_w(s_z)$.

\begin{figure}
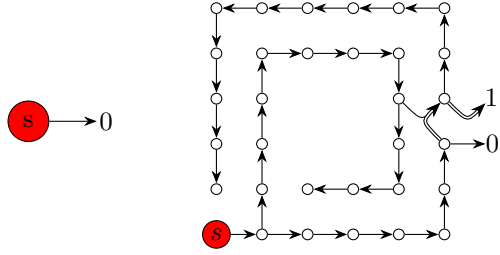

  \centering
  \begin{tikzpicture}
    \matrix[column sep=3em] {
      \input{tikz/seed_alone}&
      \input{tikz/seed_local}\\
    };
  \end{tikzpicture}
  \caption{$\kappa_1(w_s, \gcolor{seed})$, the first iteration of $\kappa_1$ on the seed gate. The wirings are given graphically; all gates have the \gcolor{normal} label, except the $s$ at $(0,0)$ which has label \gcolor{seed}. For the wiring with two inputs, the fat arrow represents the input number 0; all output wires have number 0: $\woutputmap[w](d) = 0$}
  \label{fig:kappa_seed}
\end{figure}

For any wiring $w$, $\kappa_1(w)$ will only contain wirings with at most two inputs sides, and these input sides are adjacent. Thus, all gates in $C_w$ have at most two input sides and they are adjacent. Hence, $\kappa_w$ only needs to be defined on wirings with that property.

The definition of $\kappa_w$ is isotropic: $\kappa_w$ commutes with a rotation of $\pi / 2$ and with reflections. This property will be upheld simply by giving the definition of $\kappa_w$ \emph{up to rotation and reflection}. 

Lastly, for each input wire in $w$ there are two wires in the input bus of $\kappa_w(w)$, and for each output wire in $w$ there are two wires in the output bus of $\kappa_w(w)$. The position of these wires are as follows, up to rotation:

\begin{tabular}{cc|ccc}
  \multicolumn{2}{c}{input of $w$} & \multicolumn{3}{c}{input $i$ of $\kappa_w(w)$}\\
  $\winputset[w]$ & direction & position & direction & $\winputset[i]$ \\
  \hline
  \multirow{2}{*}{$(W)$} & \multirow{2}{*}{$W$} & $(0,2)$ & $W$ & $(W)$\\
                           &                      & $(0,3)$ & $W$ & $(S, W)$\\
  \hline
  \multirow{4}{*}{$(S, W)$} & \multirow{2}{*}{$W$} & $(0,0)$ & $W$ & $(S, W)$\\
                              &                      & $(0,1)$ & $W$ & $(S, W)$\\
                              & \multirow{2}{*}{$S$} & $(0,0)$ & $S$ & $(S, W)$\\
                              &                      & $(1,0)$ & $S$ & $(S, W)$\\
  \hline
  \hline
  \multicolumn{2}{c}{output of $g$} & \multicolumn{3}{c}{output $o$ of $\kappa_w(g)$}\\
  $\woutputwires[w](d)$ & direction $d$ & position & direction $d'$ & $\woutputwires[i](d')$ \\
  \hline
  \multirow{2}{*}{$(W)$} & \multirow{2}{*}{$E$} & $(5, 2)$ & $E$            & $W$\\
                           &                      & $(5, 3)$ & $E$            & $(S, W)$\\
  \hline
  \multirow{2}{*}{$(S, W)$} & \multirow{2}{*}{$E$} & $(5, 0)$ & $E$            & $(S, W)$\\
                              &                      & $(5, 1)$ & $E$            & $(S, W)$\\
  \hline
  \multirow{2}{*}{$(S, W)$} & \multirow{2}{*}{$N$} & $(0, 5)$ & $N$            & $(S, W)$\\
                              &                      & $(1, 5)$ & $N$            & $(S, W)$\\
  \hline
\end{tabular}

Thus, if $w$ has $\winputset[w] = (W)$, $\woutputset[w] = \{N, E, S\}$ with $\woutputwires[w](N) = Sw$, $\woutputwires[w](E) = sW$ and $\woutputwires[w](S) = N$, then $\kappa_w(w)$ has the corresponding input and output busses and wires:
\[
  \begin{cases}
    \vec{I} = (\wirepos{(-1,2) \to (0, 2)}{W}, \wirepos{(-1,3) \to (0, 3)}{Ws})\\
    \vec{O} = (\wirepos{(0, 5) \to (0, 6)}{Sw}, \wirepos{(1, 5) \to (0, 6)}{Sw}, \wirepos{(5, 0) \to (6,0)}{Ws}, \wirepos{(5, 1) \to (6, 1)}{Ws}, \wirepos{(2,0) \to (2, -1)}{N}, \wirepos{(3,0) \to (3, -1)}{Nw})\\
  \end{cases}
\]

\begin{figure}
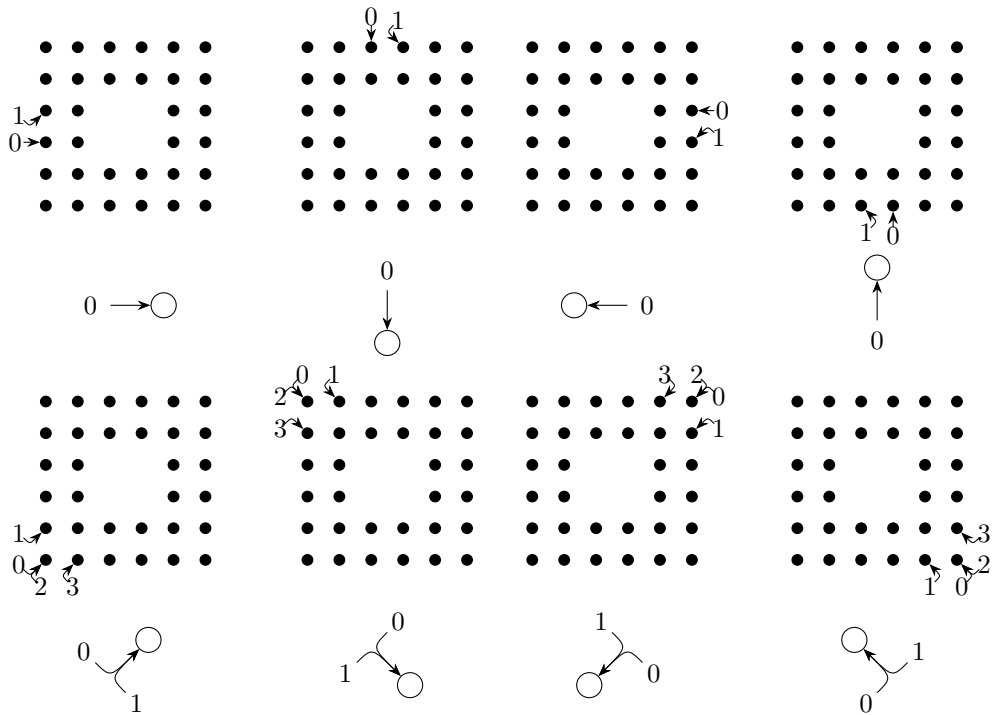

  \centering
  \begin{tikzpicture}
    \matrix[column sep=2em, row sep=0] {
      \scoped[rotate=180, scale=.7]{\input{tikz/in_wires_E}} &%
      \scoped[rotate=90, scale=.7]{\input{tikz/in_wires_E}} &%
      \scoped[rotate=0, scale=.7]{\input{tikz/in_wires_E}} &%
      \scoped[rotate=-90, scale=.7]{\input{tikz/in_wires_E}} &\\
      \scoped [rotate=180]{\input{tikz/empty_gate_E}}; &
      \scoped [rotate=90]{\input{tikz/empty_gate_E}};  &%
      \scoped [rotate=0]{\input{tikz/empty_gate_E}};   &
      \scoped [rotate=-90]{\input{tikz/empty_gate_E}}; &\\
      \scoped[rotate=180, scale = .7]{\input{tikz/in_wires_NE}} &
      \scoped[rotate=90, scale=.7]{\input{tikz/in_wires_NE}} &
      \scoped[rotate=0, scale=.7]{\input{tikz/in_wires_NE}} &
      \scoped[rotate=-90, scale=.7]{\input{tikz/in_wires_NE}} &\\
      \scoped [rotate=180]{\input{tikz/empty_gate_NE}};  &
      \scoped [rotate=90]{\input{tikz/empty_gate_NE}};   &
      \scoped [rotate=0]{\input{tikz/empty_gate_NE}};    &
      \scoped [rotate=-90]{\input{tikz/empty_gate_NE}};  &\\      
    };
  \end{tikzpicture}
  \caption{Location of the input wires in $\kappa_1(g)$ according to the input sides of $w = \gatewiring{g}$. The output wires mirror the input sides of the neighboring tiles.}

  \label{fig:kappa_wires}
\end{figure}

\paragraph*{Wirings with no inputs}

All wirings without input are sent by $\kappa_w$ to the array of wirings represented on figure~\ref{fig:kappa_seed}. The seed wiring $s_w$ is the only one to effectively appear when iterating $\kappa_w$.

\paragraph*{Wirings with one input}

A wiring with one input is cut according to figure~\ref{fig:kappa_normal1}.

\begin{figure}
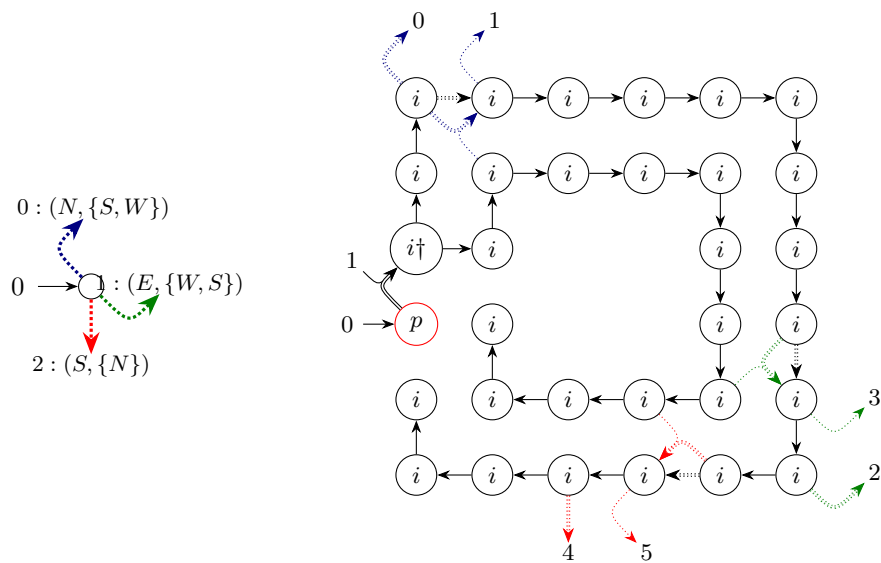

  \centering
  \begin{tikzpicture}
    \matrix[column sep=3em, row sep=5mm, cells={anchor=center}] {
      \scoped [rotate=180]{\input{tikz/empty_gate_E_to_N_nW_Se}};&
      \scoped [rotate=180]{\input{tikz/local_normal1}};\\
    };
  \end{tikzpicture}
  \caption{$\kappa_1(w, l)$ when $\winputset[w] = (W)$. The wirings are given graphically; on layer 1, all gates have the \gcolor{normal} label (black), except for the gate at $(0, 3)$ which has label \gcolor{input} (red). Dotted segments are parts of the wiring only if the corresponding element is in the outputs of $w$. For wirings with two inputs, the fat arrow represents the first input. The $\dagger$ marks the potential position for label $\gcolor{dA}$ on layer 2.}
  \label{fig:kappa_normal1}
\end{figure}

The wirings in $\kappa_w(g)$ depend on the input and outputs of $w$. Since $w$ has only one input, it is sufficient, up to rotation, to examine the case where $\winputset[w] = \{W\}$. The black arrows of figure~\ref{fig:kappa_normal1} depict the case where $\woutputset[w] = \emptyset$. For each $d \in \woutputset[w]$, some extra wires are needed. The additional wires for an output in the $N$ direction depend on whether the directions appearing in $\woutputwires[w](N)$ are $\{S\}$, $\{S, E\}$ or $\{S, W\}$. The wires for the other output directions are derived from these by rotation. Figure~\ref{fig:kappa_normal1} has one direction with each case, showing the complete range of possibilities for the extra wires; they are represented in dotted lines on the figure.

\paragraph*{Wirings with two (adjacent) inputs} The image of such a wiring $w$ by $\kappa_w$ is defined on figure~\ref{fig:kappa_input}, which is rotated and reflected so that the $S$ side of the figure is mapped to the $0$ input of $w$, and the $W$ side to the $1$ input of $w$.

\begin{figure}
  \centering
    \begin{tikzpicture}
    \matrix[column sep=3em, row sep=5mm, cells={anchor=center}] {
      \input{tikz/empty_gate_SW_input2}&
      \input{tikz/local_input2};\\
    };
  \end{tikzpicture}
  \caption{$\kappa_1(w, l)$ when $\winputset{w} = (S, W)$. The wirings are given graphically; all gates have the \gcolor{normal} label except for the input gate at $(0, 0)$ which has label $\gcolor{input}$. This configuration gets rotated and reflected according to the inputs of the gate, so that the arrows on the outer ring form a path from the side with input $0$ to the side with input $1$, ``the long way around''. The '*' marks the \emph{special} position on layer 2, and the $\dagger$ marks the position where label $\gcolor{dA}$ may appear.}
  \label{fig:kappa_input}
\end{figure}

\begin{lemma}
  \label{lem:wiring-normal}
    For each gate $w$, $\kappa_w(w)$ is the wiring of a normal circuit. Moreover, $C_w$ is the wiring of a normal, closed circuit.
\end{lemma}

\begin{proof}
  This follows from observation of the schemata describing each $\kappa_w(w)$ and observing that for two wirings $w, w'$, if an output side in direction $d$ of $w$ matches an input side in direction $-d$ of $w'$, then the corresponding sides of $\kappa_w(w)$ and $\kappa_w(w')$ also match.

  Additionally, each iteration of $\kappa_w$ on $s_w$ is without inputs, so $C_w$ is closed.
\end{proof}

\subsection{Layer 1}
\label{sec:local_layer}

The components of layer 1 are an alphabet $\Sigma_1$, a finite set of labels $L_1$, a substitution $\kappa_1: \Wirings \mapsto L_1^K$, and for each label $l \in L_1$, and an instantiation function $\operatorname{instantiate}_1$. Layer 1 can then be realized as described above into a circuit with the wirings of $C_w$ and the functions given by instantiating the fixpoint of $\kappa_1$.

The set of labels is $L_1 = \{\gcolor{s}_1, \gcolor{normal}, \gcolor{input} \}$.

The elements of $\Sigma_1$ are \emph{layer 1 messages}.

\begin{definition}[Layer 1 Message]
  A layer 1 message is a triple $\{\lmpos = \vec{z} \in K_1, \lmpwiring = w \in \Wirings, \lmpar = l \in L_1\}$. \end{definition}

The circuit $C_1$ is going to be defined by iterating $\kappa_w$ and $\kappa_1$ starting from $s_w$ and $\gcolor{s}_1$, then instantiating the labels.

\paragraph*{Labels}

The label function $\kappa_1: \Wirings \to L_1^K$ assigns the label $\gcolor{normal}$ at all positions except:
\begin{itemize}
\item $\kappa_1(s_w)(\vec{0}) = \gcolor{s}_1$
\item if $w$ has at least one input, then $\kappa_1(w)(\vec{z}) = \gcolor{input}$ for the position $\vec{z}$ which receives the input wire $0$ in $\kappa_w(w)$.
\end{itemize}

\paragraph*{The seed function}

The starting label is $\gcolor{s}_1$, it only ever appears at position $(0,0)$ in $\kappa_1(s_w, \gcolor{s}_1)$.

Since $s_1$ has no inputs, its associated function $f_s = \finst{1}(s_w, \gcolor{s}_1, s_w, \gcolor{s}_1)$ is a constant function with value $f_{s}() = \{\lmpos: (0, 0), \lmpcolor: \gcolor{s}_1, \lmpwiring: s_w\}$.

\paragraph*{Gate functions for layer 1}

There are two types of functions for the gates output by $\kappa_1(g)$ other than the seed. They have a either an \emph{increment} function $\fincr{k,D}$ with $k \in \{1, 2\}$ and $D \in \Dir$ or a \emph{reparenting} function $\frepar{k, D, w, c}$, with $k \in \{1, 2\}$, $D \in \Dir$, $w$ a wiring and $c \in C$ a color. The versions with $k=2$ take two inputs but ignore the second one. The functions $\fincr{}$ and $\frepar{}$ are defined as follows:

\begin{eqnarray*}
  \lmpar[\fincr{1, D}(m)] &=& \lmpar[m],\\
  \lmpos[\fincr{1, D}(m)] &=& \lmpos{m} + D\\
  \lmpar[\frepar{1, D, w, c}(m)] &=& (w, c)\\
  \lmpos[\frepar{1, D, w, c}(m)] &=& \lmpos[m] + D\\
  \fincr{2, D}(m, m') &=& \fincr{1, D}(m)\\
  \frepar{2, D, w, c}(m, m') &=& \frepar{1, D, w, c}(m).\\
\end{eqnarray*}

The function of each gate is fixed from its label and wiring, and those of its parent through \finst{1} as follows:
\[
  \begin{cases}
    \finst{1}(w_p, l_p, w, \gcolor{normal}) = \fincr{k, D} \\
    \finst{1}(w_p, l_p, w, \gcolor{input}) = \frepar{k, D, w_p, l_p},\\
  \end{cases}
\]
where $k$ is the number of inputs of $w$, and $D$ is the direction of its first input.

\paragraph*{Behavior and Self-Description of Layer 1}

The circuit $C_1$ obtained from $\kappa_1$ is normal, by lemma~\ref{lem:wiring-normal}. Let $e = \innerevalfunc{C_1}: K^\infty \times \Dir \to \Sigma_l$ be the evaluation function of $C_1$. That function $e$ enjoys a simple description, which reflects the fact that in $C_1$, the different meta-gates do not actually communicate. On layer 2 however, there will be some communication between meta-gates, as described in the next section.

\begin{lemma}
  \label{lem:local_eval}
  Let $w \in \Wirings$, $l \in L_1$ and $a: p \mapsto p'$ be an internal or outgoing arc in $\kappa_w(w)$. If $w$ has no inputs, assume $l = \gcolor{s}_1$.

  For any input $\vec{\imath}$ of $\mu_1(w, l)$, let $m = \innerevalfunc{\mu_1(w,l)}(\vec{\imath}, a)$ be the value of arc $a$ on input $\vec{\imath}$. Then $\lmpos[m] = p$, $\lmpwiring[m] = w$ and $\lmpcolor[m] = l$.
\end{lemma}
\begin{proof}
  
  By induction on the non-incoming arcs of the (acyclic) dependency graph $D$ of $\kappa_w(w)$.

  If $w$ is the seed wiring $s_w$, then the root of $D$ is also $s_1$ and its outputs satisfy $\lmpos[e_g(a)] = (0, 0)$, $\lmpwiring[e_g(a)] = \gatewiring[s_1]$ and $\lmpcolor[e_g(a)] = \gcolor{s}_1$.

  Otherwise, $\kappa_1(w)$ has a unique position $\vec{z_i}$ with label $\gcolor{input}$, corresponding to a gate in $\mu_1(l,w)$ with function $\frepar{k, D, w, l}$, for some $k$ and $D$. By definition of $\frepar{}$, its outputs satisfy $\lmpos[e_g(a)] = \vec{z_i}$ , $\lmpwiring[e_g(a)] = w$ and $\lmpcolor[e_g(a)] = l$.

  In both cases, each other gate $g'$ at position $\vec{z}$ of $\mu_1(l,w)$ has function $\fincr{k, D}$, where $k$ is the number of inputs of $g'$, and $D$ is the direction of its first input arc $i = \vec{z} - D \overset{D}{\mapsto} z$. By induction, $\lmpos[e_g(i)] = z - D$ and $\lmpar[e_g(i)] = (\gatewiring[g], \gatecolor[g])$. By definition of $\fincr{k, D}$, each of the output arcs $o$ of $g'$ verify $\lmpos[e_g(o)] = (\vec{z} - D) + D = \vec{z}$ and $\lmpwiring[e_g(a)] = w$ and $\lmpcolor[e_g(a)] = l$.
\end{proof}

Additionally, $C_1$ is ``mostly self-describing'': the first input of a gate $g$ is enough to recover $g$, except for the value of $w$ and $c$ in gates with a function of the form $\frepar{k, D, w, c}$.

\begin{lemma}
  \label{lem:layer1-self-descr}
  For a position  $p \in K^{\infty}$, let $ \vec{e} = \winputset{\gatewiring[C_1(p)]}$ be the input arcs of $C_1(p)$; let $d_i$ be the direction of $e_i$.
  
  There is a function \(\decgate: \Sigma_1 \times \Dir \to \Gates{\Sigma_1} \cup \bot\) such that for all $p \in K^{\infty}$,
  \begin{itemize}
  \item if \(\gatefun[C_1(p)]\) is $\fincr{k, d_0}$, then \(\decgate(\innerevalfunc{C_1}(e_0), d_0) = C_1(p)\)
  \item if \(\gatefun[C_1(p)]\) is $\frepar{k, d_0, -, -}$ then \(\decgate(\innerevalfunc{C_1}(e_0), d_0) = \bot\).
  \end{itemize}
\end{lemma}

\begin{proof}
  The function $\decgate$ is defined as follows: let $m \in \Sigma_1$ be a local message, and $d \in \Dir$. Let $p = \lmpos[m]$, if $p - d \notin \{0, \ldots, 5\}^2$, then $\decgate(m, d) = \bot$. Otherwise, $\decgate(m, d)$ is the gate at position $p$ in $\mu_1(\lmpwiring[m], \lmpar[m])$.
  
  By lemma~\ref{lem:local_eval}, $\decgate$ satisfies the lemma.
\end{proof}

Moreover, when $\decgate(e(e_O), d_0) = \bot$, $g$ itself cannot be determined, but its label and wiring can, as well as $\mu_1(g)$.

\begin{corollary}
  \label{thm:layer1-decode}
    For a position  $p \in K^{\infty}$, let $ \vec{e} = \winputset[\gatewiring[C_1(p)]]$ be the input arcs of $C_1(p)$; let $d_i$ be the direction of $e_i$,
  
    \begin{itemize}
    \item there is a function \(\decgatew: \Sigma_1 \times \Dir \to \Wirings\) such that for each position $p \in K^{\infty}$, $\decgatew(\innerevalfunc{C_1}(e_0), d_0) = \gatewiring[C_1(p)]$
    \item there is a function \(\decgatel: \Sigma_1 \times \Dir \to \Wirings\) such that for each position $p \in K^{\infty}$, $\decgatel(\innerevalfunc{C_1}(e_0), d_0) = \Lambda_1(p)$
\end{itemize}
\end{corollary}

\begin{proof}

  For \decgatew{} and \decgatel{}, it suffices to observe that whenever $\decgate(\innerevalfunc{C_1}(e_0), d_0) = \bot$ at some position $p$, $\innerevalfunc{C_1}(p)$ has label \gcolor{input}, and its wiring only depends on its position within its metagate.
\end{proof}

\subsection{Layer 2}

A second layer is needed in order to get full self-description of the circuit on $K^\infty$. This layer routes global information between the meta-gates so that the input gate of each meta-gate can be indentified by its incoming global message. The construction needs to ``tie the knot'', so it is not only needed to recover the identity of the input gate on the layer 1, but also on layer 2 itself.

This layer is defined by a set $L_2$ of labels, an alphabet $\Sigma_2$, the label substitution $\kappa_2$ and its instantiation function. In contrast with layer 1, $\kappa_2$ takes as input a wiring, as well as a label. This makes the definition of $\Lambda_2$ recursive. In contrast with layer 1, on layer 2, the two-step definition of gates using labels and the function $\operatorname{instantiate}$ is actually \emph{needed} because of a subtelty related to the tying of the knot. The functions of some gates make use of $\kappa_2$. Thus $\kappa_2$ shall not directly manipulate the gates or their function, lest the definition of layer 2 becomes cyclical and possibly ill-founded. 

The set of layer 2 labels is $L_2 = \{\gcolor{s}, \gcolor{i_1}, \gcolor{i_2}, \gcolor{dL}, \gcolor{dA} \}$.

\paragraph*{Layer 2 messages}

\tikzexternaldisable
A message on layer 2 is an element of $\Sigma_2$; it is built from:
\begin{itemize}
\item $\mloc[m] \in L_2$\tikzmark{local},
\item $\mglob[m]$\tikzmark{global}, itself consisting of:
  \begin{itemize}
  \item $\gmcolor[ {\mglob[m]} ] \in L_2$.\tikzmark{anc-layer2}
  \item $\gmanclocal[ {\mglob[m]} ] \in \Sigma_1$\tikzmark{anc-layer1}
  \end{itemize}
\end{itemize}

These messages make $C_2$ self-describing by completing the information available in layer 1 and used in lemma~\ref{lem:layer1-self-descr}. A message $m_2$ output by a gate $g_2$ in $\kappa_2(p)$ \surligne[text-local-part]{purple!30!white}{identifies $p$} through $\mloc[m_2]$ if $g_2$ is not the input gate of $\kappa_2(p)$, as in lemma~\ref{lem:layer1-self-descr}. The global part $\mglob[m_2]$ \surligne[text-global-part]{yellow}{identifies \emph{some ancestor} of $g_2$}, on layer 1 through $\gmanclocal[ {\mglob[m]} ]$ and on layer 2 through $\gmcolor[ {\mglob[m]} ]$. This global information will enable the determination of the entry gate of each meta-gate, on both layers.
\tikz[remember picture] \draw[overlay, purple!10!white, thick] (pic cs:local) to[bend left=15] (text-local-part);
\tikz[remember picture] \draw[overlay, yellow!50!white, thick] (pic cs:global) -- (text-global-part);
\tikzexternalenable
From two messages $m_l, m_g \in \Sigma_2$, two messages $(m_1, m_2) = \extractfunc(m_l,m_g) \in \Sigma_1 \times \Sigma_2$ can be extracted by reading the local information from $m_l$, and the global information from $m_g$, as follows:
\[
  \begin{cases}
    m_1 = \gmanclocal[ {\mglob[m_l]} ]\\
    \mloc[m_2] = \gmcolor[ {\mglob[m_l]} ]\\
    \mglob[m_2] = \mglob[m_g].\\
  \end{cases}
\]

The converse operation, embedding, takes as input three messages, a payload $(p_1, p_2) \in (\Sigma_1 \times \Sigma_2)$ and a context $c \in \Sigma_2$, and yields two messages $m_l$ and $m_g$
\[
  \begin{cases}
    \mloc[m_l] = \mloc[m_g] = \mloc[c]\\
    \gmcolor[ {\mglob[m_l]} ] = \mloc[p_2].\\
    \gmanclocal[ {\mglob[m_l]} ] = p_1\\
    \mglob[m_g] = \mglob[p_2].\\
  \end{cases}
\]

Extracting and embedding are dual operations, in the sense that  \[\forall c, p_1, p_2, \extractfunc \circ \embedfunc(c, p_1, p_2) = (p_1, p_2).\]

\paragraph*{The substitution $\kappa_2$}

Given a wiring $w$ and a label $l \in L_2$, the substitution $\kappa_2$ yields a label for each position in $K$;
\begin{itemize}
\item $\gcolor{s}_2$ for the seed gate, i.e. for position $(0,0)$ if $w$ has no inputs;
\item $\gcolor{dL}$ for the gate receiving the input number 0 of the meta-gate;
\item $\gcolor{dA}$ for the gate receiving the input number 1 of the meta-gate whenever $l = \gcolor{dL}$;
\item when $w$ has two inputs, there is a \emph{special} position within the meta-gate where the value depends on $l$ as follows:
  \begin{itemize}
  \item $\gcolor{i2}$ if $l \in \{\gcolor{i2}, \gcolor{dL}\}$,
  \item $\gcolor{dA}$ if $l = \gcolor{dA}$;
  \end{itemize}
\item otherwise, $\gcolor{i_1}$ for gates with one input and $\gcolor{i_2}$ for gates with two inputs.
\end{itemize}

The special position is marked on figures~\ref{fig:kappa_normal1} and~\ref{fig:kappa_input} by a star.

\paragraph*{The layer-2 seed gate}

The function $f^2_s = \finst{2}(s_w, \gcolor{s}_2, s_w, \gcolor{s}_2)$ of the seed gate is a constant function with value $f_{s}^2() = \{\mloc: \gcolor{s_2}, \mglob: \{\gmcolor: \gcolor{s_2}, \gmanclocal: f^1_s() \} \}$. Recall that $f^1_s()$ is the message output by the seed gate on layer 1.

\paragraph*{Gate functions for layer 2}

The functions of the gates depend on their label and on the direction $D$ of their first input, as dictated by the function $\operatorname{instantiate}$:
\[
  \begin{cases}
    \operatorname{instantiate}(D, \gcolor{s}) = f^2_s\\
    \operatorname{instantiate}(D, \gcolor{i1}): x \mapsto x\\
    \operatorname{instantiate}(D, \gcolor{i2}): (x, y) \mapsto (x, y)\\
    \operatorname{instantiate}(D, \gcolor{dL}) =  \fdecode[D]{L} \circ \fgincr{G}{D}\\
    \operatorname{instantiate}(D, \gcolor{dA}) = (m_l, m_g) \mapsto \fdecode[D]{A}(\fgincr{G}{D}(m_l), m_g),\\
  \end{cases}
\]
where $D$ is the direction of the first input of $w_p$.

The gate with labels \gcolor{i1} or \gcolor{i2} are wires; their function is the identity function of arity $1$ or $2$ respectively.

Given $m \in \Sigma_2$, the increment functions $\fgincr{G}{D}$ increments $\lmpos[\gmanclocal{\mglob{m}}]$ by the unit vector of direction $D$.

The two decoding functions $\fdecode{L}$ and $\fdecode{A}$ are based on the function $\fdecode{}: \Dir \times \Sigma_1 \times \Sigma_2 \to \Sigma_1 \times \Sigma_2$ defined as follows: let $m_1 \in \Sigma_1, m_2 \in \Sigma_2$, pose
$z = \lmpos[m_1]$, $z_a = \lmpos[ {\gmanclocal[ {\mglob[m_2]} ]} ]$, $a = \gmanclocal[ {\mglob[m_2]} ] $. Let $l_p = \kappa_2(\lmpwiring[a], \lmpcolor[a], \gmcolor[ {\mglob[m_2]}], p')$ if $z \in K_1$, and $\gcolor{dL}$ otherwise. Then $\fdecode[D]{}(m_1, m_2)$ is the pair $(m'_1, m'_2)$ with:
\begin{eqnarray*}
  m'_1 &=& \begin{Bmatrix}
    \lmpos&:& z \operatorname{mod} K\\
    \lmpwiring&:& \decgatew(a, D)\\
    \lmpcolor&:& \decgatel(a, D)\}
  \end{Bmatrix}\\
  m'_2 &=& \begin{Bmatrix}
    \mloc &:& p\\
    \gmcolor[ { \mglob } ] &: &\gmcolor[\mglob{m}]\\
    \gmanclocal[ { \mglobb } ] &: &
      \begin{Bmatrix}
        \lmpwiring &:& \lmpwiring[ { \gmanclocal[ { \mglob[m] } ] } ]\\
        \lmpcolor &:& \lmpcolor[ { \gmanclocal[ { \mglob[m] } ] } ]\\
        \lmpos &:& z_a \bmod K_1\\
      \end{Bmatrix}
    \end{Bmatrix}
\end{eqnarray*}



  

For a direction $D$ and a message $m_2 \in \Sigma_2$, take an arbitrary $m_1 \in \Sigma_1$ and let $(m'_1, m'_2) = \fdecode[D]{}(m_1, m_2)$; the value of $m'_2$ does not depend on $m_1$, so $\fdecode[D]{L}(m_2)$ is defined to be the $m'_2 \in \Sigma$ returned by $\fdecode{}(m_1, m_2)$ for any $m_1$.






For a direction $D$ and $m_l, m_g \in \Sigma_2$, $\fdecode[D]{A}(m_l, m_g)$ is defined as follows. Let $(m_1, m_2) = \extractfunc(m_l,m_g)$, $(m'_1, m'_2) = \fdecode[D]{}(m_1, m_2)$. Then $\fdecode[D]{A}(m_l, m_g)$ is the pair $m'_l, m'_g$ with:
\[
  \begin{cases}
    \mloc[m'_l] = \mloc[m'_g] = \mloc[m_l]\\
    \gmanclocal[ { \mglob[m'_l] } ] = m'_1\\
    \gmcolor[ { \mglob[m'_l] } ] = \mloc[m'_2]\\
    \mglob{m'_g} = \mglob[m']\\
  \end{cases}
\]

The function $\fdecode{A}$ is engineered in order to enjoy the following property, a kind of commutation between $\fdecode{}$ and $\extractfunc$.

\begin{lemma}
  \label{prop:decodeL_incr}
  For any direction $D \in \Dir$,
  \[\fdecode[D]{} \circ \extractfunc = \extractfunc \circ \fdecode[D]{A} \]
\end{lemma}

\begin{proof}
  By computation.
\end{proof}

\paragraph*{Properties of $C_2$}

The messages passing through the circuit $C_2$ built above hold all the necessary information for $C_\square$ to be self-describing: in other words, $C_2$ carries all the information needed to determine the entry gate of each meta-gate in $C_1$, as well as the information needed to determine each of its gates.

The $\gmanclocal[\mglob]$ part of the messages on input $0$ each meta-gate simulate the gates of layer 1, as long as each meta-gate simulating a gate with label $\gcolor{input}$ receives on input $1$ the message of its parent meta-gate.

\begin{lemma}
  \label{lem:eval_layer1}

  let $w \in \Wirings$ with $k$ inputs, $l_2 \in L_2$, and $M$ be the circuit $\mu_2(w, l_2)$. Note that $M$ has $2k$ inputs. Let $D$ be the direction of the first input of $w$.

  Let $\vec{\imath} \in \Sigma^{2k}$, and let $\vec{o}$ be the output of $M$ on input $\vec{\imath}$. Pose $i_a = \gmanclocal[\mglob[i_0]]$ and $i_p = \gmanclocal[\mglob[i_1]]$.

  Then, if $g_1 = \decgate(i_a, D) \neq \bot$, then $o_0$ is the output of $g_1$ on input $i_a$; else, if $\decgate(i_a, D) = \bot$, then $o_0$ is the output of $\decgate_c(i_p, D)(\lmpos[i_a])$ on input $i_a$.
\end{lemma}

\begin{proof}
  By computation.
\end{proof}

The rest of the $\mglob$ part of the messages allows $C_2$ to simulate itself.

\begin{lemma}
  \label{lem:eval_subst1}

  Let $D \in \Dir$, $w \in \Wirings$ with one input in direction $D$, $l_2 \in L_2$. Let $f = \operatorname{instantiate}(D, l_2)$ and $C = \mu_2(w, l_2)$.

  Let $c \in \Sigma_1, i_1 \in \Sigma_1, i_2 \in \Sigma_2$ and $(m_l, m_g) = \embedfunc(c, i_1, i_2)$. Let $(o_0, o_1)$ be the outputs of $C$ on input $(m_l, m_g)$. Then $\extractfunc(o_0, o_1) = f(i_2)$.
\end{lemma}

\begin{proof}
  The proof proceeds by case on $l_2$, which can be either $\gcolor{i_1}$ or $\gcolor{dL}$.
  If $l_2 = \gcolor{i_1}$, then $f$ is the identity function, and it suffices to follow the wirings to check the result.

  If $l_2 = \gcolor{dL}$, then following the wirings reduces the desired equality to the definition of $\fdecode[D]{L}$.

\end{proof}

\begin{lemma}
  \label{lem:eval_subst2}
  Let $w \in \Wirings$ with two inputs, $l_2 \in L_2$. Let $f = \operatorname{instantiate}(l_2)$ and $C = \mu_2(w, l_2)$.

  For $k \in \{0, 1\}$, let $c^k \in \Sigma_1, i_1^k \in \Sigma_1, i_2^k \in \Sigma_2$ and $(m_l^k, m_g) = \embedfunc(c^k, i_1^k, i_2^k)$. Let $(o_0^0, o_1^0, o_0^1, o_1^1)$ be the outputs of $C$ on input $(m_l^0, m_g^0, m_l^1, m_g^1)$. Then for $k \in \{0, 1\} \extractfunc(o_0^k, o_1^k)$ is the $k$-th component of $f(i_2^0, i_2^1)$.
\end{lemma}

\begin{proof}
  The proof proceeds by case on $l_2$. If $l_2 \in \{\gcolor{s}, \gcolor{i2}, \gcolor{dA} \}$, following the wirings in $\kappa_2(w, l_1, l_2)$ confirms that the lemma holds.

  If $l_2 = \gcolor{dL}$, the lemma follows from lemma~\ref{prop:decodeL_incr} by again following the wirings.
\end{proof}

Together, these properties entail a substitutive structure of the messages in $C_\square$. Going up the hierarchy, the message between two gates of $C_\square$ can be extracted from the messages between the corresponding meta-gates.

\begin{lemma}
  \label{lem:eval_mg}
  Let $a = p \to p'$ be a wire between two positions $p, p'$ of $K^{\infty}$ in direction $D$. Let $a_l, a_g$ be the two wires crossing the edges between $p K$ and $p' K$ in clockwise order looking in direction $D$ (i.e., if $D$ is $E$, $a_l$ is the northernmost of the two; if $D$ is $S$, the westernmost…).

  Let $m = \innerevalfunc{C_\square}(a)$, $l = (l_1, l_2) = \innerevalfunc{C_\square}(a_l)$ and $g = (g_1, g_2) = \innerevalfunc{C_\square}(a_g)$.

  Then $\extractfunc(l_2, g_2) = m$.
\end{lemma}

\begin{proof}
  By induction on $p$, following the wires of $C_\square$.
\end{proof}


\begin{lemma}
  \label{lem:info-loc}
  Let $\vec{p} \in K^\infty$. Let $g_1 = C_1(\vec{p})$, $g_2 = C_2(\vec{p})$, $g'_1 = C_1(\lfloor \frac{\vec{p}}{K} \rfloor)$, $g'_2 = C_2(\lfloor \frac{\vec{p}}{K} \rfloor)$. Let $g = g_1 \otimes g_2 = C_\square(\vec{p})$, and $g' = g'_1 \otimes g'_2 = C_\square(\lfloor \frac{\vec{p}}{K} \rfloor)$. Let $e$ be an output wire of $g$, and $m_1 \times m_2 = \innerevalfunc{C_\square}(e)$ its value in $C_\square$.

  Then $\mloc[m_2]$ is the label of $g'_2$, $\lmpcolor[m_1]$ is the label of $g'_1$, $\lmpwiring[m_1]$ is $\gatewiring[g']$, and $\lmpos[m_1]$ is $\vec{p} \bmod K$.
\end{lemma}

\begin{proof}
  By the previous lemma, input $0$ of each meta-gate $\mu_1(l_1, w) \otimes \mu_2(l_2, w)$ encodes $l_1, l_2$ and $w$. The gate after that input has label $\gcolor{dL}$, so by definition of its function $\fdecode{L}$, its output satisfies the lemma. The other gates in the meta-gate preserve the local part of the messages.
\end{proof}

This local information is just what is needed to reconstruct each gate from its \emph{output}, which is just short of self-description.

\begin{corollary}
  There are is a function $\decgate': \Sigma_1 \times \Sigma_2 \to \Wirings \times L_1 \times L_2 \times K$ such that for any position $p \in K^{\infty}$ and any wire $a: p \to p'$ of $C_\square$, \[\decgate'(\innerevalfunc{C_\square}(a)) = (\Lambda_1(\quot{p}{K}), \Lambda_2(\quot{p}{K}), C_w(\quot{p}{K}), p \bmod K).\]
\end{corollary}

With a tiny bit of extra work, each gate $g$ can be reconstructed from its input number 0, making $C_\square$ self-descriptive.

\begin{theorem}
  The circuit $C_\square$ is self-descriptive.
\end{theorem}

\begin{proof}
  Each gate $g$ in position $p$ can be reconstructed from its set of input directions and the message on its input number $0$.

  If $g$ has no inputs, then $g = s_1 \otimes s_2$.

  Otherwise, let $D$ be the direction of its first input wire, and $m = (m_1, m_2)$ the message coming into its first wire.

  Let $(w, l_1, l_2, p') = \decgate'(m)$. Note that $p' = (p - D) \bmod K$. If $p' + D$ belongs to $K$, then $\bigquot{p}{K} = \bigquot{p - D}{K}$, so $p$ belongs to the same meta-gate as its predecessor $p - D$, hence $C_\square(p)$ is $\mu_1(w, l_1)(p' + D) \otimes \mu_2(w, l_2)(p' + D)$.

  Otherwise, since $p' + D \notin K$, it must be the case that $\bigquot{p}{K} = \bigquot{p-D}{K} + D$: $p$ and its predecessor $p - D$ are in neighboring meta-gates. Let $a = \bigquot{p}{K}$ be the position of the parent gate. The position $a \bmod k$ of $a$ within its meta-gate is the position where input $0$ enters that meta-gate.

  A close examination of the values of $\mu_1(w, l_1)$ and $\mu_2(w, l_2)$ for all $w, l_1, l_2$ reveals that in each (non-seed) meta-gate the label and wiring of the gate in the position where input $0$ comes only depends on the input directions of $w$. By lemma~\ref{lem:info-loc}, those directions are exactly those of $\woutputwires[ { \lmpwiring[ { \gmanclocal[ { \mglob[m_2] } ] } ] } ](D)$.
\end{proof}

This concludes the proof of \cref{thm:main}. By theorem~\ref{thm:circuit-to-atam}, this means that there is an aTAM system $S_\square$ which strictly self-assembles $K^\infty$.

\section{A Characterization of Admissible Generators}
\label{sec:limits}

The previous construction can be adaptated to any self-similar discrete fractal within which the communication pattern used by $C_\square$ can be embedded. Whether a particular fractal is amenable to hosting such a communication pattern depends on the ability of its generator $G$ to transport information to copies of itself around it.

\begin{definition}
  Let $G$ be a finite subset of $\mathbb{N}^2$, the grid $G^{\#}$ is the subset of $\mathbb{Z}^2$ defined by:
  \[ G^{\#} = \{\ p \in \mathbb{Z}^2 | p \bmod G \in G \} \]


  The grid neighborhood graph $G^{+}$ of $G$ is the subgraph of $G^{\#}$ induced by the distance $1$ neighborhood of $G$:
  \[ G^+ = G \cup \{ p \in G^{\#} | \exists d \in \Dir, p + d \in G \}. \]

  For a direction $d \in \Dir$, the $d$-port in $G^{+}$ is \( G^{+d} = \{ p \in G^{\#} | p - d \in G \}. \)

  For $d, d' \in \Dir$, the $(d, d')$-bandwidth of $G$ is the number \( G[d \leftrightarrow d'] \) of vertex-disjoint paths from $G^{+d}$ to $G^{+d'}$ in $G^+$.
\end{definition}

\begin{figure}
  \centering
  \includegraphics[width=\textwidth]{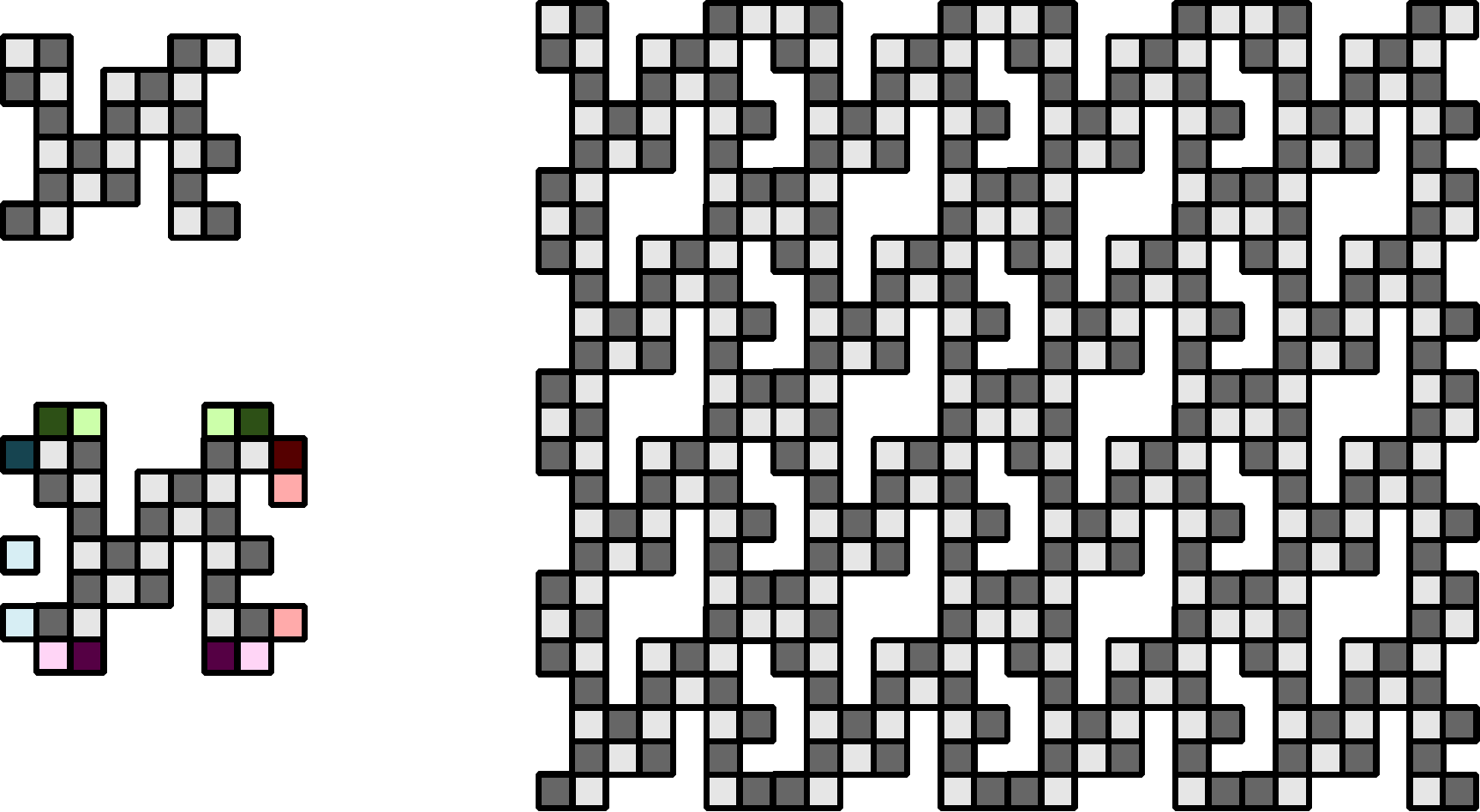}
  \caption{A finite shape $G$, its grid neighborhood $G^{+}$ and its grid graph $G^{\#}$.}
  \label{fig:grid_neigh}
\end{figure}

In order to compare generators, the classical notions of \emph{subgraph} and \emph{graph subdivision} is useful, accounting for marked vertices.

\begin{definition}[Pointed subgraph]
  Let $G, H$ be two graph, each with marked vertices. The graph $H$ is a \emph{pointed subgraph} of $G$ if $H$ is a subgraph of $G$ in such a way that any marked vertex of $H$ is mapped to a marked vertex of $G$. The graph $G$ may have extra marked vertices.
\end{definition}

\begin{definition}[Edge subdivision]
  Let $G = (V, E)$ be a graph, with some marked vertices $(v_0, \ldots, v_{k-1}) \in V^k$. The edge subdivision operation for an edge $e = \{u, v\} \in E$ is the deletion of $e$ from $G$ and the addition of a new vertex $w \notin V$ and of the edges $\{u,w\}$ and $\{w,v\}$.

  This operation generates a new graph H, where the same vertices are marked as in $G$.
  \[H = (V \cup \{w\}, (E \setminus \{u,v\}) \cap \{\{u,w\}, \{w,v\}\})\]
\end{definition}

\begin{definition}[Graph Subdivision]
A graph with marked vertices which has been derived from $G$ by a sequence of edge subdivision operations is called a pointed subdivision of $G$.
\end{definition}

\begin{definition}[Subconnector]
  Let $G, H$ be finite, connected shapes of $\mathbb{N}^2$, with $(0, 0) \in G \cap H$.

  Then $H$ is a subconnector of $G$, written $H \preceq G$ if $G^+$ has a (pointed) subgraph which is a (pointed) subdivision of $H^+$.
\end{definition}

The notion of subconnector is well-suited to the study of substitutions given the following properties.

\begin{remark}
  \label{lem:preceq_sigma}
  Let $G, H, I$ be  finite, connected shapes of $\mathbb{N}^2$, with $G \preceq H$, then:
  \begin{eqnarray*}
    \sigma_I(G) \preceq \sigma_I(H)\\
    \sigma_G(I) \preceq \sigma_H(I)\\
  \end{eqnarray*}
\end{remark}

\begin{lemma}
  \label{lem:build_through_subconnector}
  Let $G \ni (0,0)$ a finite, connected subgraph of $\mathbb{N}^2$ such that $K \preceq G$. Then there is a self-descriptive circuit with domain $G^\infty$.
\end{lemma}

\begin{proof}
  First, notice that $\kappa_w$ can be completed to cover the case where $w$ has two opposite input directions, as represented on figure \ref{fig:subst_w_extra} (modulo rotation and reflection).

  \tikzexternaldisable
  \begin{figure}
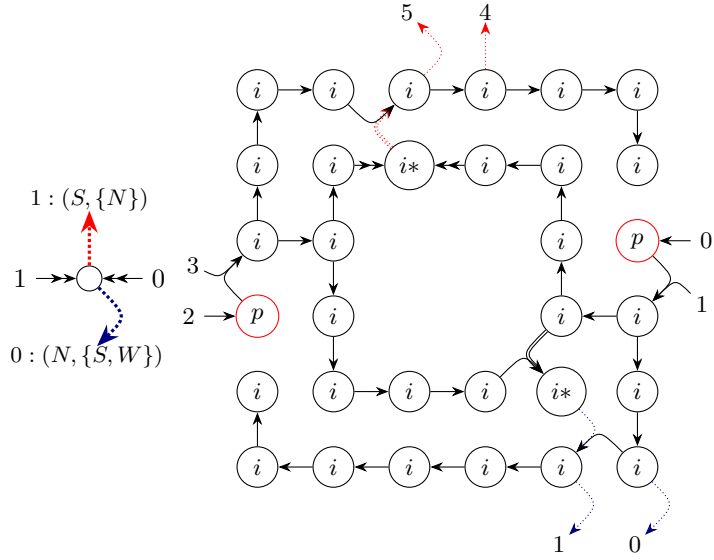

    \centering
    \begin{tikzpicture}
      \matrix {
        \begin{scope}
          \input{tikz/empty_gate_EW_to_N_Se}
        \end{scope}; &
        \begin{scope}
          \input{tikz/local_opposite}
        \end{scope}; \\
      };
    \end{tikzpicture}
    \caption{The value of $\kappa_w(w)$ when $\winputset[w] = \{E, W\}$. This case is not needed in section \ref{sec:circuit}, but it is necessary to generalize $\kappa_w$ to $G$ when $K^+$ is a subdivision of $G^+$. No position needs to be provisionned for $\gcolor{dA}$, since both inputs must be in the same metagate.}
    \label{fig:subst_w_extra}
  \end{figure}
  \tikzexternalenable

  Fix the vertices of $G^+$ which represent the vertices of $K^+$, the ones that sit in the middle of its edges, and the ones which are pending leaves.

  Then it is possible to define a substitution $\gamma_w: \Wirings \to \Wirings^G$ by using $\kappa_w$ to fix a subdivision of $G^+$, then adding the extra vertices to $\gamma_w(w)$. This process may create vertices in $\gamma_w(w)$ with two opposite inputs, for which the extra cases of figure~\ref{fig:subst_w_extra} are necessary.

  \begin{figure}
    \centering
    \includegraphics[width=.5\textwidth]{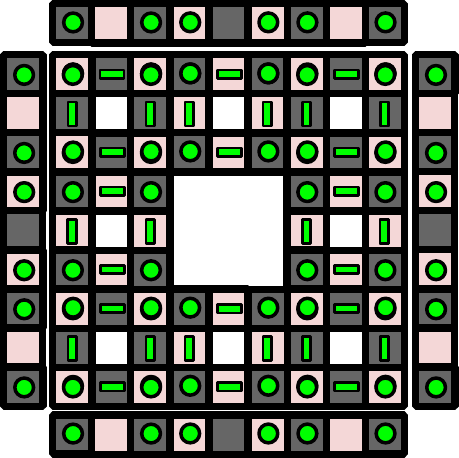}
    \caption{The second iteration $S$ of the Sierpinski Carpet (in the center) is a subconnector of $K$: $S^+$ contains a subdivision of $K^+$, in which the cells with a circle correspond to vertices of $K^+$, and those with a bar are obtained by dividing the edge of $K^+$ the bar represents. Not all edges are subdivided: adjacent circles in $G^+$ are indeed adjacent in $K^+$.}
    \label{fig:gamma_w}
  \end{figure}

  The label substitutions $\kappa_1$ and $\kappa_2$ can also be adaptated to $G^+$, defining $\gamma_1$ and $\gamma_2$. In $\gamma_1$, all vertices which do not represent a vertex of $K$ have label $\gcolor{normal}$. For layer $2$, each vertex of $G^+$ representing a vertex of $K$ keeps its label, each new vertex gets a label \gcolor{i_1}.
  
  Set $\Sigma_1^G = \Sigma_1$, except that $\lmpos$ takes values in $G$ rather than $K$. Then $\Sigma_2^G$ is $\Sigma_2$, except that $\gmanclocal$ takes values in $\Sigma_1^G$ rather than $\Sigma_1$.

  Then a circuit $C_G$ on $\Sigma_1^G \times \Sigma_2^G$ can be defined from $\gamma_w$, $\gamma_1$ and $\gamma_2$ like $C_\square$. That circuit $C_G$ has the same properties as $C_\square$ and is also self-descriptive. When accounting for the local part of the messages, gates which are upstream from all the $K$-vertices show the local message of the corresponding input meta-gate.

  Thus $C_G$ is self-descriptive.
\end{proof}

The question is now which $G$ are such that $G^{+}$ has a subdivision of $K^+$ as a subgraph. This condition seems constraining since $K$ is a rather large graph with its 32 vertices. Yet, one can make any $G$ larger by iterating the substitution generated by $G$ before trying to self-assemble $G^{\infty}$.

\begin{remark}
  Let $G \ni (0, 0)$ be a finite shape of $\mathbb{N}^2$. For any $k > 0$, $G^{\infty} = (G^k)^\infty$.
\end{remark}


By iterating $\sigma_G$, it is quite easy to get an instance of $K$ as a subconnector, as long as one starts with sufficient vertical and horizontal bandwidth.

\begin{lemma}
  Let $G \ni (0, 0)$ be a finite shape of $\mathbb{N}^2$. If $G[N \leftrightarrow S] \geq 2$ and $G[E \to W] \geq 2$, then $K \preceq G^3$
\end{lemma}

\begin{proof}
  Let $H = \{0, 1\}^2$. If $G[N \leftrightarrow S] \geq 2$ and $G[E \to W] \geq 2$, then $H \preceq G$: pick two disjoint North-South paths, two disjoint East-West paths, their four intersections can act as the vertices of $H$.

  By remark~\ref{lem:preceq_sigma}, since $H \preceq G$, $H^3 \preceq G^3$. But $H^3$ is an $8 \times 8$ grid, so $K \preceq H^3$.

  Hence, $K \preceq G^3$.
\end{proof}

Finally, this characterization is tight, by a generalization of the impossibility result of Hendricks et al.~\cite{hendricks_hierarchical_2020}.

\begin{theorem}
  Let $G \ni (0, 0)$ a finite shape of $\mathbb{N}$.

  If for all $k > 0$, there is an $x_k$ such that there is only one $y$ with $(x_k, y) \in (G^k)^+$ and $(x_k, y+1) \in (G^k)^+$, then $G^\infty$ cannot be strictly self-assembled in the aTAM model unless $G = \{\vec{0}\}$.
\end{theorem}

\begin{proof}
  Let $w$ and $h$ be the width and height of $G$. Assume that there is an aTAM $\mathcal{G}$ which self-assembles $G^{\infty}$ at temperature $\tau$.

  Let $C_k$ be the set of all glues which appear on the eastern edge of position $(x_k, y)$ for some value of $y$. Since $\mathcal{G}$ assembles $G^{\infty}$, all the $C_k$ are non-empty, so let $C_{\infty} = \bigcap_k (\bigcup_{k' > k}C_{k'})$ be the set of glues appearing in infinitely many $C_k$; $C_{\infty}$ is also non-empty. Let $F_k \subseteq C_k$ be the set of glues $g$ such that there is a position $\vec{z} = (x_k, y)$ and a production $\Pi_{k,t}$ of $\mathcal{G}$ where:
  \begin{itemize}
  \item the eastern glue of $\Pi_{k,t}(\vec{z})$ is $g$
  \item $\Pi_{k,t}$ contains no tile right of $x_k$
  \end{itemize}
  Like $C_\infty$, $F_\infty = \bigcap_k(\bigcup_{k' > k} F_{k'})$ is non-empty. The sets $C_k$ and $F_k$ are illustrated on \cref{fig:Ck_and_Fk}.

  \begin{figure}
    \centering
    \begin{tikzpicture}[scale=.48, yscale=-1]
      \tikzpicturedependsonfile{tikz/Ck_and_Fk.tex}
      \input{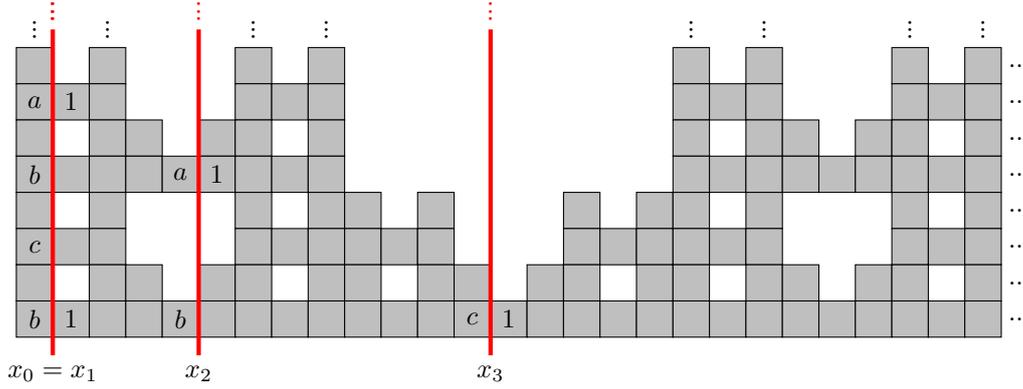}
    \end{tikzpicture}
    \caption{The set $C_k$ is the set of all glues which attach across the vertical lines at coordinate $x_k$. Ignoring any tiles which are not visible in the picture, $C_0 = C_1 = \{ a \cdot E, b \cdot E, c \cdot E \}$, $C_2 = \{ a \cdot E, b \cdot E\}$ and $C_3 = \{ c \cdot E \}$. The tiles marked with a $1$ are the ones which appear in their column in some production without other tiles right of the red line; the subset $F_k$ of $C_k$ contains their east glues: , $F_0 = F_1 = \{ a \cdot E, b \cdot E \}$, $C_2 = \{ a \cdot E \}$ and $C_3 = \{ c \cdot E \}$}
    \label{fig:Ck_and_Fk}
  \end{figure}
  
  Now let $t$ be a glue of $\mathcal{G}$, and $s_t$ be a (new) tile with $t$ on its eastern side, and $\epsilon$ on all other sides. Let $\mathcal{G}_t[x \geq 1]$ be the set of productions of $\mathcal{G}$, starting from the seed configuration with $s_t$ at $\vec{0}$ and attaching no tiles at positions $x < 1$. Let $u_t = \min_{p \in \mathcal{G}_t[x>1]} \max \{y | (x,y) \in p\}$ and $d_t =  \min_{p \in \mathcal{G}_t[x>1]} \min \{y | (x,y) \in p\}$. These two quantities are illustrated on \cref{fig:ut_and_dt}.

  \begin{figure}
    \centering
    \begin{tikzpicture}
      \input{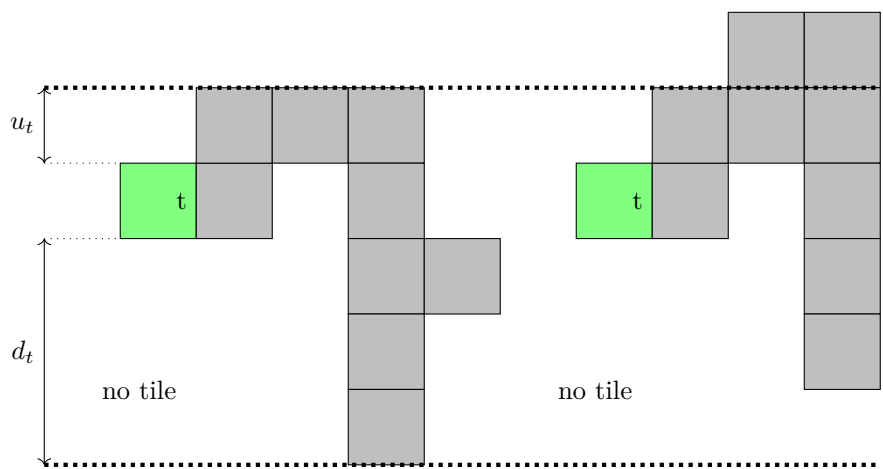}
    \end{tikzpicture}
    \caption{The definition of $u_t$ and $d_t$ given the two productions reachable from the glue $t$, without going left of the seed.}
    \label{fig:ut_and_dt}
  \end{figure}

  Since $G^\infty \subset \mathbb{N}^2$, it does not contain any infinite southward path. Therefore, for any $k$ and $t \in F_k$, $d_t > - \infty$, otherwise from $\Pi_{k,t}$ it would be possible to produce paths going arbitrarily far down, out of $\mathbb{N}^2$. Let $d_\infty = \max_{t \in F_{\infty}} \{d_t  | t \in \mathbb{G} \}$, $d_\infty$ is finite. On the other hand, $u_\infty = \max \{ u_t | t \in F_{\infty} \}$ must be infinite. Indeed, if it is finite, let $s = u_\infty -d_\infty + 1$. Consider a maximal production $\Pi^*$ obtained by only placing tiles left of the vertical line $L^*$ at $x$-coordinate $x_s$. Starting from $\Pi^*$, the only attachable positions are isolated points on $L^*$, separated by a distance at least $w^s$. From each of them, it is possible to grow an arm which reaches no higher than $s$ upwards, whence, as illustruted on \cref{fig:dt_finite_ut_infinite}, the part right of $L^*$ of that production is contained within a union of $(u_\infty - d_\infty)$-width horizontal bands, so its domain cannot be $G^\infty$.

  \begin{figure}
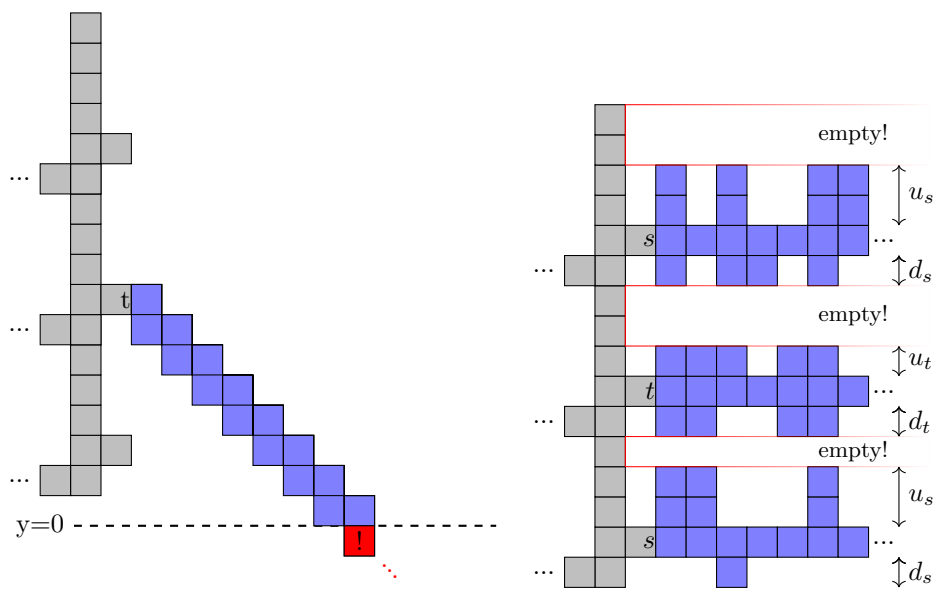

    \centering
    \begin{tikzpicture}[scale=.4]
      \input{tikz/dt_finite}
    \end{tikzpicture}%
    \quad%
    \begin{tikzpicture}[scale=.4]
      \input{tikz/ut_infinite}
    \end{tikzpicture}
    \caption{Faulty terminal productions which can be reached, left if some $d_{t}$ is infinite, right if all $u_t$ are finite: $\min_t(d_t) = -1$, $u_t = 1$, $u_s = 2$.}
    \label{fig:dt_finite_ut_infinite}
  \end{figure}

  For $t$ such that $u_t = +\infty$, all terminal productions on the right half-plane reach infinitely high. Thus, by applying~\cref{lem:tree_pump} with increasing values of $m$, either one of these terminal productions contains a periodic path going up, or there are productions $P$ with arbitrarily large squares in their fill-in $\complete{P}$. In both cases, for each $k$, there is a production $F_k$ for which some $k \times k$ square is inaccessible from any point $(0, y)$ with $y > 0$ without crossing $F_k$, as illustrated on~\cref{fig:peephole}.
    
  For a large enough level of substitution, any glue $t$ for which $u_t$ is finite does not even grow one ``meta-cell'' up. Let $F_{!} = \{ g \in F_\infty | u_g = +\infty \}$, and $s_{!}$ be such that $\max \{ u_g | g \notin F_{!} \} < h^{s_{!}}$ and $|d_\infty| <  h^{s_{!}}$.

  Consider a vertical line $L_x$ at some position $x$ between copies of $G^{s_!}$ in $G^\infty$. Let $P_y$ be a maximal subproduction of $P$ which can be assembled without attaching any tile right of $L_y$. In $P_y$, let $\vec{z}_{x} = (x, y)$ be the lowest position on $L_{x}$ containing a tile $T$ of $F_{s!}$.

  \begin{figure}
    \centering
    \begin{tikzpicture}
      \tikzpicturedependsonfile{tikz/peephole.tex}
      \input{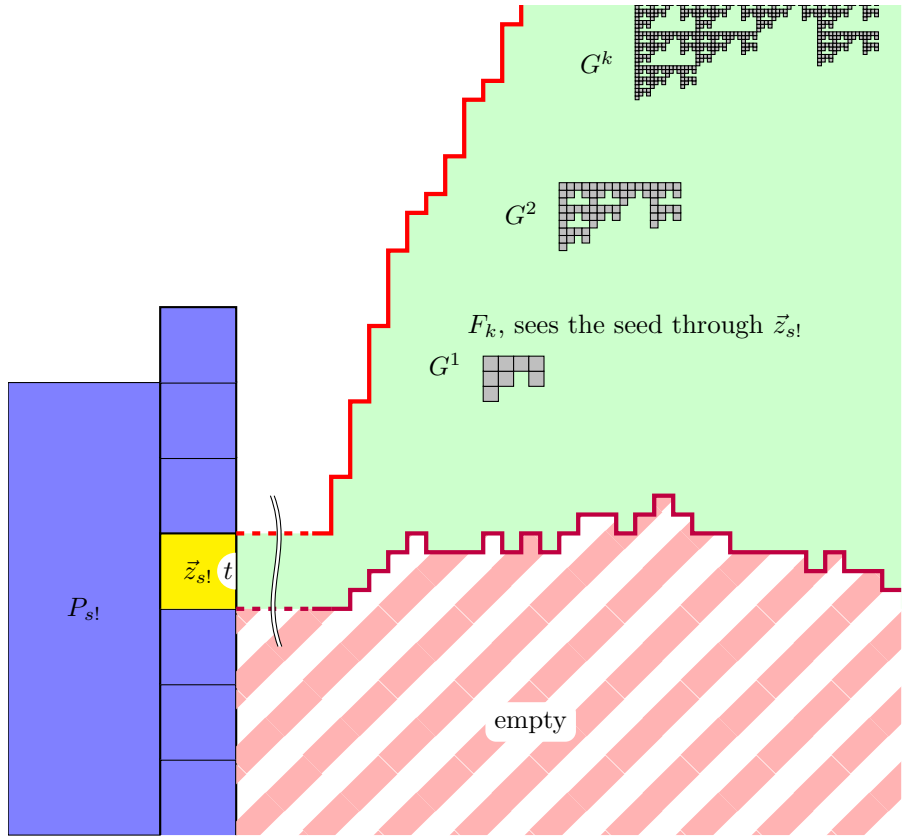}
    \end{tikzpicture}
    \caption{The part of the production that sees the seed only through $\vec{z}_!$ contains cycles of $G^\infty$ of arbitrary size.}
    \label{fig:peephole}
  \end{figure}

  By the pigeonhole principle, there must be two vertical lines $L_n$ and $L_m$ between copies of $G^{s_{!}}$ such that $P_n(\vec{z}_n) = P_m(\vec{z}_m)$. Let $\vec{v}$ be the vector $\vec{z}_m - \vec{z}_n$. By construction of $P_n$ and $P_m$, for any production $P_n^{!}$ obtained from $P_n$, there is a production $P_m^{!}$ obtained from $P_m$ such that the part of $P_n^{!}$ above and to the right of $\vec{z}_n$ and the part of $P_m^{!}$ above and to the right of $\vec{z}_m$ are identically filled. Hence, there is a vector which preserves either arbitrarily large cycles of $G^\infty$ or its intersection with a half-plane, thus $G$ must be trivial.
\end{proof}

This leaves open the case where in all the iterates $(G^k)^+$, there is a vertex which disconnects either the north and south or the east and west, but no straight line which crosses $(G^k)^+$ only once. A more precise variant of the previous proof takes care of that case.

\begin{lemma}
  \label{thm:kill_llama}
  Let $G \ni (0, 0)$ a finite shape of $\mathbb{N}$, and for $k \in \mathbb{N}$, $G^k = \sigma_G^k(\{(0, 0)\})$.

  If $\sup_{k}(G^k[N \leftrightarrow S]) = 1$ or $\sup(G^k[E \leftrightarrow W]) = 1$, then $G_\infty$ cannot be strictly self-assembled in the aTAM model unless it is trivial.
\end{lemma}

\begin{proof}
  Omitted due to laziness~\cite{lafargue_droit_1880}. In the previous proof, take a cut per $k$ which passes through the first occurence of each level of east-west bridges. Note that these cuts are not necessarily straight, but they remain within a bounded band around the bridges. The proof concludes as the previous one.
\end{proof}

\cref{thm:kill_llama} gives a dichotomy between generators for which the associated fractal can be strictly self-assembled and those for which it cannot. There is a polynomial time algorithm for determining which of \cref{lem:build_through_subconnector} or \cref{thm:kill_llama} applies.

\begin{theorem}
  There is a polynomial time algorithm which, on input $G$ decides whether:
  \begin{itemize}
  \item there is a $k$ such that $G^k[N \leftrightarrow S] \geq 2$ and $G^k[E \leftrightarrow W] \geq 2$ and thus $G^\infty$ can be strictly self-assembled in the aTAM, or
  \item for all $k$, $G^k[N \leftrightarrow S] < 2$ and thus $G^\infty$ cannot be strictly self-assembled in the aTAM, or
  \item for all $k$, $G^k[E \leftrightarrow W] < 2$ and thus $G^\infty$ cannot be strictly self-assembled in the aTAM.
  \end{itemize}

  Notice that the second and third case are \emph{not} mutually exclusive.
\end{theorem}

\begin{proof}
  Without loss of generality, assume $G$ is connected.

  First, use a max-flow algorithm to compute $G[d \leftrightarrow d']$ for all directions $d$, $d'$.
  If both $G[E \leftrightarrow W]$ and $G[N \leftrightarrow S]$ are $2$ or more, then the first case applies.

  Let $D$ be a set of pairs of directions, $v$ is a $D$-disconnector if $v$ disconnects $G^{+d}$ from $G^{+d'}$ for any $(d, d') \in D$. If $G[E \leftrightarrow W] = 1$, then there is at least one $\{EW\}$-disconnector. For any $D$-disconnector, define $\Cause(v, D)$ as the set of pair of directions $(\delta, \delta')$ such that there are $(d, d') \in D$ and a path from $G^{+d}$ to $G^{+d'}$ which enters $v$ from direction $\delta$ and leaves it through direction $\delta'$.

  Compute the disconnect causation graph $\Delta$. It is a directed graph whose vertices are couples $(v, D)$ where $v$ is a $D$-disconnector, and there is an arc from $(v, D)$ to each vertex $(v', \Cause(v, D))$. The causation graph $\Delta$ has size at most $64 |G|$ and the existence of each of its arcs can be tested in polynomial time.

  The graph $\Delta$ contains a cycle reachable from $(v, \{EW\})$, if and only if for all $k$, $G^k[E \leftrightarrow W] < 2$, likewise for $N$ and $S$. Indeed, for $k > 0$, a vertex $v$ is a $D$-disconnector in $G^k$ if and only if:
  \begin{itemize}
  \item The vertex of $\quot{v}{G^{k-1}}$ of $G$ corresponding to the copy of $G^{k-1}$ containing $v$ is a $D$-disconnector, and
  \item the vertex $(v \bmod G^{k-1})$ of $G^{k-1}$ corresponding to the position of $v$ within that copy is a $\Cause(v,D)$-disconnector.
  \end{itemize}

  Hence, from $\Delta$, it is possible to distinguish between the three cases.
\end{proof}

\appendix
\section{Proof of the Tree Pump Lemma}
\label{apx:tree_pump}

The proof of~\cref{lem:tree_pump} needs some ancillary definitions. The objects they introduce may seem overly sophisticated in view of the concreteness of the statement of \cref{lem:tree_pump}, but one needs to soldier on. The reader may find solace in knowing that they are the result of persistent failure in the author's search of a more down-to-earth argument. The proof plainly needs systems endowed with the ability to escape the confines of the Euclidian plane $\mathbb{Z}^2$ and those of a merely infinite time. Hence, the definition and use of \emph{Ordinal-length} sequences of \emph{Free Assemblies}.

\subsection{Free Assemblies}

 In this appendix, assemblies sit on a graph called an \emph{assembly support} instead of $\mathbb{Z}^2$.

\begin{definition}[Assembly support]
  An \emph{assembly support} is a directed graph $G$ with labels in $\{N, E, S, W\}$ on its arcs, such that:
\begin{itemize}
\item each vertex has at most one incoming arc with each label,
\item each vertex has at most one outgoing arc with each label,
\item for any cycle $c$ of $G$, the sum of the labels of the arcs of $c$ is equal to $\vec{0}$ as a vector of $\mathbb{Z}^2$,
\item for any path $\pi$ of $1$ or $3$ arcs from $u$ to $v$, if the sum in $\mathbb{Z}^2$ of the labels of $\pi$ is $\vec{d} \in \Dir$, then there is an arc from $v$ to $u$ with label $-\vec{d}$.
\end{itemize}
\end{definition}

An example of a correct assembly support is given on \cref{fig:ext_arcs_steps_holes}, while \cref{fig:bad-ass-support} shows some counter-examples of labeled graphs which are not assembly supports. Because of the conditions on the incoming and outgoing arcs of each vertex, it makes sense to represent the vertices as tiles, and arcs as sides the tiles share. This hopefully helps build the intuition that assembly support are ``like $\mathbb{Z}^2$, but long separated paths which should come to the same position may avoid each other''. On \cref{fig:bad-ass-support}, the direction corresponding to each side of the tile is given explicitely. With correct Assembly Supports, like the ones on \cref{fig:ext_arcs_steps_holes}, the convention is that the direction of each side corresponds roughly to its direction on the page, and small adjustments allow tiles which would otherwise be adjacent on the page not to be. The appearance of tiles avoiding each other by going in the third dimension is deliberate, as a way to build intuition.

An \emph{embedding} from an assembly support $G$ to another assembly support $G'$ is a graph embedding which preserves the label of the arcs. An \emph{isomorphism} between $G$ and $G'$ is a graph isomorphism preserving the arc labels.

\begin{definition}[Perspective difference, Translation]
  For any two vertices $(z, z')$ of a connected component of an assembly support $G$, their \emph{perspective difference} $z' \pd z$ is $\sum_{a \in p} a$, where $p \in \Dir^*$ is the sequence of labels of a path from $z$ to $z'$. For any embedding $e: G \to \mathbb{Z}^2$, $e(z') - e(z) = z' \pd z$. By convention, $z, z'$ are in different connected components, $z' \pd z = \infty$.

  A translation between two subgraphs $A \subset G$ and $B \subset G'$ of two assembly support is an isomorphism from $A$ to $B$ which preserves perspective differences. Any isomorphism between connected subgraphs is a translation.
\end{definition}

For any Assembly Support $G$ and $z \in G$, $\ezEmbed[z, G]: G \to \mathbb{Z}^2$ is the embedding $z' \mapsto (z' \pd z)$. Since changing $z$ in $\ezEmbed[z, G]$ only amounts to a translation, it will generally be omitted. The index $G$ will be omitted as well when this causes no ambiguity.

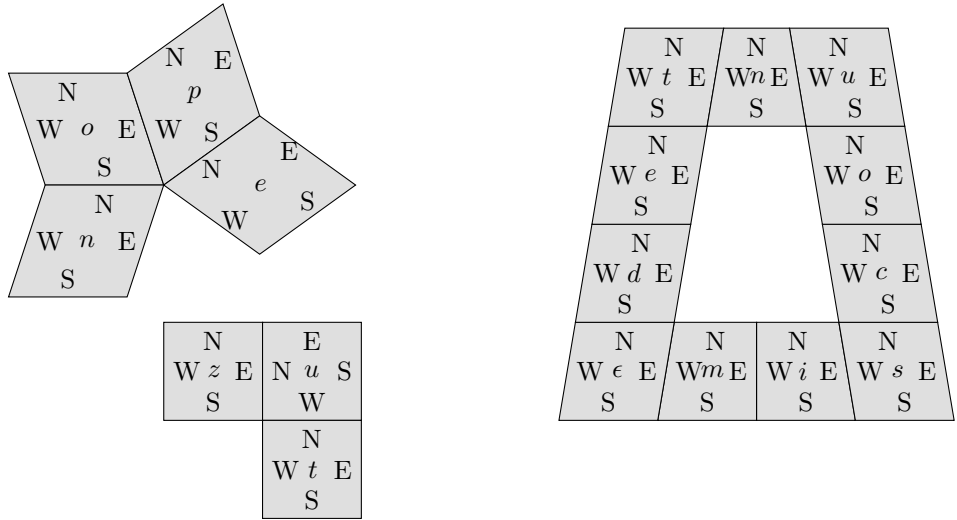
\begin{figure}
  \centering
  \newcommand{\crookedTile}[5]{
  \begin{scope}
    \filldraw[fill=gray!25!white] #1 -- node (s) [above] {S} #2 -- node (e) [left] {E} #3 -- node (n) [below] {N} #4 -- node (w) [right]{W} cycle;
    \node at (barycentric cs:s=1,e=1,n=1,w=1) {#5};
  \end{scope}
}

\begin{tikzpicture}
  \begin{scope}[scale=1.3]
    \crookedTile{(0,0)}{(1,0)}{(1,1)}{(0,1)}{$z$}
    \filldraw[fill=gray!25!white] (1,0) -- node[above] {W} (2,0) -- node[left] {S} (2,1) -- node[below] {E} (1,1) -- node[right] {N} cycle;
    \node at (1.5,.5) {$u$};
    \crookedTile{(1,-1)}{(2,-1)}{(2,0)}{(1,0)}{$t$}

    \begin{scope}[scale=1.2,shift={(0,2)}]

      \crookedTile{(-1,0)}{(0,0)}{++(108:1)}{++(-1,0)}{$o$}
      \begin{scope}[shift={(72:-1)}]
        \crookedTile{(-1,0)}{(0,0)}{++(72:1)}{++(-1,0)}{$n$}
      \end{scope}
      \crookedTile{(0,0)}{++(36:1)}{++(108:1)}{++(36:-1)}{$p$}
      \begin{scope}[shift={(-36:1)}]
        \crookedTile{(0,0)}{++(36:1)}{++(-36:-1)}{++(36:-1)}{$e$}
      \end{scope}

    \end{scope}

    \begin{scope}[shift={(6,0)}]
      \crookedTile{(0,0)}{(1,0)}{(0.833,1)}{(0,1)}{$i$}
      \crookedTile{(1,0)}{(2,0)}{(1.833,1)}{(0.8333,1)}{$s$}
      \crookedTile{(-1,0)}{(0,0)}{(0,1)}{(-0.8333,1)}{$m$}
      \crookedTile{(-2,0)}{(-1,0)}{(-0.8333,1)}{(-1.8333,1)}{$\epsilon$}
      \crookedTile{(0.8333,1)}{(1.8333,1)}{(1.666,2)}{(0.666,2)}{$c$}
      \crookedTile{(-1.8333,1)}{(-0.8333,1)}{(-0.666,2)}{(-1.666,2)}{$d$}
      \crookedTile{(0.666,2)}{(1.666,2)}{(1.5,3)}{(0.5,3)}{$o$}
      \crookedTile{(-1.666,2)}{(-0.666,2)}{(-.5,3)}{(-1.5,3)}{$e$}
      \crookedTile{(0.5,3)}{(1.5,3)}{(1.333,4)}{(0.333,4)}{$u$}
      \crookedTile{(-1.5,3)}{(-0.5,3)}{(-0.333,4)}{(-1.333,4)}{$t$}
      \crookedTile{(-0.5,3)}{(0.5,3)}{(0.333,4)}{(-0.333,4)}{$n$}
    \end{scope}
  \end{scope}

\end{tikzpicture}
  \caption{\emph{Ceci n'est pas un Assembly Support}: three labelled graphs which fail to be assembly supports. Each vertex is represented by a tile, the labels of its sides are the labels of its outgoing arcs. A label on a side with no neighbor corresponds to an exterior arc. The bottom-left graph has a path of length $1$ with label $E$ from $z$ to $u$ but the $-E = W$ neighbor of $u$ is $t$. Also, the $N$ neighbor of $t$ is $u$, which has no $S$ neighbor. The one in the top right has a $3$ arc path ``$nope$'' whose labels sum to $N + E + S = E$, but $n$ is not the $W$ neighbor of $e$ (actually, there is none). The one on the right has a cycle ``$\epsilon\, miscounted$'' for which the sum of the labels is $3E + 3N + 2W + 3S = E$, which is not $\vec{0}$.}
  \label{fig:bad-ass-support}
\end{figure}

The analysis in the Tree Pump Lemma relies on an examination of the \emph{holes} of the productions of $\Ss$. These holes need to be suitably defined now that the assemblies are no longer necessarily planar.

\begin{definition}[Exterior Arc]
  An \emph{exterior arc} of an assembly support $G$ is a pair of a vertex $v$ of $G$ and a direction $d$ such that $v$ has no arc labelled $d$ in $G$.
\end{definition}

\begin{definition}[Converging Exterior Arcs, Growth]
  Two exterior arcs $(v, d)$, $(v', d')$ are \emph{converging} when $v' \pd v = d' - d$ --they virtually point to the same empty position.

  Given an assembly support $G$ and a set $Z$ of converging arcs, the assembly support $G + Z$ ($G$ grown by $Z$) is $G \cup \{\zeta\}$, where there is an arc $z \to \zeta$ (for $z \in G$) in direction $d$ whenever $(z, d) \in Z$.
\end{definition}

\begin{definition}[Step, Exterior Path, Hole]
  There is a \emph{step} between two arcs or exterior arcs $(v, d)$ and $(v', d')$ if:
  \begin{itemize}
  \item $G = G'$ and $d, d'$ form a $\pm \frac{\pi}{2}$ angle,
  \item $d = d'$ and there is an (non-exterior) arc from $v$ to $v'$ with label $d''$, with $d, d''$ forming a $\pm \frac{\pi}{2}$ angle
  \item $d, d'$ form a $\pm \frac{\pi}{2}$ angle and there is a cell $v''$ with $v'' + d = v$ and $v'' + d' = v'$.
  \end{itemize}

  A (simple) \emph{exterior path} is a sequence $P = (p_0, p_1, \ldots, p_{k})$ of arcs, such that:
  \begin{itemize}
  \item for each $i$, there is a step between $p_i$ and $p_{i+1}$, and
  \item for $0 < i < k$, $p_i$ is an exterior arc.
  \end{itemize}

  A \emph{hole} is a simple exterior path which loops back to its starting exterior arc.

  The \emph{perimeter} of a hole $H = (h_0, \ldots h_k = h_0)$ is the set of vertices appearing in the exterior vertices $h_i$.
\end{definition}

These definitions are illustrated on~\cref{fig:ext_arcs_steps_holes}, which features an assembly support with two holes of perimeter $20$. These holes are ``morally infinite'', so they resist attempts at a definition of area ---at least the author's naive ones. For an assembly which embeds into $\mathbb{Z}^2$, the free definition of hole corresponds to the intuition of a hole in an assembly, up to the detail that the exterior of the assembly forms a hole, so that all finite assemblies have at least one hole. Note that exterior paths are oriented: they turn in the positive direction around each tile.

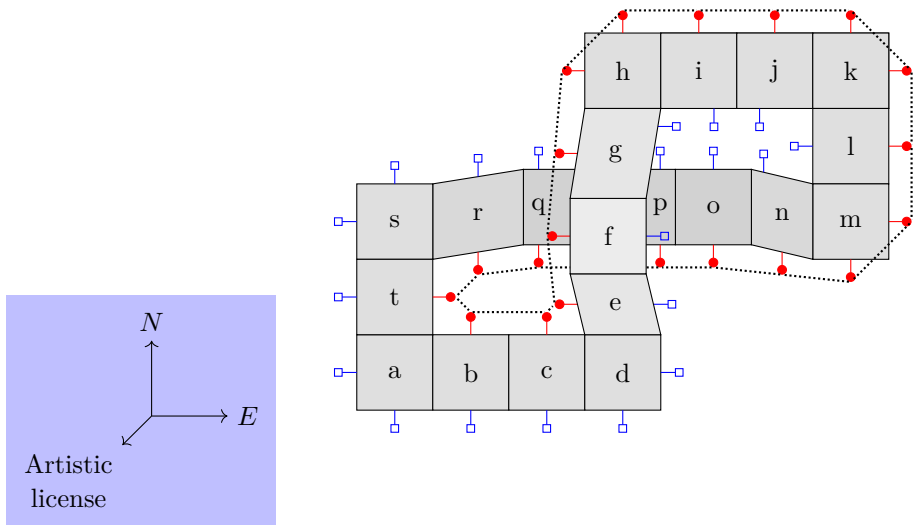
\begin{figure}
  \centering
  \begin{tikzpicture}
\matrix[left=3em of {(0,0)}, fill=blue!25!white]{
    \begin{scope}
      \draw[->] (0, 0, 0) -- +(xyz cs:x=1) node[right] {$E$};
      \draw[->] (0, 0, 0) -- +(xyz cs:y=1) node[above] {$N$};
      \draw[->] (0, 0, 0) -- +(xyz cs:z=1) node[below left, align=center] {Artistic\\ license};
    \end{scope}\\
};   

  \foreach \x / \y / \lab in { 0 / 0 / a, 1 / 0 / b, 2 / 0 / c, 3 / 0 / d }
  \filldraw[fill=gray!25!white] (\x,\y) rectangle +(1,1) node [pos=.5] {\lab};

  \foreach \x / \y / \lab in { 3 / 4 / h}
  \filldraw[fill=gray!25!white] (\x,\y) rectangle +(1,1)  node[pos=.5] {\lab};

  \foreach \x / \y / \lab in { 4 / 4 / i, 5 / 4 / j, 6 / 4 / k}
  \filldraw[fill=gray!25!white] (\x,\y) rectangle +(1,1)  node[pos=.5] {\lab};

  \foreach \x / \y / \lab in { 4 / 4 / i, 5 / 4 / j, 6 / 4 / k}
  \filldraw[fill=gray!25!white] (\x,\y) rectangle +(1,1)  node[pos=.5] {\lab};

  \foreach \x / \y / \lab in { 6 / 3 / l, 6 / 2 / m}
  \filldraw[fill=gray!25!white] (\x,\y) rectangle +(1,1)  node[pos=.5] {\lab};

  \filldraw[fill=gray!30!white] (5, 2, -.5) -- ++(1, 0, .5) -- ++(0, 1, 0) -- ++(-1, 0, -.5) -- cycle;
  \node at (5.5, 2.5, -.25) {n};
  
  \foreach \x / \y / \lab / \labshift in { 4 / 2 / o / 0, 3 / 2 / p / .3, 2 / 2 / q / -.3}
  \filldraw[fill=gray!35!white] (\x,\y, -.5) rectangle +(1,1, 0)  node[pos=.5, shift={(\labshift, 0)}] {\lab};

  \filldraw[fill=gray!30!white] (1, 2, 0) -- ++(1, 0, -.5) -- ++(0, 1, 0) -- ++(-1, 0, .5) -- cycle;
  \node at (1.5, 2.5, -.25) {r};

  \foreach \x / \y / \lab in {0 / 2 / s, 0 / 1 / t}
  \filldraw[fill=gray!25!white] (\x,\y, 0) rectangle +(1,1, 0)  node[pos=.5] {\lab};

  \foreach \x / \y / \z / \dir [count=\n] in { %
    1.5 / 1 / 0 / 90, 2.5 / 1 / 0 / 90, 3 / 1.5 / .25 / 180, 3 / 2.5 / .5 / 180, 3 / 3.5 / .25 / 180, %
    3 / 4.5 / 0 / 180, 3.5 / 5 / 0 / 90, 4.5 / 5 / 0 / 90, 5.5 / 5 / 0 / 90, 6.5 / 5 / 0 / 90, %
    7 / 4.5 / 0 / 0, 7 / 3.5 / 0 / 0, 7 / 2.5 / 0 / 0, 6.5 / 2 / 0 / -90, 5.5 / 2 / -.25 / -90, %
    4.5 / 2 / -.5 / -90,  3.8 / 2 / -.5 / -90, 2.2 / 2 / -.5 / -90, 1.5 / 2 / -.25 / -90, 1 / 1.5 / 0 / 0 %
  }
    \draw[-{Circle}, red] (\x, \y, \z) -- +(\dir:.3) node[name=circle\n, inner sep=0] {};

  \foreach \x / \y / \z / \dir [count=\n] in { %
    .5 / 0 / 0 / -90, 1.5 / 0 / 0 / -90, 2.5 / 0 / 0 / -90, 3.5 / 0 / 0 / -90,
    4 / .5 / 0 / 0, 4 / 1.5 / .25 / 0, 4 / 2.5 / .5 / 0, 4 / 3.8 / .1 / 0, %
    4.7 / 4 / 0 / -90, 5.3 / 4 / 0 / -90, %
    6 / 3.5 / 0 / 180, 5.2 / 3 / -.4 / 90, %
    4.5 / 3 / -.5 / 90,  3.8 / 3 / -.5 / 90, 2.2 / 3 / -.5 / 90, 1.5 / 3 / -.25 / 90, .5 / 3 / 0 / 90, %
    0 / 2.5 / 0 / 180, 0 / 1.5 / 0 / 180, 0 / .5 / 0 / 180 
  }
  {
    \draw[-{Square[open]}, blue] (\x, \y, \z) -- +(\dir:.3) node[name=square\n, inner sep=0] {};
  }

  \foreach \k [remember=\k as \lastk (initially 20)] in {1, ..., 20}
  \draw[thick, densely dotted] (circle\lastk) -- (circle\k);

  \filldraw[fill=gray!20!white] (3, 1, 0) -- (3, 2, .5) -- (4, 2, .5) -- (4, 1, 0) -- cycle;
  \node at (3.5, 1.5, .25) {e};

  \filldraw[fill=gray!15!white] (3, 2, .5) rectangle +(1,1, 0)  node[pos=.5] {f};

  \filldraw[fill=gray!20!white] (3, 3, .5) -- (3, 4, 0) -- (4, 4, 0) -- (4, 3, .5) -- cycle;
  \node at (3.5, 3.5, .25) {g};

\end{tikzpicture}
  \caption{An example of an assembly support, the shifts in the 3rd dimension between $e$ and $g$ as well as between $r$ and $n$ are artistic license to show the absence of adjacencies between those two branches. The sum of the vectors between neighbors on the cycle $abcdefghijklmnopqrsta$ is $3E + 4N + 3E + 2S + 6E + 2S = \vec{0}$. The sum of vectors between neighbors on the path $fedcbatsrqp$ is $2S + 3W + 2N + 3E = \vec{0}$, yet the vertices $f$ and $p$ are distinct. The small pending vertices next to the tiles are the exterior arcs; they form two holes, one for the square exterior arcs, and one for the circles. The dotted lines shows the steps between the exterior arcs of the ``circle'' hole. The step between the exterior arcs $(m, S)$ and $(m, E)$ stems from the first case of the definition, the one between $(b, N)$ and $(c, N)$ from the second case, and the one between $(t, E)$ and $(r, S)$ because of the third. In this last case, $s$ acts as the pivot between $r$ and $s$. Both holes have perimeter 20.}
  \label{fig:ext_arcs_steps_holes}
\end{figure}

\begin{definition}[Free Assembly]
  A free assembly $A$ of the TAS $\Ss$ is composed of an assembly support $\support[A]$ and a total function $\tiles[A]: \support[A] \rightarrow \tileset[\Ss]$.

  For a TAS $\Ss$, the set of free assemblies of $\Ss$ is written $\FreeAssemblies{\Ss}$.

\end{definition}

When there is an injective embedding $e: \support[A] \to G$, it naturally defines a free assembly $e(A)$ with $\support[e(A)] = e(\support[A])$ and for $z \in e(\support[A])$, $e(A)(e(z)) =  A(z)$. An assembly $A$ \emph{embeds} into $\mathbb{Z}^2$ if $\ezEmbed$ is a one-to-one embedding $\support[A] \to \mathbb{Z}^2$; as above, this embedding naturally defines an assembly $\ezEmbed(A)$.

\begin{definition}[Free Attachment, Free Assembly Sequence]
  Let $A, A'$ be free assemblies of some TAS system $\Ss$, and $\eta: A \to A'$ an embedding. The tuple $(A, A', \eta)$ is an assembly candidate if:
  \begin{itemize}
  \item there is a $z \in \support[A']$ such that $\eta: \support[A] \to \support[A'] \setminus \{ z \}$ is a translation,
  \item for any $z' \in \support[A]$, $\tiles[A](z') = \tiles[A](\eta(z'))$.
  \end{itemize}

  The definitions of attachment, assembly sequence, production and terminal production extend to free assemblies. The $\mathbb{Z}^2$ version of each definition is simply the particular case where every assembly embeds into $\mathbb{Z}^2$.

  For a TAS $\Ss$, the set of free assembly sequences of $\Ss$ is written $\FreeSeqs{\Ss}$.  
\end{definition}

The definition of free attachment is illustrated on~\cref{fig:free-attach}. The fact that $\eta$ is an isomorohpism ensures that parts of the assembly which do not touch the position $z$ are unaffected by the attachment. This preserves the locality of the aTAM process, in contrast to what happens in the FTAM~\cite{DBLP:journals/nc/Durand-LoseHPPS20} for instance. Hence, while the above definition of attachment is given from the point of view of $A'$, ``after the fact'', from the point of view of $A$, an attachment candidate can be defined from a tile type $t$ and a set $Z$ of convergent exterior arcs: $t@Z$ is the attachment candidate $(A, A', \eta)$ where:
\begin{itemize}
\item $\support[A'] = \support[A] + Z$, and $\zeta$ is the vertex of $\support[A']$ which is not in $\support[A]$,
\item $\eta$ is the identity on $\support[A]$,
\item $\tiles[A'](\zeta) = t$,
\item $\tiles[A'](z) = \tiles[A](z)$ whenever $z \neq \zeta$.
\end{itemize}

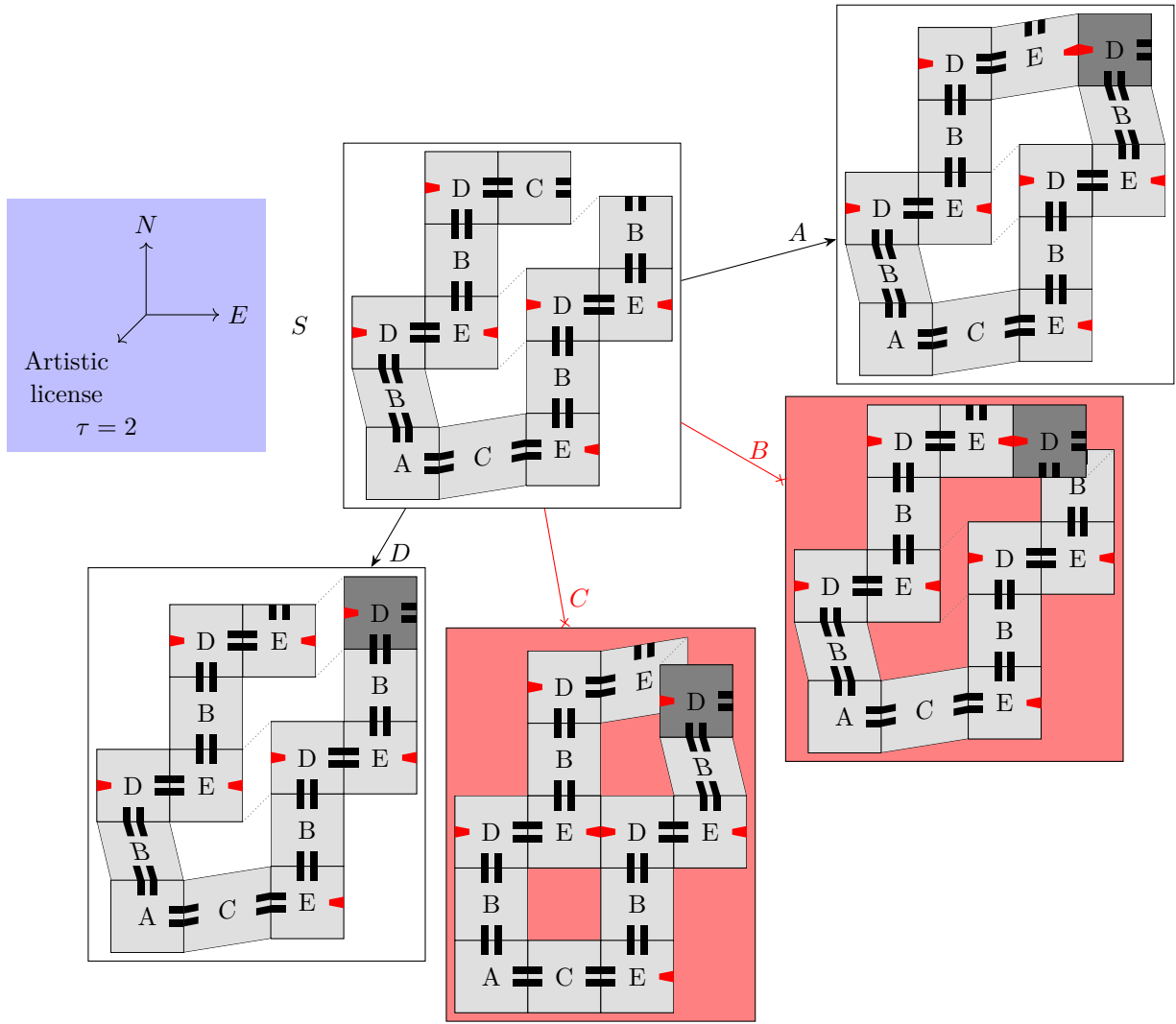
\begin{figure}
  \centering
  \begin{tikzpicture}

  \tikzset{
    tileA/.pic={
      \clip (0, 0) rectangle (1,1);
      \filldraw[fill=gray!25!white] (0, 0) rectangle (1, 1) node[anchor=center, pos=.5] {A};
      \pic at (.5, 1) {doubleGlue};
      \pic[rotate=-90] at (1, .5) {doubleGlue};
    }
  }
  
  \tikzset{
    tileB/.pic={
      \clip (0, 0) rectangle (1,1);
      \filldraw[fill=gray!25!white] (0, 0) rectangle (1, 1) node[anchor=center, pos=.5] {B};
      \pic at (.5, 1) {doubleGlue};
      \pic[rotate=180] at (.5, 0) {doubleGlue};
    }
  }

  \tikzset{
    tileBUp/.pic={
      \clip (0, 0, 0) -- (0, 1, 0.5) -- (1, 1, 0.5) -- (1, 0, 0) -- cycle;
      \filldraw[fill=gray!25!white] (0, 0)  (0, 0, 0) -- (0, 1, 0.5) -- (1, 1, 0.5) -- (1, 0, 0) -- cycle;
      \node[xslant=-.3, anchor=center] at (.5, .5, .25) {B};
      \pic[xslant=-.3] at (.5, 1, .5) {doubleGlue};
      \pic[rotate=180, xslant=-.3] at (.5, 0, 0) {doubleGlue};
    }
  }

  \tikzset{
    tileC/.pic={
      \clip (0, 0) rectangle (1,1);
      \filldraw[fill=gray!25!white] (0, 0) rectangle (1, 1) node[anchor=center, pos=.5] {C};
      \pic[rotate=-90] at (1, .5) {doubleGlue};
      \pic[rotate=90] at (0, .5) {doubleGlue};
    }
  }

  \tikzset{
    tileCDown/.pic={
      \clip (0, 0, 0) -- (1, 0, -.5) -- (1, 1, -.5) -- (0, 1, 0) -- cycle;
      \filldraw[fill=gray!25!white] (0, 0, 0) -- (1, 0, -.5) -- (1, 1, -.5) -- (0, 1, 0) -- cycle;
      \node[anchor=center, yslant=.2] at (.5, .5, -.25) {C};
      \pic[yslant=.2, rotate=-90] at (1, .5, -.5) {doubleGlue};
      \pic[yslant=.2, rotate=90] at (0, .5, 0) {doubleGlue};
    }
  }

  \tikzset{
    tileD/.pic={
      \clip (0, 0) rectangle (1,1);
      \filldraw[fill=gray!25!white] (0, 0) rectangle (1, 1) node[anchor=center, pos=.5] {D};
      \pic[rotate=-90] at (1, .5) {doubleGlue};
      \pic[rotate=180] at (.5, 0) {doubleGlue};
      \pic[rotate=-90] at (0, .5) {singleGlue};
    }
  }

  \tikzset{
    tileDnew/.pic={
      \clip (0, 0) rectangle (1,1);
      \filldraw[fill=gray] (0, 0) rectangle (1, 1) node[anchor=center, pos=.5] {D};
      \pic[rotate=-90] at (1, .5) {doubleGlue};
      \pic[rotate=180] at (.5, 0) {doubleGlue};
      \pic[rotate=-90] at (0, .5) {singleGlue};
    }
  }

  \tikzset{
    tileE/.pic={
      \clip (0, 0) rectangle (1,1);
      \filldraw[fill=gray!25!white] (0, 0) rectangle (1, 1) node[anchor=center, pos=.5] {E};
      \pic[rotate=90] at (0, .5) {doubleGlue};
      \pic at (.5, 1) {doubleGlue};
      \pic[rotate=90] at (1, .5) {singleGlue};
    }
  }

  \tikzset{
    tileEDown/.pic={
      \clip (0, 0, 0) -- (1, 0, -.5) -- (1, 1, -.5) -- (0, 1, 0) -- cycle;
      \filldraw[fill=gray!25!white] (0, 0, 0) -- (1, 0, -.5) -- (1, 1, -.5) -- (0, 1, 0) -- cycle;
      \node[anchor=center, yslant=.2] at (.5, .5, -.25) {E};
      \pic[yslant=.2, rotate=90] at (0, .5, 0) {doubleGlue};
      \pic[yslant=.2] at (.5, 1, -.25) {doubleGlue};
      \pic[yslant=.2, rotate=90] at (1, .5, -.5) {singleGlue};
    }
  }

  \tikzset{
    tileF/.pic={
      \clip (0, 0) rectangle (1,1);
      \filldraw[fill=blue!25!white] (0, 0) rectangle (1, 1)  node[anchor=center, pos=.5] {F};
      \pic[rotate=180] at (.5, 0) {doubleGlue};
      \pic[rotate=180] at (.5, 1) {singleGlue};
    }
  }

  \tikzset{
    startAss/.pic={
      \pic {tileA};
      \pic at (0, 1, 0) {tileBUp};
      \pic at (0, 2, .5) {tileD};
      \pic at (1, 2, .5) {tileE};
      \pic at (1, 3, .5) {tileB};
      \pic at (1, 4, .5) {tileD};
      \pic at (2, 4, .5) {tileC};
      \pic at (1, 0, 0) {tileCDown};
      \pic at (2, 0, -.5) {tileE};
      \pic at (2, 1, -.5) {tileB};
      \pic at (2, 2, -.5) {tileD};
      \pic at (3, 2, -.5) {tileE};
      \pic at (3, 3, -.5) {tileB};
      \draw[densely dotted, gray] (2, 2, .5) -- (2, 2, -.5);
      \draw[densely dotted, gray] (2, 3, .5) -- (2, 3, -.5);
      \draw[densely dotted, gray] (3, 4, .5) -- (3, 4, -.5);
    }
  }

  \tikzset{
    attachDown/.pic={
      \pic {tileA};
      \pic at (0, 1, 0) {tileBUp};
      \pic at (0, 2, .5) {tileD};
      \pic at (1, 2, .5) {tileE};
      \pic at (1, 3, .5) {tileB};
      \pic at (1, 4, .5) {tileD};
      \pic at (2, 4, .5) {tileE};
      \pic at (1, 0, 0) {tileCDown};
      \pic at (2, 0, -.5) {tileE};
      \pic at (2, 1, -.5) {tileB};
      \pic at (2, 2, -.5) {tileD};
      \pic at (3, 2, -.5) {tileE};
      \pic at (3, 3, -.5) {tileB};
      \pic at (3, 4, -.5) {tileDnew};
      \draw[densely dotted, gray] (2, 2, .5) -- (2, 2, -.5);
      \draw[densely dotted, gray] (2, 3, .5) -- (2, 3, -.5);
      \draw[densely dotted, gray] (3, 4, .5) -- (3, 4, -.5);
      \draw[densely dotted, gray] (3, 5, .5) -- (3, 5, -.5);
    }
  }

  \tikzset{
    attachUp/.pic={
      \pic {tileA};
      \pic at (0, 1, 0) {tileBUp};
      \pic at (0, 2, .5) {tileD};
      \pic at (1, 2, .5) {tileE};
      \pic at (1, 3, .5) {tileB};
      \pic at (1, 4, .5) {tileD};
      \pic at (2, 4, .5) {tileE};
      \pic at (1, 0, 0) {tileCDown};
      \pic at (2, 0, -.5) {tileE};
      \pic at (2, 1, -.5) {tileB};
      \pic at (2, 2, -.5) {tileD};
      \pic at (3, 2, -.5) {tileE};
      \pic at (3, 3, -.5) {tileB};
      \pic at (3, 4, .5) {tileDnew};
      \draw[densely dotted, gray] (2, 2, .5) -- (2, 2, -.5);
      \draw[densely dotted, gray] (2, 3, .5) -- (2, 3, -.5);
      \draw[densely dotted, gray] (4, 4, .5) -- (4, 4, -.5);
    }
  }

    \tikzset{
    attachBoth/.pic={
      \pic {tileA};
      \pic at (0, 1, 0) {tileBUp};
      \pic at (0, 2, .5) {tileD};
      \pic at (1, 2, .5) {tileE};
      \pic at (1, 3, .5) {tileB};
      \pic at (1, 4, .5) {tileD};
      \pic at (2, 4, .5) {tileEDown};
      \pic at (1, 0, 0) {tileCDown};
      \pic at (2, 0, -.5) {tileE};
      \pic at (2, 1, -.5) {tileB};
      \pic at (2, 2, -.5) {tileD};
      \pic at (3, 2, -.5) {tileE};
      \pic at (3, 3, -.5) {tileBUp};
      \pic at (3, 4, 0) {tileDnew};
      \draw[densely dotted, gray] (2, 2, .5) -- (2, 2, -.5);
      \draw[densely dotted, gray] (2, 3, .5) -- (2, 3, -.5);
    }
  }

    \tikzset{
    attachUpSolder/.pic={
      \pic {tileA};
      \pic at (0, 1, 0) {tileB};
      \pic at (0, 2, 0) {tileD};
      \pic at (1, 2, 0) {tileE};
      \pic at (1, 3, 0) {tileB};
      \pic at (1, 4, 0) {tileD};
      \pic at (2, 4, 0) {tileEDown};
      \pic at (1, 0, 0) {tileC};
      \pic at (2, 0, 0) {tileE};
      \pic at (2, 1, 0) {tileB};
      \pic at (2, 2, 0) {tileD};
      \pic at (3, 2, 0) {tileE};
      \pic at (3, 3, 0) {tileBUp};
      \pic at (3, 4, .5) {tileDnew};
      \draw[densely dotted, gray] (3, 5, .5) -- (3, 5, -.5);
    }
  }

  \node[matrix, draw] (start) at (0, 0) {\pic {startAss};\\};
  \node[matrix, draw] (both) at (15:7) {\pic {attachBoth};\\};
  \node[matrix, fill=red!50!white, draw] (up) at (-30:7) {\pic {attachUp};\\};
  \node[matrix, fill=red!50!white, draw] (solder) at (-80:7) {\pic {attachUpSolder};\\};
  \node[matrix, draw] (down) at (-120:7) {\pic {attachDown};\\};
  \node[left=1em of start] {$S$};

  \matrix[left=3em of start, fill=blue!25!white]{
    \begin{scope}
      \draw[->] (0, 0, 0) -- +(xyz cs:x=1) node[right] {$E$};
      \draw[->] (0, 0, 0) -- +(xyz cs:y=1) node[above] {$N$};
      \draw[->] (0, 0, 0) -- +(xyz cs:z=1) node[below left, align=center] {Artistic\\ license};
    \end{scope}\\

    \node {$\tau = 2$};\\
  };

  \draw[-{Stealth}] (start) -- (both) node[near end, above] {$A$};
  \draw[-{Rays[]}, red] (start) -- (up) node[near end, above] {$B$};
  \draw[-{Rays[]}, red] (start) -- (solder) node[near end, right] {$C$};
  \draw[-{Stealth}] (start) -- (down) node[near end, right] {$D$};
  

\end{tikzpicture}

  \caption{Some possible and impossible free attachments from a starting free assembly $S$, at temperature $\tau$. Again, the third coordinate is a visual aid to show non-adjacencies; small dashed lines link points which embed in the same position in $\mathbb{Z}^2$. Assemblies $A$ and $D$ are reachable through one attachment $t_D@z$, where in each case, $z$ is the position in dark gray. The assembly $B$ is an attachment candidate, but it is not stable, since the new tile $t_D$ only binds through a strength-$1$ glue. The assembly $C$ can not even form an attachment candidate from $S$, as $\support[C]$ does not contain an isomorphic copy of $\support[S]$}
  \label{fig:free-attach}
\end{figure}

For clarity, a ``normal'' assembly will be referred to as a $\mathbb{Z}^2$-assembly, likewise for the other concepts which were just endowed with a ``free'' variant.

Embeddings act not only on assembly supports and assemblies, but also on assembly sequences. In doing so, they may break the correctness of attachments if they are not one-to-one; when that happens, the sequence is cut short.

\begin{definition}[Assembly Sequence Embedding]
  Let $\alpha = (A_0 \to A_1 \to \ldots) \in \FreeSeqs{\Ss}$, let $G = \support[(\lim \alpha)]$ and $G'$ an assembly support. Let $e: G \to G'$ be an embedding; since up to translation, $\support[A_0] \subset \support[A_1] \subset \ldots \subset \support[(\lim \alpha)]$, for each $i$, $e$ induces an embedding from $\support[A_i]$ into $G'$.

  Let $k$ be the largest $i$ such that $e$ is an isomorphism on $\support[A_i]$, then \(e(\alpha)\) is the assembly sequence $(e(A_0), \ldots, e(A_k))$.
\end{definition}

\subsection{Ordinal Assembly Sequences}

It is often practical to consider what happens in an assembly sequence $\alpha$ after it has placed an infinity of tiles. For this, \emph{ordinal} assembly sequences come useful.

\begin{definition}[Ordinal Assembly Sequence]
  Let $o \in \omega_1$ be a countable ordinal, an $o$-Assembly Sequence starting from an assembly $A_0$ is a sequence of length $o$ of free assemblies such that for any $i < j \leq o$, there is an injective embedding $\eta_{i \to j}: A_i \to A_j$ with:
  \begin{itemize}
  \item for $i < o$, $(A_i, A_{i+1}, \eta_{i \to i+1})$ is a free attachment,
  \item for $i < j < k \leq o$, $\eta_{j \to k} \circ \eta_{i \to j}  = \eta_{i \to k}$
  \end{itemize}

  The $i$-th production of $\alpha$ is $A_i$, and its $i$-th attachment is $\eta_{i \to i+1}$.
  
  For an $o$-Assembly Sequence $\alpha$, the notation $\lim \alpha = A_o$ is consistent with the case of usual Assembly Sequences, i.e. $\omega$-Assembly Sequences.
\end{definition}

These sequences correspond to the intuitive generalization of attachment sequences to ``times larger than infinity''. In particular, any tile attached through an ordinal attachment sequence reaches the starting assembly of the sequence through a finite number of attachments. Indeed, for a tile to have been added through an ordinal assembly sequence $\alpha$ at time $t$, the neighbors to which it attaches must have been attached at some time $t' < t$. Since these times are ordinals, such an decreasing sequence of times reaches $0$ in a finite number of steps. In particular, these ordinal assembly sequences produce the same productions as the usual $\omega$-assembly sequences.

This remark formalizes thus:
\begin{remark}
  \label{lem:pos-finite-time}
  Let $\Ss$ be an assembly system, $o$ a countable ordinal, $\alpha$ an $o$-Assembly Sequence, $k < o$ and $t@Z$ the $k$-th attachment of $\alpha$. There is a sequence $\alpha'$ of length $o'$ and a bijection $\iota: o \to o'$ such that:
  \begin{itemize}
  \item for all $t < o'$, the attachments $\alpha_t$ and $\alpha'_{\iota(t)}$ are the same,
  \item $\iota(k)$ is finite.
  \end{itemize}
\end{remark}

\begin{lemma}
  \label{lem:ord-omega}
  Let $\Ss$ be an assembly system, $o$ a countable ordinal, and $\alpha$ an $o$-Assembly Sequence. Then if $o$ is infinite, there is a $\omega$-Assembly Sequence $\alpha'$ such that $\lim \alpha' = \lim \alpha$.
\end{lemma}

\begin{proof}
  The proof goes by transfinite induction.

  If $o = \omega$, then $\alpha$ itself satisfies the conclusion of the lemma.

  If $o = p + 1$ is a successor ordinal, by induction hypothesis, there is an $\omega$-assembly sequence $\alpha''$ such that $\lim \alpha'' = A_p$, where $A_p$ is the $p$-th production of $\alpha$. Let $t@Z$ be the $p$-th attachment of $\alpha$;  there is an integer $k$ such that at time $k$, all the origin vertices of the arcs in $Z$ are attached in $\alpha''$. The attachment sequence $\alpha' = (\alpha''_0, \ldots, \alpha''_k, t@Z, \alpha''_{k+1}, \ldots)$ satisfies the conclusion of the lemma.

  If $o$ is a limit ordinal, there is a sequence $(\alpha'^i)_{i < o}$ of $\omega$-assembly sequences such that for each $i$, $\lim \alpha'^i = A_i$. Let $A'^i_j$ be the $j$-th production of $\alpha'^i$. Since $o$ is countable, it has cofinality $\omega$; pick an increasing sequence $\epsilon: \omega \to o$ such that $\sup_{i < \omega} \epsilon(i) = o$. Let $\alpha'$ be the enumeration of the set $\{A'^{\epsilon(i)}_j | j \leq i < \omega \}$ ordered lexicographically according to $(i, j)$. This sequence $\alpha'$ is an attachment sequence, and it satisfies $\lim \alpha' = \bigcup_{i < o} (A'^{\epsilon(i)}_i) = \lim \alpha$.
\end{proof}


\subsection{Holes and Fizziness}
\label{sec:fizziness}

The Tree Pump lemma is about tree-like assemblies. Since trees are acyclic graphs, the \emph{holes} of the assemblies play an important part in characterizing how these tree-like assemblies behave. Fizziness is the tendency of assembly sequences to create a lot of holes.

\begin{definition}[Fizziness]
  Let $A, A'$ be two free assemblies with $\support[A] \subset \support[A']$, the \emph{fizziness} $\fizziness(A, A')$ is the number of holes of $A'$ whose perimeter is not contained in $A$.
  
  Let $\alpha = (A_0 \rightarrow A_1 \rightarrow \ldots)$ be an Assembly Sequence in $\FreeSeqs{\Ss}$. The fizziness of $\alpha$, noted $\fizziness(\alpha)$ is the sequence $i \mapsto \fizziness(A_i, A_i+1)$.

  A sequence $\alpha$ is more fizzy than $\beta$, written $\alpha \moreFizzy \beta$ if $\fizziness(\alpha) > \fizziness(\beta)$ lexicographically.
\end{definition}

\begin{lemma}[Maximal Fizziness]
  Let $\Ss$ be a seeded assembly system. Assume $\seed[\Ss]$ has a finite number of non-null glues on its external edges. Then there is an $\omega$-Assembly Sequence $\alpha_{\max{}} \in \FreeSeqs{\Ss}$ with maximal fizziness among $\omega$-Assembly Sequences.
\end{lemma}

\begin{proof}
  Using K\H{o}nig's Lemma.
\end{proof}

Note that this lemma does not hold for $o$-assembly sequences with $o > \omega$. Indeed, after time $\omega$, there might be an infinity of possible attachments, thwarting K\H{o}nig's Lemma.

\begin{lemma}[Fizziness-Increase of Embeddings]
  Let $\Ss$ be a seeded TAS, and $\alpha \in \FreeSeqs{\Ss}$. Let $G = \support{(\lim \alpha)}$, $G'$ an assembly support, and $e: G \to G'$ an embedding.

  Then:
  \begin{enumerate}
  \item \label{lbl:embed_correct} $e(\alpha)$ is a free assembly sequence,
  \item \label{lbl:embed_heavier} $e(\alpha) \moreFizzyEq \alpha$
  \item \label{lbl:embed_same_weight} $e(\alpha) \moreFizzy \alpha$ if $e$ is not injective on $\dom(\lim \alpha)$.
  \end{enumerate}
\end{lemma}

\begin{proof}
  Let $\alpha = (t_i@z_i)_i$ be an assembly sequence in $\FreeSeqs{\Ss}$.

  For \cref{lbl:embed_correct}, the attachments of $e(\alpha)$ are valid by definition of $k$. They are stable since for each attachment, the edges adjacent to $z_i$ which make this attachment stable are preserved by $e$.

  For \cref{lbl:embed_heavier}, for each assembly $A$ of $\alpha$ which is mapped injectively, each hole of $\alpha$ is mapped by $e$ to a hole of $e(A)$ with the same labels on its border.

  For \cref{lbl:embed_same_weight}, if $e$ is not injective on $\dom(\lim \alpha)$, then there are $i < k$ such that $e(z_i) = e(z_k)$. Since $\alpha$ is a valid assembly sequence, $z_i \neq z_k$, there are two different paths from the seed to $e(z_i)$. From these two paths, it is possible to construct a hole in $\dom(e(\alpha))$ which is not a hole in $\dom(\alpha)$. Hence, from the first attachment where this hole appears, the assemblies of $e(\alpha)$ are more fizzy than the corresponding ones in $\alpha$.
\end{proof}

\begin{lemma}[Maximal Sequences are Flat]
  Let $\Ss$ be a seeded TAS, $S$ an assembly of $\Ss$ and $X \subset \FreeSeqs{\Ss, S}$ be a set of Free, Ordinal Assembly Sequences such that for any $\alpha \in X$, there is a $\alpha' \in X$ such that $\alpha' \moreFizzyEq \ezEmbed(\alpha)$.
  
  Assume $\alpha$ is a free assembly sequence with maximal fizziness within $X$. Then $\alpha$ embeds into $\mathbb{Z}^2$.
\end{lemma}

\begin{proof}
  If $\alpha$ does not embed into $\mathbb{Z}^2$ by $\ezEmbed$, then $\ezEmbed(\alpha)$ is fizzier than $\alpha$ and thus $\alpha$ cannot be maximal in $X$.
\end{proof}

  \begin{figure}
    \centering
    \tikzexternaldisable

\tikzset{
  seed/.pic={
    \begin{scope}[scale=#1]
      \filldraw[fill={pink}] (0,0) rectangle +(1,1) node[pos=.5, font=\footnotesize] {\tikzpictext};
      \pic[scale=#1, rotate=180] at (0.5, 1) {singleGlue};
      \pic[scale=#1, rotate=90] at (1, 0.5) {singleGlue};
    \end{scope}
  }
}

\tikzset{
  tileA/.pic={
    \begin{scope}[scale=#1]
      \filldraw[fill={blue!50}] (0,0) rectangle +(1,1) node[pos=.5, font=\footnotesize] {\tikzpictext};
      \pic[scale=#1, rotate=-90] at (0, 0.5) {singleGlue};
      \pic[scale=#1, rotate=90] at (1, 0.5) {singleGlue};
    \end{scope}
  }
}

\tikzset{
  tileB/.pic={
    \begin{scope}[scale=#1]
      \filldraw[fill={yellow!50}] (0,0) rectangle +(1,1) node[pos=.5, font=\footnotesize] {\tikzpictext};
      \pic[scale=#1, rotate=-90] at (0, 0.5) {singleGlue};
      \pic[scale=#1, rotate=90] at (1, 0.5) {singleGlue};
      \pic[scale=#1] at (0.5, 0) {singleGlue};
    \end{scope}
  }
}

\tikzset{
  tileC/.pic={
    \begin{scope}[scale=#1]
      \filldraw[fill={green!50}] (0,0) rectangle +(1,1) node[pos=.5, font=\footnotesize] {\tikzpictext};
      \pic[scale=#1, rotate=180] at (0.5, 1) {singleGlue};
      \pic[scale=#1] at (0.5, 0) {singleGlue};
    \end{scope}
  }
}

\tikzset{
  tileD/.pic={
    \begin{scope}[scale=#1]
      \filldraw[fill={gray!50}] (0,0) rectangle +(1,1) node[pos=.5, font=\footnotesize] {\tikzpictext};
      \pic[scale=#1, rotate=-90] at (0, 0.5) {singleGlue};
      \pic[scale=#1, rotate=180] at (0.5, 1) {singleGlue};
    \end{scope}
  }
}

\tikzset{
  tileE/.pic={
    \begin{scope}[scale=#1]
      \filldraw[fill={purple!50}] (0,0) rectangle +(1,1) node[pos=.5, font=\footnotesize] {\tikzpictext};
      \pic[scale=#1, rotate=90] at (1, 0.5) {singleGlue};
      \pic[scale=#1, rotate=180] at (0.5, 1) {singleGlue};
    \end{scope}
  }
}

\begin{tikzpicture}[scale=.5]

  \draw[-{Stealth[]}, thick] (3,3) -- (3,5) node[pos=.5, right] {$\vec{d}$};
  \node at (8,4) {$\tau = 1$};            
  \tikzmath {
    int \t;
    for \t in {0,...,65} {
      coordinate \z;
      int \tmod;
      \tmod = mod(\t, 9);
      int \tdiv;
      \tdiv = \t / 9;
      let \tile = tileA;
      if \tmod == 0 then {
        let \tile = tileC;
      };
      if \tmod == 1 then {
        let \tile = tileB;
      };
      if \tmod == 2 then {
        let \tile = tileA;
      };
      if \tmod == 3 then {
        let \tile = tileB;
      };
      if \tmod == 4 then {
        let \tile = tileC;
      };
      if \tmod == 5 then {
        let \tile = tileD;
      };
      if \tmod == 7 then {
        let \tile = tileE;
      };
      if \tmod == 8 then {
        let \tile = tileC;
      };
      if \t == 0 then {
        let \tile = seed;
      };
      if \tmod == 0 then {
        int \y;
        \y = \t / 9;
        \z = (0, \y);
      } else {
        if \tmod < 4 then {
          \z = (\tmod + 3 * \tdiv,0);
        } else {
          if \tmod < 6 then {
            \z = (3 + 3 * \tdiv,-\tmod + 3);
          } else {
            if \tmod < 8 then {
              \z = (3 * \tdiv - \tmod + 8,-2);
            } else {
              \z = (3 * \tdiv + 1,-1);
            };
          };
        };
      };
      {
        \begin{scope}[shift={(\z)}]
          \pic["\t"] {\tile=.5};
        \end{scope}
      };
    };
  }
  \node[rotate=90] at (0.5, 9) {$\ldots$};
  \node at (25, 0.5) {$\ldots$};
  \node[left] at (0,0) {$\alpha' = $};    
\end{tikzpicture}\\
\begin{tikzpicture}
  \node[rotate=90] {$\moreFizzy$};
\end{tikzpicture}\\
\begin{tikzpicture}[scale=.5]
  \foreach \t in {0,...,65} {
    \tikzmath {
      coordinate \z;
      int \tmod;
      \tmod = mod(\t, 9);
      int \tdiv;
      \tdiv = \t / 9;
      int \thuit;
      if \t < 9 then {
        \thuit = \t;
      }
      else {
        if \t < 17 then {
          \thuit = \t - 1;
        } else {
          \thuit = \t - 2 * floor(\t / 9) + 1;
        };
      };
      let \tile = tileA;
      if \tmod == 0 then {
        let \tile = tileC;
      };
      if \tmod == 1 then {
        let \tile = tileB;
      };
      if \tmod == 2 then {
        let \tile = tileA;
      };
      if \tmod == 3 then {
        let \tile = tileB;
      };
      if \tmod == 4 then {
        let \tile = tileC;
      };
      if \tmod == 5 then {
        let \tile = tileD;
      };
      if \tmod == 7 then {
        let \tile = tileE;
      };
      if \tmod == 8 then {
        let \tile = tileC;
      };
      if \t == 0 then {
        let \tile = seed;
      };
      if \tmod == 0 then {
        int \y;
        \y = \t / 9;
        \z = (0, \y);
      } else {
        if \tmod < 4 then {
          \z = (\tmod + 3 * \tdiv,0);
        } else {
          if \tmod < 6 then {
             \z = (3 + 3 * \tdiv,-\tmod + 3);
           } else {
             if \tmod < 8 then {
               \z = (3 * \tdiv - \tmod + 8,-2);
             } else {
               \z = (\tmod + 3 * \tdiv,-\tmod);
             };
           };
        };
      };
      if \tmod == 8 then {} else {
        if \tmod == 0 then {} else {
          {
            \filldraw[fill={gray!50}] (\z) rectangle +(1,1) node[pos=.5, font=\footnotesize] {\thuit};
            \pic["\thuit"] at (\z) {\tile=.5};
          };
        };
      };
    }
  }
  \pic["0"] at (0,0) {seed=.5};
  \pic["8"] at (1,-1) {tileC=.5};

  \node at (25, 0.5) {$\ldots$};
  \node[left] at (0,0) {$\alpha = $};
\end{tikzpicture}

\begin{tikzpicture}[scale=.5]

  \tikzmath {
    int \t;
    for \t in {0,...,65} {
      coordinate \z;
      int \tmod;
      \tmod = mod(\t, 9);
      int \tdiv;
      \tdiv = \t / 9;
      int \thuit;
      if \t < 9 then {
        \thuit = \t;
      }
      else {
        if \t < 17 then {
          \thuit = \t - 1;
        } else {
          \thuit = \t - 2 * floor(\t / 9) + 1;
        };
      };
      let \tile = tileA;
      if \tmod == 0 then {
        let \tile = tileC;
      };
      if \tmod == 1 then {
        let \tile = tileB;
      };
      if \tmod == 2 then {
        let \tile = tileA;
      };
      if \tmod == 3 then {
        let \tile = tileB;
      };
      if \tmod == 4 then {
        let \tile = tileC;
      };
      if \tmod == 5 then {
        let \tile = tileD;
      };
      if \tmod == 7 then {
        let \tile = tileE;
      };
      if \tmod == 8 then {
        let \tile = tileC;
      };
      if \t == 0 then {
        let \tile = seed;
      };
      let \name = $.\thuit$;
      if \tmod == 0 then {
        int \y;
        \y = \t / 9;
        \z = (0, \y);
        if \t > 0 then {
          \y = \y - 1;
          if \y == 0 then {
            let \name = $1.0$;
          } else {
            let \name = $1.\y$;
          };
        };
      } else {
        if \tmod < 4 then {
          \z = (\tmod + 3 * \tdiv,0);
        } else {
          if \tmod < 6 then {
            \z = (3 + 3 * \tdiv,-\tmod + 3);
          } else {
            if \tmod < 8 then {
              \z = (3 * \tdiv - \tmod + 8,-2);
            } else {
              \z = (3 * \tdiv + 1,-1);
              if \t > 8 then {
                \y = \tdiv - 1;
                let \name = $2.\y$;
              };
            };
          };
        };
      };
      {
        \begin{scope}[shift={(\z)}]
          \pic["\name"] {\tile=.5};
        \end{scope}
      };
    };
  }
  \node[rotate=90] at (0.5, 9) {$\ldots$};
  \node at (25, 0.5) {$\ldots$};
  \node[left] at (0,0) {$\alpha'' = $};    
\end{tikzpicture}\\

\tikzexternalenable
    \caption{An assembly sequence $\alpha'$ of some temperature 1 aTAM, which yields a terminal production $F = \lim \alpha'$. Below a sequence $\alpha \moreFizzy \alpha'$ which yields a smaller production $P$ which is bounded in direction $\vec{d}$ (vertically). The assembly sequence $\alpha$ is fizzier since by step $50$, it just closed its seventh hole while $\alpha'$ is lagging at only 5 holes. Because $\alpha$ skips any unprofitable attachment, $\lim \alpha$ is very much not terminal: its seed (tile number $0$) as well as every tile attached at a time of the form $16 + 7k$ has an attachable yet unfilled position. It can thus be extended into an ordinal assembly sequence $\alpha''$ which does reach $\lim \alpha$, but in time $3 \omega$. The attachment times of the form $i.j$ in $\alpha''$ should be read as $i \omega + j$. At time $\omega$, $\alpha''$ has assembled $\lim \alpha$; at time $2 \omega$, it has added the upwards path, and at time $3 \omega$, it has added the last tile to each hole of $\lim \alpha$. }
    \label{fig:recalcitrant}
  \end{figure}
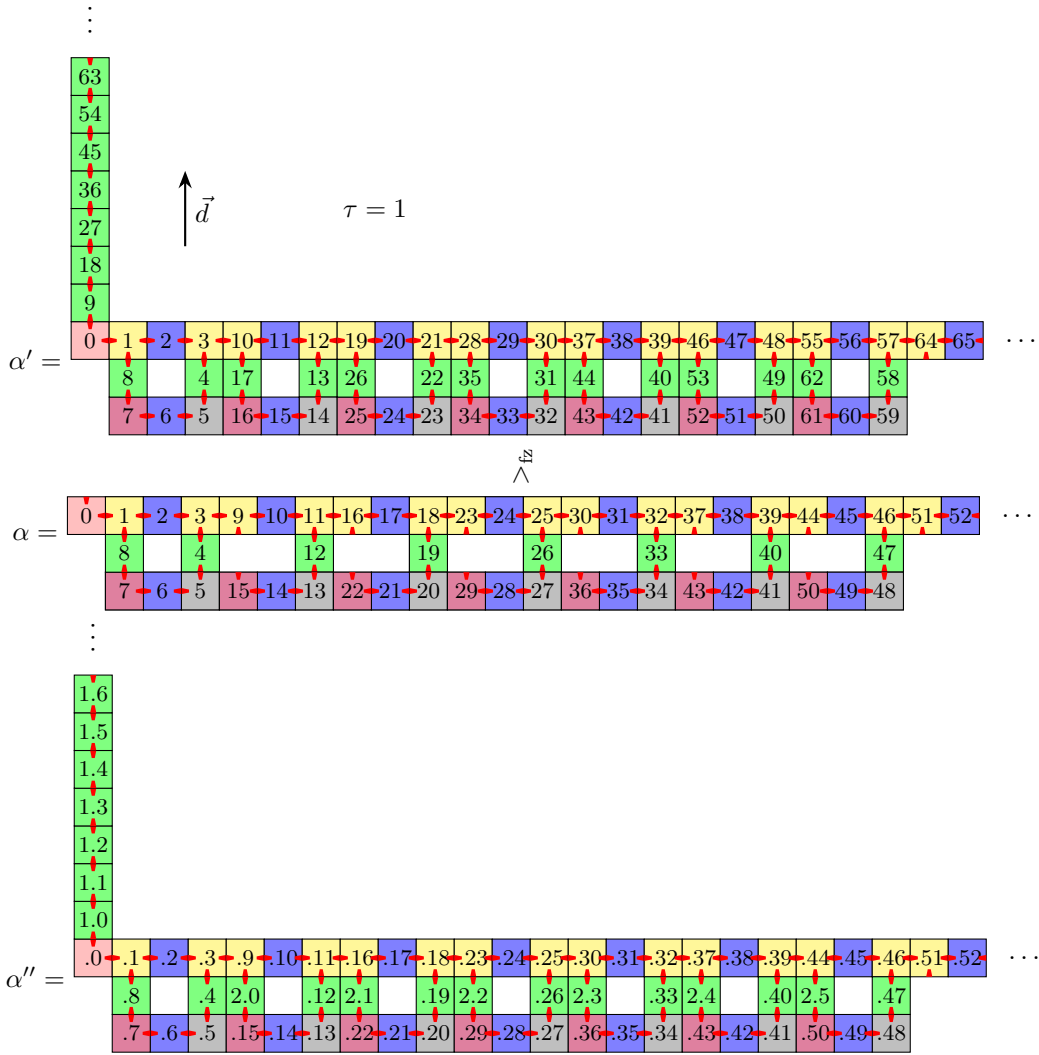

\subsection{The Window Movie Lemma}

The Window Movie Lemma~\cite{meunier_intrinsic_2014} applies to free assembly sequences as well as to $\mathbb{Z}^2$ assembly sequences. In order to account for holes, movies need to record not only the glues appearing on either side of the window, but also the paths between edges of the window through either side. 

\begin{definition}[Window, Frame, Movie]
  Let $\alpha = (A_0, A_1, \ldots)$ be an assembly sequence of $\FreeSeqs{\Ss}$. A \emph{window} $W$ of $\alpha$ is a cut of $\support[(\lim \alpha)]$ separating it into two connected components $\operatorname{Near}(W)$ and $\operatorname{Far}(W)$. 

  The width $w(W)$ of a window is the maximum perspective difference between two vertices along the window.

  Let $\operatorname{Arcs}(W)$ be the set of arcs (either normal or external) of $\support[(\lim \alpha)]$ through $W$.

  The \emph{frame} $f$ on $W$ at time $k$ records for each arc $a = (v, d) \in \operatorname{Arcs}(W)$:
  \begin{itemize}
  \item $\present[f(a)]$: whether $v \in \support[A_k]$, and if so,
  \item $\glue[f(a)] = \tiles[A_k](v)(d)$, as well as
  \item $\paths[f(a)]$ the  arc of $\operatorname{Arcs}(W)$ which is reachable from $a$ through an exterior path which does not cross $W$, if there is one.
  \end{itemize}
  
  The \emph{movie} associated with $W$ is the ordered set $\movie(W)$ of distinct frames appearing on $W$.
\end{definition}


\begin{figure}
  \centering
  \tikzset{
  extPath/.style={-{Stealth[scale=.75]}}
}

\tikzset{
  doubleGlue/.pic={
    \begin{scope}[scale=.2]
      \fill[fill=black] (-.7, -1) rectangle (-.2, 0);
      \fill[fill=black] (.2, -1) rectangle (.7, 0);
    \end{scope}
  }
}

\tikzset{
  theTileA/.pic={
    \clip (0, 0) rectangle (1,1);
    \filldraw[fill=gray!25!white] (0, 0) rectangle (1, 1) node[anchor=center, pos=.5] {#1};
    \pic at (.5, 1) {doubleGlue};
    \pic[rotate=-90] at (1, .5) {doubleGlue};
  }
}

\tikzset{
  theTileB/.pic={
    \clip (0, 0) rectangle (1,1);
    \filldraw[fill=gray!25!white] (0, 0) rectangle (1, 1) node[anchor=center, pos=.5] {#1};
    \pic at (.5, 1) {doubleGlue};
    \pic[rotate=180] at (.5, 0) {doubleGlue};
  }
}

\tikzset{
  theTileC/.pic={
    \clip (0, 0) rectangle (1,1);
    \filldraw[fill=gray!25!white] (0, 0) rectangle (1, 1) node[anchor=center, pos=.5] {#1};
    \pic[rotate=-90] at (1, .5) {doubleGlue};
    \pic[rotate=90] at (0, .5) {doubleGlue};
  }
}

\tikzset{
  theTileD/.pic={
    \clip (0, 0) rectangle (1,1);
    \filldraw[fill=gray!25!white] (0, 0) rectangle (1, 1) node[anchor=center, pos=.5] {#1};
    \pic[rotate=-90] at (1, .5) {doubleGlue};
    \pic[rotate=180] at (.5, 0) {doubleGlue};
  }
}

\tikzset{
  theTileE/.pic={
    \clip (0, 0) rectangle (1,1);
    \filldraw[fill=gray!25!white] (0, 0) rectangle (1, 1) node[anchor=center, pos=.5] {#1};
    \pic[rotate=90] at (0, .5) {doubleGlue};
    \pic[rotate=180] at (.5, 0) {doubleGlue};
  }
}

\tikzset{
  theTileF/.pic={
    \clip (0, 0) rectangle (1,1);
    \filldraw[fill=gray!25!white] (0, 0) rectangle (1, 1) node[anchor=center, pos=.5] {#1};
    \pic at (.5, 1) {doubleGlue};
    \pic[rotate=90] at (0, .5) {doubleGlue};
  }
}

\begin{tikzpicture}
  \foreach \pos / \tile [count=\i from 0] in {(0,0)/A, (1,0)/C, (0,1)/B, (0,2)/B, (0,3)/D, (1,3)/C, (2,3)/E, (2,2)/B, (2,1)/B, (2,0)/F}
  {
    \begin{scope}[shift={(0,0)}]
      \pic at \pos {theTile\tile=\i};
    \end{scope}
    
    \begin{scope}[shift={(-4, -1.5)}]
      \foreach \frame in {\i, ..., 9} {
        \begin{scope}[shift={(4*.30*\frame, 0)}, transform shape, scale=.30]
          \pic at \pos {theTile\tile={}};
        \end{scope}
      }
    \end{scope}
  }
  \path[draw, red, ultra thick] (-.2, 2) node[left] {$W$}  -- (1.2, 2) (1.8, 2) -- (3.2, 2);
  \path[draw, red, dashed, -{Stealth[]}] (1.2, 2) to[bend left] node[above] {$(1, 0)$}(1.8, 2);

  \node[above left] at (-4.5, -1.5) {Assemblies};
  \node[left] at (-4.5, -2) {Frames};
  \node[left] at (-4.5, -3.5) {Movie};

  \tikzset{frame0/.pic={
      \fill[yellow!50!white] (-.05, -.25) rectangle (0.8, .25);
      \fill[yellow!50!white] (-.05, -.15) rectangle (0.8, .15);
      \path[draw, red, thick] (0, 0) -- (0.25, 0) (0.5, 0) -- (0.75, 0);
    }
  }

  \tikzset{frame1/.pic={\pic {frame0};}}
  
  \tikzset{frame2/.pic={
      \fill[yellow!50!white] (-.05, -.25) rectangle (0.8, .25);
      \pic[scale=.5] at (0.125,0) {doubleGlue};
      \draw[extPath] (0,0) to[out=-135, in=180] (.125, -.15) to [out=0, in=-45] (.25,0);
      \path[draw, red, thick] (0, 0) -- (0.25, 0) (0.5, 0) -- (0.75, 0);
    }}

  \tikzset{frame3/.pic={
      \fill[yellow!50!white] (-.05, -.25) rectangle (0.8, .25);
      \draw[extPath] (0,0) to[out=-135, in=180] (.125, -.15) to [out=0, in=-45] (.25,0);
      \pic[scale=.5] at (0.125, 0) {doubleGlue};
      \pic[scale=.5, rotate=180] at (0.125, 0) {doubleGlue};
      \draw[extPath, rotate around={180:(0.125,0)}] (0,0) to[out=-135, in=180] (.125, -.15) to [out=0, in=-45] (.25,0);
      \path[draw, red, thick] (0, 0) -- (0.25, 0) (0.5, 0) -- (0.75, 0);
    }}

  \tikzset{frame4/.pic={\pic {frame3};}}
  \tikzset{frame5/.pic={\pic {frame3};}}
  \tikzset{frame6/.pic={\pic {frame3};}}

  \tikzset{frame7/.pic={
      \fill[yellow!50!white] (-.05, -.25) rectangle (0.8, .25);
    \pic[scale=.5] at (0.125, 0) {doubleGlue};
    
    \draw[extPath] (0.75, 0) to[out=90, in=90] (0, 0);
    \pic [scale=.5,rotate=180] at (0.125, 0) {doubleGlue};
    \draw[extPath] (0,0) to[out=-135, in=180] (.125, -.15) to [out=0, in=-45] (.25,0);
    \draw[extPath] (.25, 0) to [out=45, in=135] (.5, 0);
    \pic [scale=.5,rotate=180] at (0.625, 0) {doubleGlue};
    \path[draw, red, thick] (0, 0) -- (0.25, 0) (0.5, 0) -- (0.75, 0);
  }}

\tikzset{frame8/.pic={
    \fill[yellow!50!white] (-.05, -.25) rectangle (0.8, .25);

    \pic[scale=.5] at (0.125, 0) {doubleGlue};
    \pic [scale=.5,rotate=180] at (0.125, 0) {doubleGlue};
    \pic[scale=.5] at (0.625, 0) {doubleGlue};
    \pic [scale=.5,rotate=180] at (0.625, 0) {doubleGlue};
    \path[draw, red, thick] (0, 0) -- (0.25, 0) (0.5, 0) -- (0.75, 0);

    \draw[extPath] (0.75, 0) to[out=90, in=90] (0, 0);
    \draw[extPath] (0,0) to[out=-135, in=180] +(.125, -.15) to [out=0, in=-45] (.25,0);
    \draw[extPath] (.25, 0) to [out=45, in=135] (.5, 0);
    \draw[extPath] (0.5, 0) to[out=-135, in=180] +(.125, -.15) to [out=0, in=-45] (0.75,0);
  }}
  
  \tikzset{frame9/.pic={
    \fill[yellow!50!white] (-.05, -.25) rectangle (0.8, .25);
    \pic[scale=.5] at (0.125, 0) {doubleGlue};
    \pic [scale=.5,rotate=180] at (0.125, 0) {doubleGlue};
    \pic[scale=.5] at (0.625, 0) {doubleGlue};
    \pic [scale=.5,rotate=180] at (0.625, 0) {doubleGlue};
    \path[draw, red, thick] (0, 0) -- (0.25, 0) (0.5, 0) -- (0.75, 0);

    \draw[extPath] (0.75, 0) to[out=90, in=90] (0, 0);
    \draw[extPath] (.25, 0) to [out=45, in=135] (.5, 0);
    \begin{scope}[rotate around={180:(.375,0)}]
      \draw[extPath] (0.75, 0) to[out=90, in=90] (0, 0);
      \draw[extPath] (.25, 0) to [out=45, in=135] (.5, 0);
  \end{scope}
  }}
  
  \begin{scope}[shift={(-4,-2)}]
    \foreach \i in {0, ..., 9} \pic[shift={(4*.30*\i,0)}] {frame\i};
  \end{scope}

  \begin{scope}[shift={(-3.9,-3.5)}]
    \fill[black] (-0.3,-1) rectangle (11.8,1);

    \begin{scope}[shift={(-.35,0)}]
      \foreach \x in {0, .4, ..., 12} {
        \fill[fill=white] (\x, -.65) rectangle +(.2, -.2);
        \fill[fill=white] (\x, .85) rectangle +(.2, -.2);
      }
    \end{scope}
    \foreach \i [count=\c from 0] in {0, 2, 3, 7, 8, 9}
    \pic[shift={(2*\c,0)}, scale=2] {frame\i};
  \end{scope}
\end{tikzpicture}

  \caption{Example assembly sequence $\alpha$ and window $W$. On the main picture, the labels of the tiles of $\lim \alpha$ are the order in which they attach. The small pictures are the successive assemblies of $\alpha$, their associated frames and the obtained movie. The movie is made up of the sequence of the unique frames, in order of apparition.}
  \label{fig:glue_movie_example}
\end{figure}
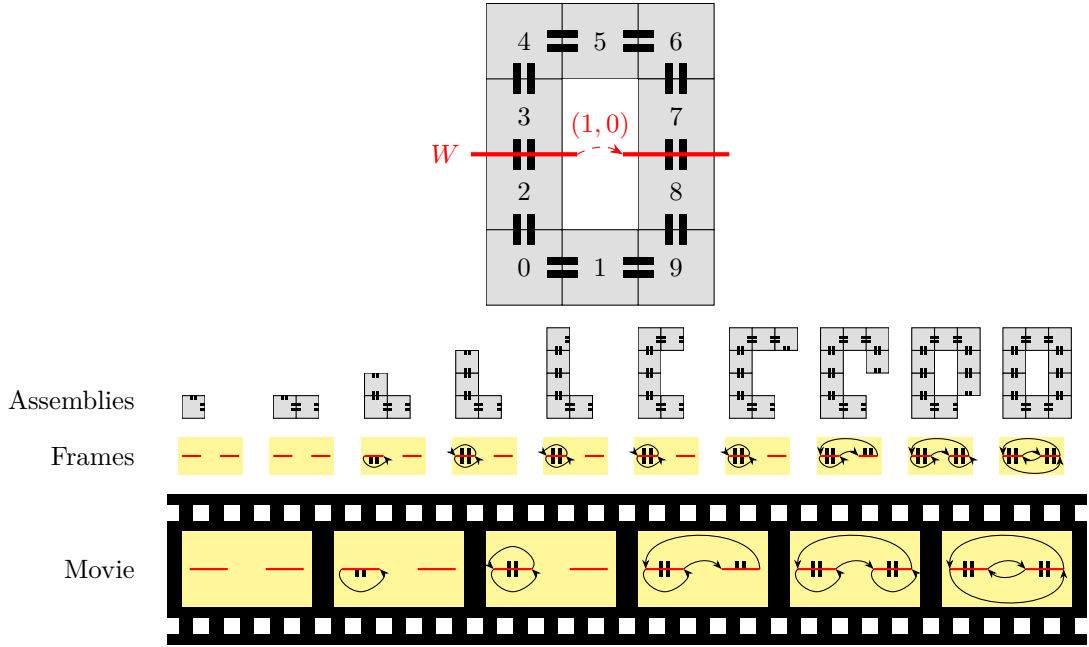

\cref{fig:glue_movie_example} gives an example free assembly sequence $\alpha$ with a window $W$. The arrows between the edges on either side of $W$ represent the relation ``$\paths$''.

\begin{lemma}[Window Movie Lemma]
  \label{lem:wml}
  Let $\alpha \in \FreeSeqs{\Ss, \sigma_A}$ and $\beta \in \FreeSeqs{\Ss, \sigma_B}$ two assembly sequences. Assume there are two windows $A$ in $\alpha$ and $B$ in $\beta$ and a translation $\tau$ such that $\tau(A) = B$, $\movie(A)$ is the same as $\movie(B)$ up to translation by $\tau$, and lastly $\tau$ can be extended to a translation from $\far(A)$ to $\far(B)$ which maps $\sigma_B \cap \operatorname{Far}(B)$ to $\sigma_A \cap \operatorname{Far}(A)$.

  Let $\alpha^\dagger = \alpha_{|\far(A)} \in \FreeSeqs{\Ss, \lim \alpha \cap \near(A)}$ be the sequence of attachments of $\alpha$ within $\far(A)$, and likewise $\beta^\dagger = \alpha_{|\far(B)} \in \FreeSeqs{\Ss, \lim \beta \cap \near(B)}$ be the sequence of attachments of $\beta$ witihn $\far(B)$.

  Let $\alpha'$ be the candidate assembly sequence consisting of $\alpha$ where for every $k$, the attachment $\alpha^\dagger_k$ is replaced with $\beta^\dagger_k$.
  
  Then there is a window $A'$ on $\alpha'$ and two translations $\tau_N: \operatorname{Near}(A') \to \operatorname{Near}(A)$ and $\tau_F: \operatorname{Far}(A') \to \operatorname{Far}(B)$ such that:
    \begin{equation}
    \begin{cases}
      \forall z  \in \operatorname{Near}(A'),  &\tiles[(\lim \alpha')](z) = \tiles[(\lim \alpha)](\tau_N(z))\\
      \forall z  \in \operatorname{Far}(A'), &\tiles[(\lim \alpha')](z) = \tiles[(\lim \beta)](\tau_F(z))
    \end{cases}\label{eq:transl}
  \end{equation}

  Moreover, if $\beta^\dagger \moreFizzy \alpha^\dagger$, then $\alpha' \moreFizzy \alpha$.
\end{lemma}

\begin{figure}
  \centering
  \input{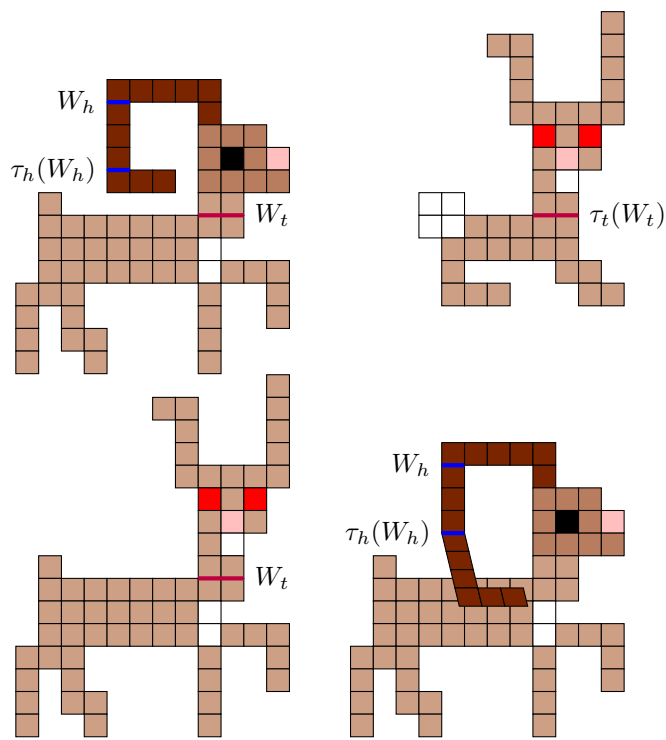}
  \caption{The Window Movie Lemma: consider two assembly sequences $\iota$ (the ibex) and $\beta$ (the bunny), with their productions depicted in the top row. The sequence $\iota$, has two windows $W_t$ and $W_h$. In $W_t$, $\far(W_t)$ is the head of the ibex, $\near(W_t)$ its body. In $W_h$, $\far(W_h)$ is the tip of the horn. Each of them is associated with a translation, $\tau_t$ and $\tau_h$ respectively, which sends the window to another window with the same movie, and satisfying the hypothesis of~\cref{lem:wml}. Applying \cref{lem:wml} in each case gives two new assembly sequences represented on the bottom row. Applying it with $W_t$ yields a fearsome chimera, while its application to $W_h$ yields an ibex with a horn so long it needs to bend in the third dimension to avoid piercing its spine, i.e. avoiding the conflict that would arise were our attachment sequences not free.}
  \label{fig:tpl}
\end{figure}

\begin{proof}
  Define $\alpha'$ as $\alpha$ where for each $k$, the $k$-th attachment $t@Z$ within $\far(A)$ has been replaced by the $k$-th attachment within $\far(B)$ of $\beta$. When doing this, for an attachment $t@Z$ within $\far(B)$, any arc $a = (v, d)$ in $Z$ with $v \in \near(B)$ is replaced by the arc $\tau^{-1}(a)$.

  The translations $\tau_N$ and $\tau_F$ can be defined inductively on the attachments of $\alpha'$, such that \cref{eq:transl} holds.
  
  The proof that $\alpha'$ is a valid free assembly sequence is the same as in the $\mathbb{Z}^2$ case. Note that since $B$ is a cut of $\support[(\lim \beta)]$, any attachment of $\beta$ in $\far(B)$ can be replayed in $\alpha'$ in $\far(A')$ without conflicts since they do not create any adjacency outside $\operatorname{Far}(\beta)$.

  Observe that every attachment in $\operatorname{Far}(B)$ which affects holes in $\beta$ affects the same number of holes in $\alpha'$: holes which are wholly within $\operatorname{Far}(B)$ have had all of their tiles attached in $\alpha'$, and holes which go through $B$ become holes through $A$ because the part within $\operatorname{Near}(B)$ has the same connections in $\operatorname{Near}(A)$ at that frame in the movie.
\end{proof}

The fact that the grafting process of \cref{lem:wml} is increasing in the fizziness of the far part of the assembly sequences implies that when the original assembly sequences have maximal fizziness, so does the chimera sequence $\alpha'$.



\cref{lem:wml} will be most useful in this paper in the case where the cuts $A$ and $B$ are on the same branch of the assembly. In this case, by iterating \cref{lem:wml}, it is possible to get a production with a periodic subassembly.

\begin{corollary}
  \label{lem:celeri_branche}
  Let $\alpha = (A_i)_{i < o}$ be an assembly sequence with two windows $A$ and $B$ satisfying the hypothesis of \cref{lem:wml} and such that $\far(B) \subset \far(A)$. Let $\tau$ be the translation between $A$ and $B$. Then, there is an assembly sequence $\alpha^{\tau}$ such that $F = \lim \alpha^{\tau} \cap \far(A)$ verifies $\tau(F) \subset F$, and within $\near(B)$, $\alpha^{\tau}$ does the same attachments in the same order as $\alpha$.
\end{corollary}

\begin{proof}
  Let $\tau$ be the translation sending $A$ to $B$.
  
  Show by induction that for any $k$, there is a sequence $\alpha^k$ such that the movie on the windows $A$ and $B$ within $\alpha^k$ are the same as in $\alpha$, and for each $z \in \near(B) \cap \far(A)$ and $j \leq k$, $(\lim \alpha^k)(\tau^j(z)) = (\lim \alpha)(z)$, and the attachments done by $\alpha^k$ and $\alpha^j$ within $\bigcup_{i \leq j} \tau^i(\near(B) \cap \far(A))$ are the same.

  For $k = 0$, it suffices to pick $\alpha^0 = \alpha$.
  Assume that $\alpha^k$ satisfies the induction hypothesis. Then \cref{lem:wml} applies, and the assembly sequence it yields, $\alpha^{k+1}$, satisfies the induction hypothesis at rank $k+1$.

  For a pair $(j, k) \in \mathbb{N}$, if for all $i \leq j$, $\alpha^k_i \in \bigcup_{l \leq k} \tau^k(\near(B) \cap \far(A))$, then the $j$ first attachments are unchanging after rank $k$: for all $m \geq k$ and $i \leq j$, $\alpha^m_i = \alpha^k_i$. Let $\alpha^{\vec{p}}$ be the sequence of such unchanging attachments. Let $m$ be the supremum of the $j \in \mathbb{N}$ such that the $j$ first attachments are unchanging after some rank $k$. Define $\alpha^{\vec{p}}$ as the sequence of length $m$ with $\alpha^{\vec{p}}_j = \alpha^k_j$, where $k$ is such that the $j$ first attachments are unchanging.

  It is easy to check that indeed, $\lim \alpha^{\vec{p}} \cap \far(A)$ has period $\vec{p}$: any attachment in $\alpha^{\vec{p}}$ within $\far(A)$ gets picked up by subsequent applications of \cref{lem:wml} which translate it by $\vec{p}$.
\end{proof}

In this situation with a branch and two cuts with the same movie, it is also possible to cut assembly sequences short, so that they can do their business in $\near(A)$ without needing to mess with $\far(B)$.

\begin{corollary}
  \label{lem:couic-couic}
  Let $\alpha$ be an assembly sequence with two windows $A$ and $B$ satisfying the hypothesis of \cref{lem:wml} and such that $\far(B) \subset \far(A)$. Then there is an assembly sequence $\alpha'$ which does the same attachments as $\alpha$ within $\near(A)$, but has no attachment in $\far(B)$.
\end{corollary}

\begin{proof}
  Let $j$ be the last time that $\alpha$ does an attachment in $\near(A)$ next to $A$, i.e. the date of the last frame in the movie of $A$ where a tile is attached on the near side of the window. By \cref{lem:pos-finite-time}, up to a reordering of $\alpha$, $j$ is finite.

    Let $\alpha^0$ be the prefix of length $j$ of $\alpha$. For each $k$, if $\alpha^k$ has any attachment in $\far(B)$, it is possible to use \cref{lem:wml} between $B$ and $A$ to obtain a sequence $\alpha^{k+1}$ with fewer attachments within $\far(B)$. Hence, there is a finite $k$ such that $\alpha^k$ has no attachments in $\far(B)$, and $\lim \alpha^k \cap \near(A) = \lim \alpha^0 \cap \near(A)$. Since by time $j$, the movie on $A$ is over, the rest of the attachments of $\alpha$ which take place in $\near(A)$ can be replayed after $\alpha^k$.
\end{proof}

\subsection{The Tree Pump Lemma}
\label{sec:tree_pump_tree_pump}

After all these considerations about the fantastic properties of Ordinal Free Assembly Sequences, it is now time to come back to Earth, or rather $\mathbb{Z}^2$. The object is now to prove \cref{lem:tree_pump}, for which an investigation of the properties systems whose $\mathbb{Z}^2$-assembly sequences do not circle large squares is in order.

The statement of \cref{lem:tree_pump} is given again, recall that it deals with the $\mathbb{Z}^2$-productions of the aTAM system $\Ss$, hence in its statement, $\TerminalProds{\Ss}$, $C_m[\Ss]$, $B_{F(n,m),\vec{d}}[\Ss]$ and $P_{\vec{d}}[\Ss]$ are sets of $\mathbb{Z}^2$-assemblies.

\treepump*

When considering the statement of \cref{lem:tree_pump}, the holes of the $\mathbb{Z}^2$-productions of $\Ss$ might as well be filled in. Hence, the definition of their fill-in, illustrated on \cref{fig:treedec}.

\begin{definition}[Fill-in]
  For a subgraph $D \subset \mathbb{Z}^2$, define the \emph{fill-in} $\complete{D}$ as the subgraph of $\mathbb{Z}^2$ induced by the positions $p$ such that there is no infinite path from $p$ in $\mathbb{Z}^2 \setminus D$.
\end{definition}

A $\mathbb{Z}^2$ assembly $A$ which does not circle any square larger than $m \times m$ looks like a tree, as expressed by the notion of \emph{Connected Treewidth}~\cite{DBLP:journals/combinatorica/DiestelM18}. This tree is obtained by grouping the positions of $\support[A]$ in connected sets of size at most $2m$, known as \emph{bags}, organized in a tree in such a way that:
\begin{itemize}
\item any arc of $\support[A]$ is between two positions which appear in the same bag, and
\item for any position $z \in \support[A]$, the set of bags which contain $z$ forms a subtree.
\end{itemize}

This decomposition is represented on \cref{fig:treedec}.

\begin{figure}
  \centering
  \tikzpicturedependsonfile{tikz/treedec}
  \input{tikz/treedec.tex}
  \caption{The domain $D$ of a production $P$, its fill-in $\complete{D}$ and an associated connected tree decomposition of $\treedec{D}$. The blue part of each bag $b$ of $\treedec{D}$ is its intersection $W_b$ with its parent, which disconnects the vertices appearing in its subtree from the rest of $\complete{D}$. The choice of the root of $\treedec{D}$ is arbitrary and does not affect the sets $W_b$, up to renaming.}
  \label{fig:treedec}
\end{figure}

\begin{lemma}
  \label{lem:treedec}
  Let $\Ss$ an aTAM system, $m$ an integer and $A$ a $\mathbb{Z}^2$ assembly of $\Ss$ such that $A \notin C_m[\Ss]$.

  Then $\support[A]$ has Connected Tree Width $2m$.
\end{lemma}

\begin{proof}
Indeed, it has treewidth $m$: otherwise it would contain an $m \times m$ grid as a minor; that minor would have to be realized in $\mathbb{Z}^2$ as a subgraph which would encircle some $m \times m$ square. Moreover, since $\complete{P}$ does not encircle any empty position, it does not have any \emph{geodesic cycle} of length more than $4$.
\end{proof}

The point of using Connected Treewidth rather than the usual treewidth is that thanks to the connectivity of the bags of $\treedec{P}$, the distances in $\treedec{P}$ reflect those in $\complete{P}$:
\begin{itemize}
\item there is a function $r$ such that for any vertex $v \in \complete{P}$, the subtree of the bags of $\treedec{P}$ containing $v$ has a size at most $r(m)$,
\item there are two increasing functions $d_m, D_m$ such that for any vertices $u, v$ at distance $\delta$ in $\complete{P}$, any bags $B_u \ni u$ and $B_v \ni v$ of $\treedec{P}$ are at distance $\Delta$ with $d_m(\delta) \leq \Delta \leq D_m(\delta)$.    
\end{itemize}

As a consequence, if two tiles are far enough in an assembly $A$ and the path between them does not go through the seed, there must be two windows cutting that path which satisfy the hypothesis of \cref{lem:celeri_branche}.

\begin{lemma}
  Let $\Ss$ be an aTAM system with $n$ tiles and $m$ an integer.
  
  Let $\alpha = (A_i)_{i \leq o} \in \FreeSeqs{\Ss, A_0}$ be such that for all $i \leq o$, $\ezEmbed(A_i) \notin C_m[\Ss]$.

  There is a constant $F(n, m)$ such that if $z, z' \in \support[\ezEmbed(\lim \alpha)] \setminus \support[A_0]$ are such that the distance between $z'$ and  $z$ is greater than $F(n, m)$, then there are two windows $A$ and $B$ with $z \in \near(A)$ and $z' \in \far(B)$ which satisfy the hypothesis of \cref{lem:wml}.
\end{lemma}
  
\begin{proof}
  let $W(n, m)$ be the number of movies with $n$ glues on a window of width $2m$. Note that $W(n, m)$ counts the number of arrangements of the edges within the window, the perspective difference between the connected components of the window, as well as the events on the frames of the movie. Let $F(n, m) = d^{-1}_m(W(n,m))$.
  
  Let $P = \support[\lim \alpha]$. By \cref{lem:treedec}, pick a tree decomposition $\treedec{P}$ with connected bags of size at most $2m$. Let $\delta$ be the distance between $z$ and $z'$. If $\delta \leq F(n, m)$, by the pigeonhole principle, there must be two windows $A$ and $B$ between $z$ and $z'$ with the same movie.
\end{proof}

In the case where the assembly embeds into $\mathbb{Z}^2$, one actually controls the direction of the vector between the two windows.

\begin{lemma}
  \label{lem:useless}
  Let $\Ss$ be an aTAM system with $n$ tiles, $m$ an integer and $\vec{d}$ a unit vector.
  
  Let $\alpha = (A_i)_{i \leq o} \in \FreeSeqs{\Ss, A_0}$ be such that for all $i \leq o$, $A_i$ embeds into $\mathbb{Z}^2$.

  There is a constant $F'(n, m)$ such that if $z, z' \in \support[\lim \alpha] \setminus \support[A_0]$ are such that the distance between $(z' - z) \cdot \vec{d} > F(n, m)$, then there are two windows $A$ and $B$ with $z \in \near(A)$ and $z' \in \far(B)$ which satisfy the hypothesis of \cref{lem:wml}, and such that the vector $\vec{p}$ of the translation between $A$ and $B$ satisfies $\vec{p} \cdot \vec{d} > 1$.
\end{lemma}

\begin{proof}
  Pick $F'(n, m) = 2 F(n,m) + 1$. If $(z' - z) \cdot \vec{d} > F'(n, m)$, then in $\treedec{\support[\lim \alpha]}$ there there are two bags, one $b$ containing $z$, the other $b' \ni z'$ such that on the path between $b_z$ and $b_{z'}$, there are $2 F(n,m) + 1$ bags $(b_i)_i$ such that for all $x \in b_i, y \in b_j, (x - y) \cdot \vec{d} \geq i - j$. Thus, there are $i, j$ such $j > i + 1$ and two windows, $A$ between $b_i$ and $b_{i+1}$, and $B$ between $b_j$ and $b_{j+1}$ which have the same movie. Because the windows $A$ and $B$ are separated by at least two elements of $(b_i)$, the translation vector $\vec{p}$ between them satisfies $\vec{p} \cdot \vec{d} > 1$.
\end{proof}
\tikzexternaldisable
Thus, there is a simple argument to prove \cref{lem:tree_pump}: pick \surligne[shaky_foundation]{orange!50!white}{a sequence of maximal fizziness} \surligne{orange!50!white}{which produces a terminal assembly} of $\Ss$ which is neither in $B_{F'(n,m),\vec{d}}[\Ss]$ nor in $C_m[\Ss]$. \cref{lem:useless} yields two windows with the same movie; by applying \cref{lem:celeri_branche} between them, a production in $P_{\vec{d}}[\Ss]$ appears. Alas, the \surligne[collapse]{orange!50!white}{sequence} on which  this argument is founded may simply fail to exist: the $\omega$-sequences of maximal fizziness may fail to reach a terminal production. One could extend them in order to reach a terminal production, but there may not be a sequence of maximal fizziness among these (ordinal) extensions.
\tikz[remember picture] \draw[overlay, orange!50!white, thick, ->] (collapse) to[bend right=15] (shaky_foundation);
\tikzexternalenable


Thus, it is necessary to prevent the pesky tendency ordinal sequences have to try and buy time by spacing out their attachment so as to generate an infinite family of sequences with increasing fizziness. \emph{Straight} sequences are the answer: they do not thave the wiggle room to misbehave.

\begin{definition}
  Let $\Ss$ be an aTAM system, and $\prec$ an arbitrary total order on its sequences. A free, ordinal assembly sequence $\alpha = (A_i)_{i < o} \in \FreeSeqs{\Ss}$ is $\prec$-\emph{straight} if for every pair of window $N, F$ such that $N$ and $F$ have the same movie and $\support[A_0] \subset \near(N)$, $\alpha$ is the smallest sequence for $\prec$ obtained by applying \cref{lem:celeri_branche} to $N$ and $F$ in $\alpha$.

  The order $\prec$ is henceforth fixed and left implicit.
\end{definition}

\begin{lemma}
  \label{lem:finite_choice}
  Let $\Ss$ be an aTAM system and $w$ an integer.

  Let $X \subset \FreeSeqs{\Ss}$ be a set of assembly sequences, and $d$ an integer.
  If for any sequence $\alpha \in X$ and any position $z \in \support[{ \lim \alpha }]$ at distance more than $d$ from $\seed[\Ss]$, there are two windows $A$ and $B$ such that $\support[ {\seed[\Ss]} ] \subset \near(A)$ and $z \in \far(B)$, then there is a finite number of straight sequences in $X$.
\end{lemma}

\begin{proof}
  In this case, each straight sequence in $X$ is determined by what it does in a radius $d$ around $\seed[\Ss]$.
\end{proof}

\begin{lemma}
  let $\Ss$ be an aTAM system and $\alpha$ a straight assembly sequence of $\Ss$. There is a straight assembly sequence $\alpha'$ which extends $\alpha$ such that $\lim \alpha' \in \TerminalProds{\Ss}$.
\end{lemma}

\begin{proof}
  Let $P = \lim \alpha$. If $P \notin \TerminalProds{\Ss}$, then there is an attachment $t@Z$ which is possible in $P$.

  Let $\alpha'$ be $\alpha$ with $t@Z$ appended at its end, and $z$ the position of that last attachment in $\support[{ \lim \alpha' }]$. Assume that $Z$ is such that the distance between $\support[ {\seed[\Ss]} ]$ and $z$ is minimal.

  If there are no pairs of windows $(N, F)$ in $\alpha'$ such that $\far(F) \subset \far(N)$, $\movie(N) = \movie(F)$,  $\support[ {\seed[\Ss]} ] \subset \near(N)$ and $z \in \far(N)$, then $\alpha'$ is straight.

  If there is a pair of windows  $(N, F)$ in $\alpha'$ such that $\far(F) \subset \far(N)$, $\movie(N) = \movie(F)$,  $\support[ {\seed[\Ss]} ] \subset \near(N)$ and $z \in \far(F)$, then by \cref{lem:couic-couic}, there is a place closer to $ {\seed[\Ss]} $ than $z$ where an attachment was possible.

  Lastly, if there is a pair of windows  $(N, F)$ in $\alpha'$ such that $\far(F) \subset \far(N)$, $\movie(N) = \movie(F)$,  $\support[ {\seed[\Ss]} ] \subset \near(N)$ and $z \in \far(N) \cap \near(F)$, then the sequence $\alpha''$ obtained by applying \ref{lem:celeri_branche} in $\alpha'$ is straight and has $\alpha$ as a prefix, since the new attachments (the repetitions of $t@Z$) are done last in each repetition of the movie in $\far(N) \cap \near(F)$.

  This process yields a straight assembly sequence $\alpha'$ which strictly extends $\alpha$. It can be repeated until $\lim \alpha' \in \TerminalProds{\Ss}$.
\end{proof}

\begin{lemma}
  \label{lem:crash_recovery}
  Let $\Ss$ be an aTAM system and $\alpha$ a straight assembly sequence of $\Ss$. Then $\ezEmbed(\alpha)$ is a prefix of a straight sequence $\alpha'$. Moreover, the connected treewidth of $\support[ {\lim \alpha'} ]$ is no greater than that of  $\support[{ (\lim \ezEmbed(\alpha)) }]$.
\end{lemma}

\begin{proof}
  If $\alpha$ embeds cleanly in $\mathbb{Z}^2$, there is nothing to prove since $\ezEmbed(\alpha) = \alpha$.

  Otherwise, $\ezEmbed(\alpha)$ stops because one of its attachment $t@z$ creates a hole which does not exist in $\alpha$. Because $\alpha$ is straight, there cannot be a pair of windows with the same movie separating $\seed[\Ss]$ and $z$. Applying \cref{lem:celeri_branche} in $\ezEmbed(\alpha)$ on each pair of windows with the same movie which are the closest to $\seed[\Ss]$ yields a straight sequence $\alpha'$ extending $\ezEmbed(\alpha)$.

  Fix a tree decomposition of $\support[(\lim \ezEmbed(\alpha))]$. There is a tree decomposition of $\support[ {\lim \alpha'} ]$ where each bag which lies in the far part of a pair of windows with the same movie is a copy of the corresponding bag close to the seed. This tree decomposition has the same width as the decomposition of $\support[(\lim \ezEmbed(\alpha))]$.
\end{proof}

At last, all the ingredients are there to prove \cref{lem:tree_pump}.

\begin{proof}[Proof of \cref{lem:tree_pump}]
  Define a sequence $(\alpha_i)$ of straight assembly sequences as follows:
  \begin{itemize}
  \item Fix $\alpha_0$ to be the assembly sequence with zero attachments: $\alpha_0 = (\seed[\Ss])$;
  \item for $i$ even, $\alpha_{i+1}$ is a straight extension of $\alpha_i$ which reaches a terminal production;
  \item for $i$ odd, $\alpha_{i+1}$ is obtained from $\alpha_i$ by \cref{lem:crash_recovery}.
  \end{itemize}

  For every $i$, either $\alpha_{i+1} = \alpha_i$ or $\alpha_{i+1} \moreFizzy \alpha_i$: at the even steps, if $\alpha_{i+1} \neq \alpha_i$, then it is a proper extension and is thus more fizzy; at the odd steps, if $\alpha_{i}$ does not embed into $\mathbb{Z}^2$, then $\ezEmbed(\alpha_i)$ is more fizzy, and so is its extension $\alpha_{i+1}$.

  If $\TerminalProds{\Ss} \cap C_m[\Ss] = \emptyset$, then for each $i$ odd, the Connected Treewidth of $\ezEmbed(\alpha_i)$ is at most $2m$, hence the Connected Treewidth of $\alpha_{i+1}$ is at most $2m$. But then, by \cref{lem:finite_choice}, $\alpha_{i}$ can only take a finite number of different values for $i$ even.

  Thus, the sequence $(\alpha_i)$ is eventually stationnary, and its fixpoint $\beta$ is a straight sequence which embeds into $\mathbb{Z}^2$ and reaches a terminal production. If $\TerminalProds{\Ss} \cap B_{F(n,m), \vec{d}} \neq \emptyset$, this terminal production must reach at least $F(n, m)$ in direction $\vec{d}$. Thus $\support[ {\lim \beta} ]$ has a branch with two windows $A, B$ such that $\movie(A) = \movie(B)$ and the vector sending $A$ to $B$ verifies $\vec{p} \cdot \vec{d} > 1$ (by \cref{lem:useless}). Because $\beta$ is straight, that branch is periodic, thus $\lim \beta \in \TerminalProds{\Ss} \cap P_{\vec{d}}$.
\end{proof}

\bibliography{auto-assemblage}
\end{document}